\documentclass[a4paper,UKenglish,cleveref, autoref, thm-restate]{article}

\usepackage{amsthm}
\usepackage{amsmath}
\usepackage{tikz-cd}

\usepackage[margin=1in]{geometry}  

\usepackage{times}  
\usepackage{mathptmx}  
\usepackage{amssymb}  

\usepackage{cmll}
\usepackage{xcolor}

\usepackage{thm-restate}

\usepackage{hyperref}
\hypersetup{
  colorlinks=true,     
  linkcolor=blue!35!gray,      
  citecolor= magenta!42!gray,       
  filecolor=magenta,   
  urlcolor=cyan!45!black,       
  pdftitle={Your PDF Title}, 
  pdfauthor={Your Name},      
  pdfsubject={Subject},      
  pdfkeywords={keyword1, keyword2}  
}

\usepackage{tabularx}

\usepackage{virginialake}

\newcommand*{\findpos}[2][p]{\mathsf{Find}(#1 ; #2)}
\newcommand*{\rootof}[1]{\mathsf{root}(#1)}
\newcommand*{\adrof}[1]{\mathsf{adr}(#1)}
\newcommand*{\maxadr}[1]{\mathsf{max}(#1)}
\newcommand*{\connects}[2][S]{\mathsf{link}_{#1}(#2)}

\newcommand*{\indexedarrow}[1]{\rightarrow_{#1}}
\newcommand*{\reduce}[2][]{\smash{\xrightarrow[#1]{#2}}}
\def\redmult{\indexedarrow{\mathsf{mult}}}
\def\redglue{\indexedarrow{\mathsf{glue}}}
\def\redtensdai{\indexedarrow{\mathsf{(\otimes/\maltese)}}}
\def\redparrdai{\indexedarrow{\mathsf{(\parr/\maltese)}}}

\def\rednotmult{\indexedarrow{\neg\mathsf{mult}}}
\def\rednotglue{\indexedarrow{\neg\mathsf{glue}}}
\def\rednottensdai{\indexedarrow{\mathsf{\neg(\otimes/\maltese)}}}
\def\rednotparrdai{\indexedarrow{\mathsf{\neg(\parr/\maltese)}}}

\def\tuple#1{\langle #1 \rangle}

\newcommand{\ruleone}[3]{
\vlderivation{
\vlin{}{{\scriptstyle #1}}{#2}
  {\vlhy{#3}}
}
}
\newcommand{\ruletwo}[4]{
  \vlderivation{\vliin{}{{\scriptstyle #1}}{#2}
  {\vlhy{#3}}
  {\vlhy{#4}}}
}

\newcommand{\finish}{\vlhy{}}
\newcommand{\aboveone}[2]{
  \vlin{}{\scriptstyle #2}{#1}
}
\newcommand{\abovetwo}[2]{
  \vliin{}{\scriptstyle #2}{#1}
}
\definecolor{dkgreen}{rgb}{0,0.6,0}
\definecolor{gray}{rgb}{0.5,0.5,0.5}
\definecolor{mauve}{rgb}{0.58,0,0.82}

\def\mlldai{\mathsf{MLL}^\maltese}
\def\mll{\mathsf{MLL}}

\def\at{\cdot}

\makeatletter
\@ifpackageloaded{stix}{%
}{%
  \DeclareFontEncoding{LS2}{}{\noaccents@}
  \DeclareFontSubstitution{LS2}{stix}{m}{n}
  \DeclareSymbolFont{stix@largesymbols}{LS2}{stixex}{m}{n}
  \SetSymbolFont{stix@largesymbols}{bold}{LS2}{stixex}{b}{n}
  \DeclareMathDelimiter{\lBrace}{\mathopen} {stix@largesymbols}{"E8}%
                                            {stix@largesymbols}{"0E}
  \DeclareMathDelimiter{\rBrace}{\mathclose}{stix@largesymbols}{"E9}%
                                            {stix@largesymbols}{"0F}
}
\makeatother

\def\proofs#1{\lBrace #1 \rBrace}

\def\parsing{\rightarrow_{\mathsf P}}

\def\height{\mathsf h}
\def\subord#1{\mathcal P _\leq ( #1 )}
\def\target{\mathrm{t}}
\def\source{\mathrm{s}}
\def\labl{\ell}

\def\hgraph{\mathcal H}

\def\reps{\equiv_\mathcal{R}}

\def\switchto{\rightarrow_{\parr}}
\def\compo{\Yright} 

\def\bibot{^{\bot\bot}}
\def\ibase{\mathcal B}

\def\hseq{\mathcal H}

\def\pright{\mathscr r}
\def\pleft{\mathscr l}

\def\deseq#1{\llbracket #1 \rrbracket}

\newcommand{\link}[3][]{\mathopen{\langle} #2 \vartriangleright_{#1} #3 \mathclose{\rangle}}
\newcommand{\parrlink}[2]{\link[\parr]{#1}{#2}}
\newcommand{\cutlink}[1]{\link[\mathsf{cut}]{#1}{}}
\newcommand{\dailink}[1]{\link[\maltese]{}{#1}}
\newcommand{\tenslink}[2]{\link[\otimes]{#1}{#2}}

\newcommand{\linke}[1]{\link[\labl(#1)]{\source(#1)}{\target(#1)}}

\newcommand{\reali}[2][]{\llbracket #2 \rrbracket_{#1}}

\newcommand{\tests}[1]{\mathsf{tests}({#1})}

\def\conclu#1{\mathsf{out}(#1)}
\def\prem#1{\mathsf{in}(#1)}

\def\restrict#1{\upharpoonright_{#1}}

\newcommand*{\collapse}{{\mkern-3mu \downarrow}}

\newcommand{\corestrict}[1]{\upharpoonleft^{#1}}

\newcommand*{\arrange}[1]{\mathbf a ( #1 ) } 
\newcommand*{\sumtr}{ \mathop{\sqcup \mkern -10.5mu {\raisebox{2.pt}{$\scriptscriptstyle +$}}} \mkern2.5mu} 

\makeatletter
\NewDocumentCommand{\transiarrow}{%
s
O{}
m
m
O{0.6cm}
}{%
\begin{tikzpicture}[baseline=-0.5ex] {
\node[inner sep=0](@1) at (0,0) {#3};
\node[inner sep=0](@2) at (#5,0) {#4};
\IfBooleanTF #1
  {\draw [arrows={-Triangle[open]},shorten >= 2pt,shorten <= 2pt](@1)--(@2) node[pos=.5,above,inner sep=1pt] {#2};}
  {\draw [arrows={-Triangle},shorten >= 2pt,shorten <= 2pt](@1)--(@2) node[pos=.5,above,inner sep=1pt] {#2};}
}
\end{tikzpicture}\xspace
}
\makeatother

\usepackage{mathtools}

\theoremstyle{plain}  
\newtheorem{theorem}{Theorem}

\newtheorem{proposition}[theorem]{Proposition}
\newtheorem{lemma}[theorem]{Lemma}
\newtheorem{corollary}[theorem]{Corollary}

\theoremstyle{definition}
\newtheorem{definition}[theorem]{Definition}
\newtheorem{remark}[theorem]{Remark}
\newtheorem{example}[theorem]{Example}

\theoremstyle{remark}
\newtheorem{notation}[theorem]{Notation}

\usepackage{cite}

\usepackage{anyfontsize}

\usepackage{tikzit}

\tikzstyle{cell}=[fill=black, draw=black, shape=circle]
\tikzstyle{squary}=[anchor=north, rounded corners=2pt, inner sep=1.2pt, fill=black, tikzit fill={rgb,255: red,11; green,96; blue,255}, shape=rectangle]
\tikzstyle{linkDOWN}=[anchor=north, trapezium, trapezium angle=110, rounded corners=2pt, inner sep=1.2pt, fill=black, tikzit fill={rgb,255: red,11; green,96; blue,255}]
\tikzstyle{linkUP}=[anchor=north, trapezium, rotate=180, trapezium angle=110, rounded corners=2pt, inner sep=1.2pt, fill=black, tikzit fill={rgb,255: red,255; green,73; blue,76}]
\tikzstyle{agent}=[inner sep=3.5pt, tikzit fill={rgb,255: red,255; green,224; blue,46}]
\tikzstyle{agentRecto}=[draw, rectangle,inner sep=3.5pt, tikzit fill={rgb,255: red,100; green,224; blue,46}]
\tikzstyle{agentG1}=[inner sep=3.5pt, tikzit fill={rgb,255: red,255; green,200; blue,80},font ={\color{red}}]
\tikzstyle{agentG2}=[inner sep=3.5pt, tikzit fill={rgb,255: red,255; green,80; blue,200},font ={\color{blue}}]
\tikzstyle{agentG3}=[inner sep=3.5pt, tikzit fill={rgb,255: red,255; green,80; blue,200},font ={\color{green}}]
\tikzstyle{tee}=[line width=0.3pt]
\tikzstyle{tensorLink}=[anchor=north, trapezium, trapezium angle=110, rounded corners=2pt, inner sep=1.2pt, fill=black, font = {$\scriptstyle\color{white}\otimes$}, tikzit fill={rgb,255: red,12; green,255; blue,0}]
\tikzstyle{parrLink}=[anchor=north, trapezium, trapezium angle=110, rounded corners=2pt, inner sep=1.2pt, fill=black, font = {$\scriptstyle\color{white}\parr$}, tikzit fill={rgb,255: red,0; green,255; blue,213}]
\tikzstyle{cutLink}=[anchor=north, trapezium, trapezium angle=110, rounded corners=2pt, inner sep=1.2pt, fill=black, font ={$\scriptstyle\color{white}\mathsf{cut}$}, tikzit fill={rgb,255: red,11; green,96; blue,255}, tikzit shape=rectangle]
\tikzstyle{daimonLink}=[anchor=north, trapezium, rotate=180, trapezium angle=110, rounded corners=2pt, inner sep=1.2pt, fill=black, font ={{$\scriptstyle\color{white}\maltese$}}, tikzit fill={rgb,255: red,255; green,73; blue,76}, tikzit shape=rectangle]
\tikzstyle{axLink}=[anchor=north, trapezium, rotate=180, trapezium angle=110, rounded corners=2pt, inner sep=1.2pt, fill=black, font ={\rotatebox{180}{$\scriptstyle\color{white}\mathsf{ax}$}}, tikzit fill={rgb,255: red,255; green,73; blue,76}, tikzit shape=rectangle]
\tikzstyle{tensorwLink}=[draw = black , anchor=north, trapezium, trapezium angle=110, rounded corners=2pt, inner sep=1.2pt, fill=white, font = {$\scriptstyle\otimes$}, tikzit fill={rgb,255: red,12; green,255; blue,0}]
\tikzstyle{parrwLink}=[draw=black , anchor=north, trapezium, trapezium angle=110, rounded corners=2pt, inner sep=1.2pt, fill=white, font = {$\scriptstyle\parr$}, tikzit fill={rgb,255: red,0; green,255; blue,213}]
\tikzstyle{cutwLink}=[draw=black , anchor=north, trapezium, trapezium angle=110, rounded corners=2pt, inner sep=1.2pt, fill=white, font ={$\scriptstyle\mathsf{cut}$}, tikzit fill={rgb,255: red,11; green,96; blue,255}, tikzit shape=rectangle]
\tikzstyle{daimonwLink}=[draw=black , anchor=north, trapezium, rotate=180, trapezium angle=110, rounded corners=2pt, inner sep=1.2pt, fill=white, font ={\scalebox{0.8}{$\scriptstyle\maltese$}}, tikzit fill={rgb,255: red,255; green,73; blue,76}, tikzit shape=rectangle]
\tikzstyle{axwLink}=[draw=black , anchor=north, trapezium, rotate=180, trapezium angle=110, rounded corners=2pt, inner sep=1.2pt, fill=white, font ={\rotatebox{180}{$\scriptstyle\mathsf{ax}$}}, tikzit fill={rgb,255: red,255; green,73; blue,76}, tikzit shape=rectangle]

\tikzstyle{linkInput}=[->, line width=0.2pt, draw opacity=0, latex reversed-, postaction={draw, opacity=1, -, line width=1pt, shorten <=1pt}, tikzit draw={rgb,255: red,245; green,75; blue,200}]
\tikzstyle{linkOutput}=[->, draw opacity=0, -latex, line width=0.2pt, postaction={draw, opacity=1, -, line width=1pt, shorten >=3pt}]
\tikzstyle{filler}=[-, fill={rgb,255: red,139; green,150; blue,221}, draw=none, fill opacity=0.5, tikzit draw=black]
\tikzstyle{fillerRed}=[-, fill={rgb,255: red,230; green,76; blue,86}, draw=none, fill opacity=0.5, tikzit draw=black]
\tikzstyle{dottedEdge}=[-, dotted, draw=black]
\tikzstyle{dottedArrow}=[draw=black, ->, dotted]
\tikzstyle{simpleUndirected}=[-, draw=black]
\tikzstyle{simple}=[draw=black, ->]
\tikzstyle{simplered}=[draw=red, ->]
\tikzstyle{simpleOrange}=[->, color=orange, tikzit draw={rgb,255: red,255; green,93; blue,0}]
\tikzstyle{implicarrow}=[->, draw=black, double]
\tikzstyle{thickedge}=[line width=2pt, draw=black, -]
\tikzstyle{inputBend}=[bend right = 45 , ->, line width=0.2pt, draw opacity=0, latex reversed-, postaction={draw, opacity=1, -, line width=1pt, shorten <=1pt}, tikzit draw={rgb,255: red,245; green,75; blue,200}]
\tikzstyle{outLEFT}=[out = 180 , in = 90 , looseness = 0.85, ->, draw opacity=0, -latex, line width=0.2pt, postaction={draw, opacity=1, -, line width=1pt, shorten >=3pt},tikzit draw={rgb,255: red,100; green,75; blue,200}]
\tikzstyle{outRIGHT}=[out = 00 , in = 90 ,looseness = 0.85, ->, draw opacity=0, -latex, line width=0.2pt, postaction={draw, opacity=1, -, line width=1pt, shorten >=3pt}, tikzit draw={rgb,255: red,245; green,200; blue,100}]
\tikzstyle{inLEFT}=[out = -90 , in = 180 ,looseness = 0.85, ->, draw opacity=0, latex reversed-, line width=0.2pt, postaction={draw, opacity=1, -, line width=1pt, shorten <=1pt},tikzit draw={rgb,255: red,100; green,75; blue,255}]
\tikzstyle{inRIGHT}=[out = -90 , in = 00 ,looseness = 0.85, ->, draw opacity=0, latex reversed-, line width=0.2pt, postaction={draw, opacity=1, -, line width=1pt, shorten <=1pt}, tikzit draw={rgb,255: red,245; green,255; blue,100}]

\input{style.tikzdefs}

\usepackage{caption}
\usepackage{subcaption}

\usepackage[scr=boondox]   
           {mathalpha}

\usepackage{float}

\usepackage{enumitem}
\setlist{topsep=0pt, noitemsep}

\author{Adrien Ragot\footnote{\infoA} \, \footnote{The author is upported by a VINCI PhD fellowship from the Franco-Italian Université.} \, , 
Thomas Seiller\footnote{\infoB}\, \footnote{The author is partially supported by the ANR-22-CE48-0003-01 project DySCo.} \, , Lorenzo Tortora de Falco\footnote{\infoC}}

\def\infoA{Universit\'e Sorbonne Paris Nord (LIPN, UMR 7030); Universit\`a Degli Studi Roma Tre, Dipartimento di Matematica e Fisica.}
\def\infoB{Universit\'e Sorbonne Paris Nord ; CNRS (LIPN, UMR 7030).}
\def\infoC{Universit\`a Degli Studi Roma Tre, Dipartimento di Matematica e Fisica; GNSAGA, Istituto
Nazionale di Alta Matematica.}

\begin{document}

\title{Linear Realisability over nets: multiplicatives}


  \maketitle

  \begin{abstract} 
  
    We provide a new realisability model
    based on orthogonality
    for the multiplicative fragment of linear logic,
    both
    in presence of generalised axioms $(\mlldai)$ and in the standard case $(\mll)$.
    The novelty is the definition of cut elimination
    for generalised axioms.
    We prove that our model is adequate and complete both for $\mlldai$ and $\mll$.
  \end{abstract}



\def\ha{\textsf{HA}\;}
\def\bhk{\textsf{BHK}\;}
\def\paragrafo#1{\vspace*{5pt}\noindent\textbf{#1}\textbf{.}}

\section*{Introduction}

Since the inception of Linear Logic ($\mathsf{LL}$), 
proofs are represented
as graphs that naturally live in a wider space of agents 
called proof structures (\emph{nets} in this paper) that can freely interact.
These nets 
were introduced by J.Y. Girard in \cite{girard_1987}, 
together with the \emph{desequentialisation}:
a simple process transforming proof trees from the sequent calculus 
of $\mathsf{LL}$ into nets. 
However, not
every net is the desequentialisation of a proof:
it is \emph{impossible} to extract a proof tree 
from 
a net that ``contains'' cycles or disconnections \cite{danos89}.
Nets can therefore present forms of (what we call) \emph{geometrical incorrectness},
and geometrically correct nets are (representants of) proof trees of $\mathsf{LL}$.
More recently, 
J.Y. Girard proposed \emph{Ludics}, 
an interpretation of $\mathsf{LL}$ 
given in terms of ``desseins'': 
proof trees of the $\mathsf{LL}$ sequent calculus with the addition of the daimon $(\maltese)$ rule, 
a generalised axiom allowing to prove any sequent. 
Ludics introduces a new kind of incorrectness that we call \emph{provability incorrectness}:
dessein are geometrically correct (they are proof trees)
but can be provably incorrect.
In the standard theory of proof nets 
geometrical and provability correctness coincide;
it is the presence of daimons that 
allows to distinguish
between provability correctness and geometrical correctness.


Understanding the relationship between correctness 
and computational behavior
is (one of) the goal(s) of \emph{realisability}, 
which, 
restricted to \textsf{LL}, will be our focus in this paper.
We briefly sum up 
the existing works on 
linear realisability\footnote{We use the expression linear realisability
in the sense of \cite{seiller:hdr}
i.e. realisability 
models for \textsf{LL}.}
by positioning them with respect to \autoref{table:intro}.
We also recall if these models enjoy completeness or not. 
Two lines of research on realisability of $\mathsf {LL}$
can be identified.

One
was initiated by V. De Paiva \emph{Dialectica Interpretation} \cite{depaiva:dialecticaLL}
and led to
P. Oliva's adequate and complete realisability model
of first order $\mathsf {LL}$ \cite{oliva:modrealiLL}
where 
realisers are 
proof trees (with standard axioms) 
from a decorated 
sequent calculus of $\mathsf {LL}$.
As a consequence 
realisers are typed and are 
``by construction'' \emph{geometrically and provably correct}
(placing this model in the top left 
corner of \autoref{table:intro}).  

The other
originates in the work of J.Y. Girard:
\emph{Ludics} \cite{girard:ludics},
whose ``desseins'' are geometrically correct 
but can be provably incorrect
(top right corner of \autoref{table:intro}),
which enjoys \emph{completeness}.
E. Beffara proposed adequate models
in a concurrent $\pi$-calculus \cite{beffara:concurrentLinearRealisability} 
and conjunctive structure \cite{beffara:conjonctiveRealisability}.
T. Seiller's \emph{interaction graphs} 
(inspired by Girard's Geometry of Interaction \cite{girard:multiplicatives})
model various \textsf{LL} fragments adequately 
\cite{seiller:igm,seiller:iga,seiller:igg,seiller:ige,seiller:igf}.
Beffara's and Seiller's
approaches exhibit both geometrical and provability incorrectness 
(bottom-right corner of \autoref{table:intro}),
but contain no completeness result.

We give the first complete realisability model 
of the multiplicative fragment of
linear logic in terms of nets,
essentially the well--known untyped proof--structures 
of $\mathsf {LL}$ \cite{girard_1996} 
with \emph{daimons},
as in the work of P.L. Curien \cite{curien:criterions}:
this places us in the bottom--right corner of \autoref{table:intro}.
The main tool we use in our approach 
to realisability is $\mathsf {LL}$ cut elimination: 
we interpret formulas as types, sets of nets closed under bi--orthogonality,
where the notion of orthogonality 
is defined via the rewriting rules of nets induced by cut elimination.
We prove completeness for $\mlldai$, multiplicative $\mathsf {LL}$
with generalised axioms, 
meaning our model can capture 
\emph{geometrical correctness}.
As a byproduct we obtain completeness 
for the standard multiplicative fragment of linear logic ($\mll$), thus capturing 
\emph{provability correctness}.

Although not expressed in the terms of realisability, 
a completeness result for $\mlldai$
(in the atomic case)
using a notion of orthogonality
is already apparent in the work of P.L. Curien \cite{curien:criterions},
where the partitions involved in the Danos Regnier criterion \cite{danos89} 
are encoded using daimons.
More precisely, 
one can test the geometrical correctness of a net 
by confronting it against carefully chosen \emph{opponents} (which as in the work of B\'echet \cite{bechet:correctness}
are geometrically correct nets).
However the method in \cite{curien:criterions} 
does not allow to derive a completeness result for $\mll$.
By contrast,
we use \emph{geometrically incorrect} opponents to prove completeness for $\mll$ (\autoref{rem:provabilityCorrectNeedOrthoIncor}).

The novelty is the \emph{cut elimination} of non-homogenous cuts 
(a generalised axiom against a connective -- say a tensor):
unlike in Ludics
\footnote{In Ludics, the daimon means the end of the game, 
or the end of the proof search.}\cite{girard:ludics}
our daimon is the ``perfect'' opponent/evaluation context;
it never stops responding during computation
and 
never prevents proof search to go on 
(\autoref{fig:nhomcutelim} and \autoref{rem:nondet-proofsearch}).
These new cut elimination steps
are key to 
interactively identify provability correctness 
and so to
obtain our completeness result for $\mll$ (\autoref{rem:nhomcutelimWHY}).
The computational behavior of the daimon 
also differs from Krivine's continuations 
involved in \emph{classical realisability} \cite{krivine:realiclassi}: 
they restore a previously stored context while
our daimon rather behaves like an adaptive evaluation context.\\
\begin{table}    
    \centering
    \vspace*{-12pt}
\scalebox{0.8}{\begin{tabular}{|c|c|c|}
    \hline
    & $\mll$ & $\mlldai$ \\
    \hline
    Proof Nets & no incorrectness & provability incorrectness   \\
    \hline
    Nets & geometrical incorrectness  & geometrical and provability incorrectness \\
    \hline
\end{tabular}}
    \caption{Presence of incorrectness,
    restricted to multiplicative linear logic,
    for 
    realisability models.}\label{table:intro}
\end{table}

The general aim is to understand 
the computational content of proofs and of (incorrect) nets,
following a ``purely interactive approach to logic'' (to quote \cite{girard:ludics}).
We follow the approach initiated with Ludics,
we present a framework in which proofs and refutations 
are objects of the same nature that can freely interact:
a proof--object proves a formula $A$ whenever 
it ``defeats'' all the refutations of $A$.
The correctness of an object 
is evaluated using a dynamic criterion (we make an object interact with each of its refutations)
rather than a static one (such as a typing discipline).

\paragrafo{Outline}
In \autoref{sect:hypergraphs},
we give a detailed introduction of nets 
that we define as ordered hypergraphs.
In \autoref{sect:MLL},
we recall the elementary notions 
of multiplicative linear logic, we introduce the $\maltese$-links
and we formulate the criterion of Danos Regnier \cite{danos89}
in our setting.
In \autoref{sect:RealModel}, we define orthogonality between two nets as 
``successful interaction'' through cut  elimination (Definition~\ref{def:orthogonalityNets}); this leads to the notion of 
type: a set of nets closed under bi--orthogonality. We then show how to perform the usual multiplicative constructions
in the framework of types.
In \autoref{sect:Adequacy}, we define our realisability model 
interpreting formulas as types and we prove its adequacy: 
a net representing a proof of $A$ 
is a realiser of $A$ (Theorem~\ref{thm:adequacy}).
In \autoref{sect:Tests},
we relate correctness criteria with orthogonality.
The Danos-Regnier criterion applied to a cut-free net with conclusion $A$
yields a set of nets called tests (Definition~\ref{def:test}). 
We prove that the tests of $A$ are proofs of $A^{\bot}$ (Theorem~\ref{thm:cortest})
and that the interaction between a net $\pi$ with conclusion $A$ 
and its tests allows to determine whether or not $\pi$ is indeed a proof:
we thus extend to our framework a result of B\'echet \cite{bechet:correctness}.
In \autoref{sect:Completeness}, we prove the completeness of our realisability model: if
a net $S$ realises $A$ (in every basis), then $S$ is a proof of $A$ in $\mlldai$ (Theorem~\ref{thm:mlldaiNewComplete}). 
Finally
we show that completeness of $\mlldai$
implies that of $\mll$ (Theorem~\ref{thm:completeness}).

\section{Untyped nets}\label{sect:hypergraphs}

\def\corestrict#1{\upharpoonright ^{#1}}

\newcommand{\before}[2]{ #2_{< #1}}
\newcommand{\after}[2]{ #2_{> #1}}

\def\size#1{\left| {#1} \right|}
\def\posiset{\mathsf{Pos}}
\def\body#1{\left| {#1} \right|}

\def\pvarset{\mathsf{Var}}
\def\formulaset{\mathsf{Form}}
\def\hseqset{\mathsf{Hseq}}
\def\parsing{\rightarrow}
\def\graph{\mathsf G}
\newcommand*{\undergraph}[1]{\mathsf G(#1)}
\newcommand*{\border}[1][]{\mathsf{brd}_{#1}}
\def\daipart#1{\mathsf P_\maltese(#1)}
\newcommand*{\nat}[2][]{\mathbf{Nat}_{#1}(#2)}

We introduce the 
framework of \emph{nets} in which 
our construction takes place.
Nets are a special kind of \emph{directed hypergraphs}
together with an order of \emph{some} of their vertices
which will come in play later on to define the notion of orthogonality.
These hypergraphs enjoy a natural notion of sum (\autoref{def:sumhgraph}).
In subsection \ref{subsec:multinets},
we define our ``realisers'' that we call nets 
and
their computational rules,
the cut elimination procedure 
as known for multiplicative proof structures \cite{girard_1996}
but with a novelty: the generalised axiom or daimon--link $(\maltese)$
which behave like an adaptative evaluation context.

\subsection{Directed hypergraphs}

Given a set $X$ we will let $\subord X$
denote the set of
totally ordered finite subsets of $X$.
An element of $\subord X$ is equivalently 
a finite sequence of elements of $X$ but,
\emph{without repetitions}.


\begin{definition}\label{def:hypergraph}
  Suppose given a set $L$ of \emph{labels}.
  A \emph{directed ($L$-labelled) hypergraph} is
  a tuple $(V,E,\source,\target,\labl)$
  where $V$ is a finite set of \emph{positions}
  and $E$ is a finite set
  of \emph{links},
  $\source:E\rightarrow \subord V$ is the
  \emph{source map}, 
  $\target : E\rightarrow \subord V$
  is the \emph{target map}
  and
  $\labl: E\rightarrow L$ is the \emph{labelling} map.
\end{definition}

  
  Given a link $e\in E$, 
  since the finite sets $\target(e)$ and $\source(e)$ 
  are totally ordered, 
  to support readability
  we will 
  represent them as sequences: they are respectively
  called the \emph{target} and the \emph{source}
  sets of $e$. 
  A \emph{source} (resp. target) of $e$
  is an element
  of its source (resp. target) set $\source(e)$ (resp.$\target (e)$).
  The set of 
  targets and sources of $e$  
  is the \emph{domain} of the link $e$.
  We will use superscripts 
  to denote sequences of positions 
  (${\overline p},{\overline q},{\overline u},\ldots$).
  A link is a \emph{loop}
  when its target set and source set 
  are not disjoint.

\noindent
  \textbf{Convention.}
  Along this work
  we assume all the hypergraphs to be loop--free
  i.e. containing only links which are not loops.

  Given an hypergraph $\hgraph$ 
  with $E$ as its set of links, 
  we denote 
  $\source(\hgraph)$ (resp. $\target(\hgraph)$)
  the set of all positions which
  are 
  source (resp. target) of at least one link:
  \[\source(\hgraph)=\bigcup_{e\in E} \source(e), \qquad
  \target(\hgraph)=\bigcup_{e\in E} \target(e).\]
  A \emph{conclusion/output} (resp. a \emph{premise/input})
  of a directed hypergraph $\hgraph$
  is a position which is the source (resp. target)
  of no link in $\hgraph$, i.e. 
  an element of $V\setminus \source(\hgraph)$
  (resp. of $V\setminus \target(\hgraph)$).
  The set of conclusions (resp. premises)
  of an hypergraph $\hgraph$
  is denoted $\conclu \hgraph$
  (resp. $\prem\hgraph$).
  A position $p$ is \emph{isolated}
  in an hypergraph $\hgraph$
  if $p$ is both an output and
  an input of $\hgraph$,
  i.e. 
  $p\notin \source(\hgraph)\cup\target(\hgraph)$.
  The \emph{size} of a directed hypergraph
  is the number of its links.
  There is a unique empty hypergraph 
  $\hgraph = (V,E,\source,\target,\labl)$
  with 
  $V=E=\source=\target=\emptyset$.

  An \emph{isomorphism}
  of hypergraphs 
  $f:(V_1,E_1,\source_1,\target_1,\labl_1)\rightarrow (V_2,E_2,\source_2,\target_2,\labl_2)$ 
  is a pair of functions $(f_V,f_E)$
  such that $f_V:V_1\rightarrow V_2$ 
  and $f_E : E_1\rightarrow E_2$ are bijections,
  $f_E$ preserve labels i.e. 
  $\labl(f_E(e))=\labl(e)$,
  and $f_E$ preserves the target and source of a link,
  i.e.
  $\source_2 (f_E (e)) = f_V^* (\source_1(e))$
  and 
  $\target_2 (f_E (e)) = f_V^* (\target_1(e))$,
  where $f_V^*$ is the natural extension of $f_V$
  to sequences of positions.
  Along this work we work 
  with hypergraphs up to isomorphism.

\begin{notation}\label{nota:link}
  We denote $\link[l]{\overline u}{\overline v}$
  the hypergraph $(V,E,\source,\target,\labl)$
  such that $E=\{e \}$,
  $V=\source (e)\cup\target(e)$,
  $\source(e)=\overline u$,
  $\target(e)=\overline v$
  and $\labl(e)=l$ 
  (an example of such a single--link hypergraph 
  is found in \autoref{fig:linkHgraph}).
  In the sequel 
  $\link[l]{\overline u}{\overline v}$
  will denote both
  the described hypergraph 
  and its unique link.
\end{notation}

\begin{notation}
  We write $u\cdot v$ the concatenation of sequences.
  Given $u=(u_1,\dots,u_{n})$ a sequence
  of elements of a set $X$ and an integer $i\in\{1,\dots,n\}$,
  we denote by $u_{<i}$
  (resp. $u_{>i}$)
  the sequence $(u_1,\dots,u_{i-1})$
  (resp. $(u_{i+1},\dots,u_n)$).
  Moreover, given two -- potentially empty -- sequences $u$ and $v$ 
  we denote by 
  $u[i \leftarrow v]$
  the sequence $\before i u \cdot v \cdot \after i u$.

\end{notation}

A link is \emph{initial} (resp. \emph{final})
when it has no input (resp. no output).
A position is \emph{initial} (resp. \emph{final}) 
when it is an output (resp. input)
of an initial (resp. final) link.
In an hypergraph $\hgraph$, 
a link $e$ is \emph{terminal}
when every target of $e$ is a conclusion of $\hgraph$
-- thus a final link is a terminal link.

\begin{example}
  For instance a link $\link[\ell]{}{a,b,c}$ is an initial link
  and the positions $a,b$ and $c$ are initial,
  on the other hand a link
  $\link[\ell]{a,b}{c}$
  is not initial and neither are the positions $a,b$ or $c$.
\end{example}


  Hypergraphs 
  enjoy a natural notion of sum
  based on the disjoint union of the set of links.

  \begin{notation}
    Given two sets $X_0$ and $X_1$
    we denote $X_0\uplus X_1$ the set $X_0\cup X_1$
    whenever $X_0$ and $X_1$ are disjoint.
    Given two functions $f:X_0 \rightarrow E$
    and $g:X_1 \rightarrow E$ with disjoint domains
    we denote 
    $f\uplus g$
    the function which takes an element $x$ 
    of 
    $X_0\uplus X_1$,
    and returns 
    $f(x)$ if $x\in X_0$
    and 
    $g(x)$ if $x\in X_1$.
  \end{notation}

  \begin{definition}\label{def:sumhgraph}
    Given two hypergraphs
    $\hgraph_1 = (V_1,E_1,\target_1,\source_1,\labl_1)$
    and
    $\hgraph_2 = (V_2,E_2,\target_2,\source_2,\labl_2)$
    such that $E_1\cap E_2=\emptyset$.
    The sum of $\hgraph_1$
    and $\hgraph_2$ is defined as:
    $$
    \hgraph_1 +\hgraph_2 = 
    (
    V_1 \cup V_2, E_1\uplus E_2,
    \target_1\uplus \target_2,
    \source_1\uplus \source_2,
    \labl_1\uplus\labl_2
    ).
    $$
  \end{definition}

  \begin{figure*}[t]
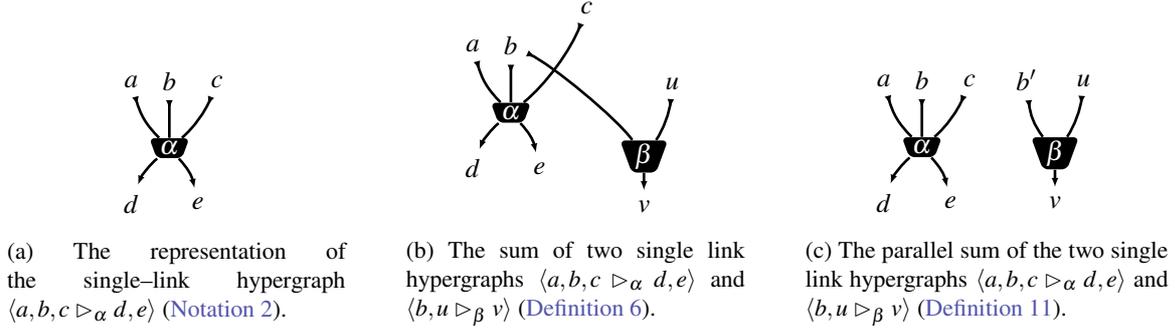

    \begin{subfigure}{0.28\textwidth}
      \centering
      \input{hypergraph_link.tikz}
      \caption{The representation 
      of the single--link hypergraph 
      $\link[\alpha]{a,b,c}{d,e}$  
      (\autoref{nota:link}).}\label{fig:linkHgraph}
    \end{subfigure}
    \qquad
    \begin{subfigure}{0.28\textwidth}
      \centering
      \input{hypergraph_sum.tikz}
      \caption{The sum of two single link 
      hypergraphs $\link[\alpha]{a,b,c}{d,e}$ 
      and $\link[\beta]{b,u}{v}$
      (\autoref{def:sumhgraph}).}\label{fig:sumhgraph}
    \end{subfigure}
    \qquad
    \begin{subfigure}{0.3\textwidth}
      \centering
      \input{hypergraph_sumparallel.tikz}
      \caption{The parallel sum of the two single link hypergraphs 
      $\link[\alpha]{a,b,c}{d,e}$ 
      and $\link[\beta]{b,u}{v}$ (\autoref{def:parallelhgraph}).
      }\label{fig:sumparallel}
    \end{subfigure}
    \caption{Hypergraphs can naturally be represented 
    in a graphical way, 
    we illustrate the notation of a hypergraph containing a single link, 
    the sum of hypergraphs and the parallel sum of hypergraphs.
    In \autoref{fig:sumparallel} 
    The position $b$ is present in both hypergraphs
      therefore we rename it in one of the two hypergraphs:
      thus $\link[\alpha]{a,b,c}{d,e} \parallel \link[\beta]{b,u}{v}$ 
      equals $\link[\alpha]{a,b,c}{d,e} \parallel \link[\beta]{b',u}{v}$ (that is, upto isomorphism).}\label{fig:firstDefs}
  \end{figure*}

  \begin{remark}\label{rem:sum}
    Whenever
    $\hgraph_1=(V_1 , E_1,\target_1,\source_1,\labl_1)$
    and $\hgraph_2=(V_2 , E_2,\target_2,\source_2,\labl_2)$
    are such that $E_1\cap E_2 \neq\emptyset$, 
    we will
    abusively write their sum as $\hgraph_1 +\hgraph_2 =
    (V_1 \cup V_2, E_1\uplus E_2,\target_1\uplus \target_2,
    \source_1\uplus \source_2,\labl_1\uplus\labl_2)$, 
    since up to renaming
    the sets of links of two hypergraphs 
    can always be considered disjoint.
  \end{remark}

\begin{remark}
  Vertices may overlap in a sum (as we take the union
  of vertex sets rather than the disjoint union). 
  As a consequence,
  a position may be input (or output) of several distinct
  links (\autoref{fig:sumhgraph}).
  We can describe hypergraphs 
  as sums of simple hypergraphs;
  namely those that contain only one link.
  Indeed using \autoref{nota:link}, an hypergraph consisting of two links 
  $\link[\ell]{\overline a}{\overline b}$
  and 
  $\link[\ell']{\overline c}{\overline d}$ 
  is in fact equal to the sum of the
  single-link hypergraphs 
  $\link[\ell]{\overline a}{\overline b}$ and
  $\link[\ell']{\overline c}{\overline d}$. 
  By induction on the number of links,
  this shows that any hypergraph $\hgraph$
  without isolated positions can be written as $\hgraph=\sum_{e\in E} \linke{e}$.
\end{remark}

\begin{example}  
  In the hypergraph 
  $\link[\ell1]{}{a,b,c} + \link[\ell2]{a}{d} +\link[\ell3]{}{e} + \link[\ell4]{e}{}$
  the set of initial positions is  $\{a,b,c,e\}$,
  while $e$ is the only final position of the hypergraph,
  and it belongs to the domain of the unique final link $\link[\ell4]{e}{}$.
\end{example}
  
\begin{remark}\label{rem:sumHgraphAbelian}
  The sum of hypergraphs 
  enjoys the properties of an abelian monoid;
  associativity, 
  commutativity,
  and a neutral element 
  which is the empty hypergraph.
\end{remark}

We will also 
use extensively 
the notion 
of \emph{parallel composition}
or \emph{parallel sum}
of hypergraphs, 
an analogue of the \emph{union--graph}
of two simple graphs.

\begin{definition}\label{def:parallelhgraph}
  Given 
  $\hseq_1 = (V_1, E_1 , \target_1,\source_1,\labl_1)$
  and 
  $\hseq_2 = (V_2,E_2,\target_2,\source_2,\labl_2)$ two hypergraphs
  such that $V_1\cap V_2 = E_1\cap E_2 = \emptyset$,
  we define their 
  \emph{parallel sum} as:
  $\hseq_1 \parallel \hseq_2 
  =
  (V_1 \uplus V_2, E_1\uplus E_2,
  \target_1\uplus \target_2 , \source_1\uplus \source_2,\labl_1\uplus\labl_2)$.
\end{definition}

\begin{remark}
  The parallel sum of two hypergraphs 
  $\hseq_1$ and $\hseq_2$
  corresponds to a regular sum 
  whenever the sets of vertices 
  are disjoint.
  Just like the sum,
  parallel composition 
  can always be performed between two hypergraphs 
  (up to a renaming, see \autoref{fig:sumparallel}).
\end{remark}



A hypergraph $\hgraph=(V,E,\target,\source,\labl)$ is:
(1) \emph{target--surjective}
    whenever $\target(\hgraph)=V$,
(2) \emph{source--disjoint}
    if the sets $\source(e)$ for $e\in E$
    are pairwise disjoint,
(3) \emph{target--disjoint}
    if the sets $\target(e)$ for $e\in E$
    are pairwise disjoint (\autoref{fig:moduleDef}). 
A \emph{module} is an hypergraph 
which is target--disjoint and source--disjoint,
which means that for each position $p$ 
there exists \emph{at most} one link $e$
such that $\source (e)$ (resp. $\target(e)$) contains $p$.
Any single--link hypergraph is a module.
Uncarefully summing two modules does not necessarily 
result in a module;
for instance 
the single link hypergraphs 
$e = \link[\ell]{}{a}$ and $e' = \link[\ell']{}{a}$
are both modules but their sum isn't 
as $a$ is the target of the two links $e$ and $e'$.

\begin{figure}
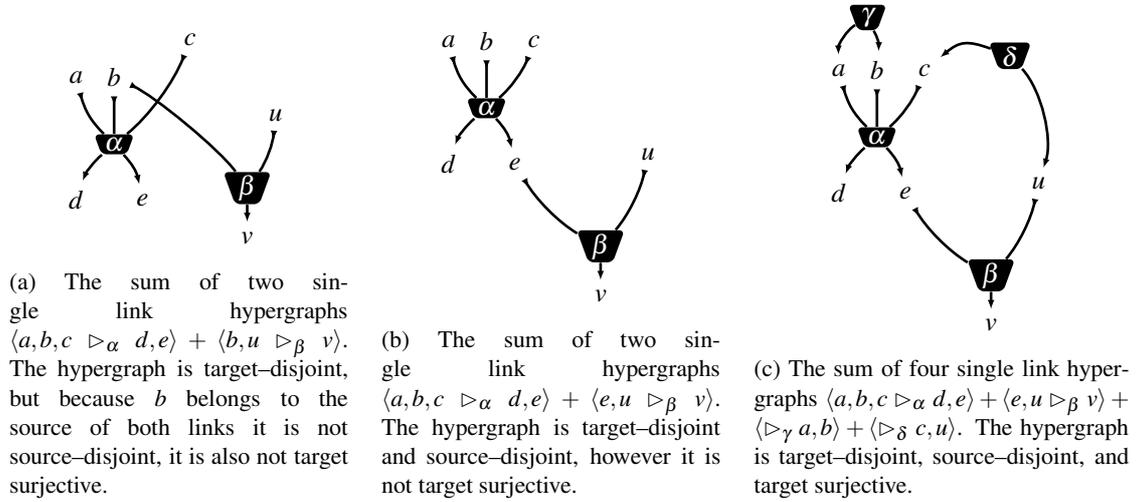

  \centering
  \begin{subfigure}{0.28\textwidth}
    \centering 
    \input{hypergraph_sum.tikz}
    \caption{The sum of two single link hypergraphs 
    $\link[\alpha]{a,b,c}{d,e} + \link[\beta]{b,u}{v}$.
    The hypergraph is target--disjoint, but
    because $b$ belongs to the source of both links 
    it is not source--disjoint,
    it is also not target surjective.
   }
  \end{subfigure}
  \quad 
  \begin{subfigure}{0.28\textwidth}
    \centering 
    \input{hypergraph_sumInj.tikz}
    \caption{The sum of two single link hypergraphs 
    $\link[\alpha]{a,b,c}{d,e} + \link[\beta]{e,u}{v}$.
    The hypergraph is target--disjoint and source--disjoint,
    however it is not target surjective.
   }
  \end{subfigure}
  \quad 
  \begin{subfigure}{0.31\textwidth}
    \centering 
    \input{hypergraph_sumInjTar.tikz}
    \caption{The sum of four single link hypergraphs 
    $\link[\alpha]{a,b,c}{d,e} + \link[\beta]{e,u}{v} + \link[\gamma]{}{a,b} + \link[\delta]{}{c,u}$.
    The hypergraph is target--disjoint, source--disjoint,
    and target surjective.
   }
  \end{subfigure}
  \caption{Properties of hypergraphs: source--disjoint, target--disjoint and target--surjective hypergraphs.}
  \label{fig:moduleDef}
\end{figure}

  An \emph{arrangement}
  of a directed hypergraph $\hgraph$
  is a total order $<_{\mathbf a}$ on its set of conclusions;
  equivalently
  the order may be identified
  as a bijection 
  $\mathbf a: 
  \{1,\dots,card (\conclu \hgraph) \} 
  \rightarrow \conclu \hgraph$.
  An \emph{ordered hypergraph} is a pair $(\hgraph,\mathbf a)$ 
  of an hypergraph $\hgraph$ together with 
  an arrangement $\mathbf a$ of $\hgraph$.
  Given an ordered hypergraph $(\hgraph,\mathbf a)$ 
  with $n$ conclusions
  for an integer $1\leq i \leq n$,
  we denote $\arrange i$
  by $\hgraph (i)$ whenever there is no ambiguity.
  The arrangement $\mathbf a$
  is denoted $\arrange \hgraph$,
  and we might refer 
  to $\hgraph$ 
  as the \emph{unordered hypergraph underlying}
  $(\hgraph , \mathbf a).$

  For $n,m\in\mathbb{N}$
  we denote by $\mathbf{[}n;m\mathbf {]}$
  the set 
  of integers 
  $i$ such that $n\leq i \leq m$.
  Given two functions
  $f:[1;n]\rightarrow E$
  and $g:[1;m]\rightarrow E$
  we denote 
  $f\sumtr g : [1;m+n]\rightarrow E$
  the function such that 
  $f\sumtr g (i) = f(i)$ when $1\leq i \leq n$
  and $f\sumtr g (i) = g(i-n)$ when $n+1\leq i \leq n+m$.
  This operation is not commutative.
  The parallel sum of two ordered 
  hypergraph $(\hgraph_1,\mathbf a_1)$ and $(\hgraph_2,\mathbf a_2)$
  naturally 
  yields an ordered hypergraph 
  as $(\hgraph_1\parallel\hgraph_2 , \mathbf a_1 \sumtr \mathbf a_2)$
  (note that however this is not a commutative operation).

\subsection{Multiplicative nets}\label{subsec:multinets}

Up to this point we have allowed
any kind of link to occur in a hypergraph. We now
consider
untyped multiplicative nets
in which only some specific kinds of links occur.
We fix the set of labels 
as the set made 
of the \emph{daimon} ($\maltese$)
the tensor ($\otimes$)
the par ($\parr$)
and the cut ($\mathsf{cut}$) symbol.
Furthermore we fix 
a family of links,
namely 
$\maltese$-labelled links that have no inputs (they are initial links),
$\mathsf{cut}$-labelled links that have exactly two inputs and no outputs (they are final links),
$\otimes$- and $\parr$-labelled links that have exactly two inputs and one output.
As a consequence, the hypergraphs considered will closely resemble
to multiplicative linear logic proof structures, with two important points of divergence:
the absence of typing and the presence of generalised axioms,
a standard $\mll$ axiom link can be seen as daimon link 
with two conclusions
\footnote{To be precise one should say that an \emph{atomic} standard $\mll$ axiom link 
\emph{is} a daimon link 
with two conclusions (\autoref{rem:mllCutElim}).}.

Formally we fix a countable set $\posiset$ of positions
and a family of links $\mathcal L$ defined as:

\vspace{-10pt}
\begin{center}
  \scalebox{0.7}{
    $\mathcal L \triangleq 
  \{ 
\tenslink{p_1,p_2}{p},~
\parrlink{p_1,p_2}{p},~
\cutlink{p_1,p_2} \mid p_1,p_2,p \in \posiset
\}
\quad\cup\quad 
\{ \dailink{p_1,\dots,p_n} \mid n\in \mathbb N , p_1,\dots,p_n \in \posiset \}.$
  }
\end{center}


\begin{definition}
  A multiplicative module is 
  an ordered hypergraph $M = (\body {M} , \arrange M)$ 
  where $\body M$ is a sum of links of $\mathcal L$ which is a module.

  A multiplicative net 
  is a multiplicative module $S = (\body {S} , \arrange S)$
  where $\body S$ is target--surjective.
\end{definition}

From now on we will omit the word \emph{multiplicative}
but a module (resp. net) will always be 
a multiplicative module (resp. net).
For a module $M$ (resp. a net $S$)
we refer to $\body M$ (resp. $\body S$) 
as the unordered hypergraph underlying $M$ (resp. $S$).
An \emph{unordered} module (resp. net)
is the unordered hypergraph 
underlying a module (resp. net).

\begin{remark}
  For two nets 
  $S_1 = (V_1,E_1,\source_1,\target_1,\labl_1)$
  and 
  $S_2 = (V_2,E_2,\source_2,\target_2,\labl_2)$,
  if $S_1+ S_2$ remains a net 
  then 
  $S_1 + S_2 = S_1\parallel S_2$.
  Indeed, 
  by \autoref{def:sumhgraph},
  $E_1\cap E_2 = \emptyset$.
  Then,
  by target--disjointness 
  $\target(S_1)\cap\target(S_2)=\emptyset$;
  and finally because $S_1$ and $S_2$ are target surjective 
  we have $V_1\cap V_2 = \target(S_1)\cap\target(S_2)=\emptyset$,
  so that \autoref{def:parallelhgraph}
  applies.
\end{remark}



\def\daimon#1{\maltese_{#1}}

\begin{notation}
  Given an integer $n$ 
  we denote by $\daimon n$
  any multiplicative net
  consisting of a single daimon link 
  with $n$ outputs,
  i.e. isomorphic
   to $\dailink{p_1,\dots,p_n}$.
\end{notation}

\begin{definition}\label{def:cuttypes}
  Given a multiplicative net
  $S$ 
  the \emph{type} of a cut link $c = \cutlink{p,q}$
  occurring in $S$
  is the multiset of the two labels
  of the links of output $p$ and $q$;
  for readability 
  we write these multisets as ordered pairs.
  Thus there are six
  \emph{types} of cuts (up to symmetry).
  More precisely, we distinguish:
    \emph{multiplicative} cuts, of {type} $(\otimes /\parr)$;
    \emph{clash} cuts, of type $(\otimes / \otimes)$ or $(\parr / \parr)$;
    \emph{glueing} cuts, of type $(\maltese / \maltese)$;
    \emph{non--homogeneous} cuts, of type $(\otimes /\maltese)$
    or $(\parr / \maltese)$,
    which are respectively called 
    \emph{reversible} and \emph{irreversible} cuts.
In a net $S$, 
a cut $\cutlink{p,q}$
is \emph{cyclic}
whenever $p$ and $q$ are targets of the same link.
\end{definition}

\begin{remark}
Each cut link occurring in a net $S$
has a type since 
a net is target--surjective.
However in a module this isn't true:
for instance in the module $\cutlink{p,q}$ 
consisting of a single cut link,
the type of the cut link is not defined.
\end{remark}

\begin{remark}
  The inputs of a cut link $\cutlink{p,q}$ are ordered, 
  making the two links 
  $\cutlink{p,q}$ 
  and $\cutlink{q,p}$ distinct.
  However (up to isomorphism) this plays no role during cut elimination.
\end{remark}


\begin{figure*}[!t]
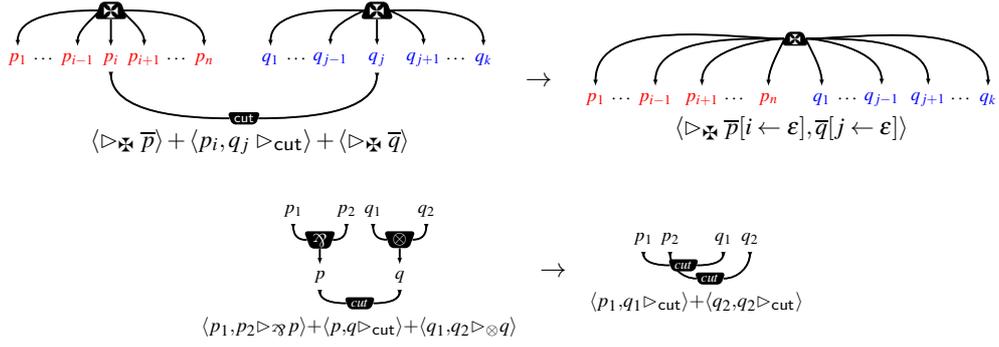

  \centering
  \begin{tabular}{c}    
    \scalebox{0.7}{\input{cut_glue.tikz}}\\
    \small$\dailink{\overline p} +
    \cutlink{p_i,q_j} +
    \dailink{\overline q} 
    $
  \end{tabular}
  $\rightarrow$
  \begin{tabular}{c}
    \scalebox{0.7}{\input{reduct_glue.tikz}}\\
    \small$\dailink{\overline p [i\leftarrow \epsilon],\overline q [j\leftarrow \epsilon]}$
  \end{tabular}
  \\
  \vspace{15pt}
  \begin{tabular}{c}
    \scalebox{0.7}{\input{cut_mult.tikz}}\\
    $
        \scriptstyle \parrlink{p_1,p_2}{p} +
        \scriptstyle \cutlink{p,q} +
        \scriptstyle \tenslink{q_1,q_2}{q}
    $
  \end{tabular}
  $\rightarrow$ 
  \begin{tabular}{c}
    \scalebox{0.7}{\input{redex_mult.tikz}} \\
    $
      \scriptstyle \cutlink{p_1,q_1} +  
      \scriptstyle \cutlink{q_2,q_2}    $
      
  \end{tabular}
  \caption{Rewriting defining the homogeneous cut elimination.
  We provide a representation 
  of each hypergraph involved    
  above its expression.
  In the step of the glueing cut 
  we assume the two daimons to be distinct
  i.e. the cut is acyclic.
  In this figure $\overline p =p_1,\dots,p_n$
  while $\overline q = q_1,\dots,q_k$.}\label{fig:homcutelim}
\end{figure*}

\begin{figure*}[ht]
  \begin{tabular}{c}
    \scalebox{0.7}{\input{cut_daiparr.tikz}}\\
    \small $\parrlink{q_1,q_2}{q} + \cutlink{q,p} + \dailink{\overline a, p , \overline b}$
  \end{tabular}
  $\rightarrow$ 
  \begin{tabular}{c}
      \scalebox{0.7}{\input{redex_daiparr.tikz}} \\
    \small $\cutlink{q_1,p^1}+ \cutlink{q_2,p^2} 
    + \dailink{\overline {\sigma(a)} , p^1,\overline{\tau(b)}}
    + \dailink{\overline {\sigma'(a)} , p^2,\overline{\tau'(b)}}$
  \end{tabular}
  \\
  \vspace*{15pt}
  
  \begin{tabular}{c}
      \scalebox{0.7}{\input{cut_daitens.tikz}}\\
      \small $\tenslink{p_1,p_2}{p} + \cutlink{p,q} + \dailink{\before i q, q , \after i q}$
  \end{tabular}
  $\rightarrow$ 
  \begin{tabular}{c}
    \scalebox{0.7}{\input{reduct_daitens.tikz}}\\
    \small $\cutlink{p_1,q^1} + \cutlink{p_2,q^2} + \dailink{\before i q, q^1 , q^2 , \after i q}$
  \end{tabular}
  \caption{Rules
defining the non--homogeneous cut elimination.
In the elimination of the $(\parr/\maltese)$ cut - first row -
$\overline  a = (a_1,\dots,a_n)$
and $\overline b=(b_1,\dots,b_m)$
while 
$\sigma(\overline a)=(a_{\sigma(1)},\dots,a_{\sigma(k)})$
,
$\sigma'(\overline a)=(a_{\sigma'(1)},\dots,a_{\sigma'(k')})$
,
$\tau(\overline b)=(b_{\tau(1)},\dots,b_{\tau(h)})$
,
$\tau'(\overline b)=(b_{\tau'(1)},\dots,b_{\tau'(h')})$
(with $n=k+k'$ and $m=h+h'$)
are sequences that define a partition of
$\{ a_1,\dots,a_n,b_1,\dots,b_n \}$
more precisely 
$\{a_1,\dots,a_n\} = \{ a_{\sigma(1)},\dots,a_{\sigma(k)} , a_{\sigma'(1)},\dots,a_{\sigma'(k')}\}$
and 
$\{b_1,\dots,b_n\} = \{ b_{\tau(1)},\dots,b_{\tau(h)} , b_{\tau'(1)},\dots,b_{\tau'(h')}\}$,
and $\sigma,\sigma',\tau,\tau'$ are permutations.
Furthermore 
$p^1,p^2,q^1,q^2$ are fresh positions.
The figure is slightly misleading: 
$q_1$ and $q_2$ 
may be elements of $\overline a$
or $\overline b$ (in the first row)
while $p_1$ and $p_2$ 
may be elements of $q_1,\dots,q_{i-1},q_{i+1},\dots,q_n$ (in the second row),
these cases are illustrated in \autoref{fig:cutelimcomplement}.
This has an important consequence: a cut can belong to a cycle 
and still be reducible (\autoref{rem:cyclicReduction}).
}\label{fig:nhomcutelim}
\end{figure*}

Multicative nets 
comes with their notion of computation 
called \emph{cut elimination}:
it is a rewriting on nets 
and more precisely it rewrites 
a redex (that is a sub--net made of a single cut link
and two non--cut links)
into redexes or daimons (in the very specific case of glueing cuts).
Up to isomorphism, 
how a redex is rewritten
depends solely on
its type (\autoref{def:cuttypes}).



\begin{definition}\label{def:homcut}
The relation of \emph{homogeneous cut elimination}
on unordered nets
is denoted by $\rightarrow_h$ 
and it is the rewriting relation
defined as the contextual closure (with respect to the sum) of
the relation defined in \autoref{fig:homcutelim}.
\end{definition}

\begin{remark}\label{rem:netlift}
  The (homogeneous) cut elimination procedure
  on unordered nets leave the conclusions unchanged.
  As a consequence the homogeneous cut elimination 
  can be lifted from unordered nets to nets:
  whenever two unordered nets are such that 
  $S\rightarrow S'$,
  for any arrangement $\mathbf a$ of $S$
  we have $(S,\mathbf a)\rightarrow (S',\mathbf a)$.
\end{remark}

The following result is easily established, in particular since the number of links strictly decreases
during homogeneous cut elimination.

\begin{proposition}\label{thm:snhomogeneous}\label{thm:confluencehomogeneous}
  Homogeneous cut elimination is confluent and strongly normalizing.
\end{proposition}

%


\begin{definition}\label{def:nohomcut}
  The non homogeneous reduction
  is denoted $\rightarrow_{nh}$
  and it
  is defined on unordered nets 
  as the contextual closure 
  of the relation given in \autoref{fig:nhomcutelim}.
\end{definition}

\begin{remark}
  The non--homogeneous reduction 
  preserves the conclusion 
  of the nets,
  hence it can be lifted to ordered nets -- 
  as in remark \ref{rem:netlift}.
\end{remark}



\begin{remark}\label{rem:nondet-proofsearch}
  In the framework of 
  Multiplicative Linear Logic (\autoref{sect:MLL}, \autoref{fig:mllrules}),
  non homogeneous cut elimination simulates 
  proof search in the sequent calculus:
  \begin{center}
  \scalebox{0.7}
  {
    $\vlderivation{
      \abovetwo{\Gamma,A\parr B}{\mathsf{cut}}
      {\aboveone{\Gamma,A\parr B}{\maltese}\finish}
      {\aboveone{A^\bot\otimes B^\bot , A\parr B}{\parr}
        {\abovetwo{A^\bot\otimes B^\bot ,A ,B}{\otimes}
          {\aboveone{A^\bot,A}{\maltese}\finish}
          {\aboveone{B^\bot,B}{\maltese}\finish}
          }
        }
    }$
  }
  $\rightarrow^*$
  \scalebox{0.7}{
    $\vlderivation{
      \aboveone{\Gamma,A\parr B}{\parr}
      {\aboveone{\Gamma,A,B}{\maltese}\finish}
    }$
  }
  \qquad 
  \scalebox{0.7}
  {
    $\vlderivation{
      \abovetwo{\Gamma,A\otimes B}{\mathsf{cut}}
      {\aboveone{\Gamma,A\otimes B}{\maltese}\finish}
      {\aboveone{A^\bot\parr B^\bot , A\otimes B}{\parr}
        {\abovetwo{A^\bot, B^\bot ,A \otimes B}{\otimes}
          {\aboveone{A^\bot,A}{\maltese}\finish}
          {\aboveone{B^\bot,B}{\maltese}\finish}
          }
        }
    }$
  }
  $\rightarrow^*$
  \scalebox{0.7}{
    $\vlderivation{
      \abovetwo{\Gamma,A\otimes B}{\otimes}
      {\aboveone{\Gamma_1,A}{\maltese}\finish}
      {\aboveone{\Gamma_2,B}{\maltese}\finish}
    }$
  }
\end{center}
  This also illustrates the non determinism 
  of the $(\maltese/\parr)$ reduction step
  which corresponds to proof search
  on a formula of the form $A\otimes B$:
  going from bottom to top 
  the 
  $\otimes$--introduction rule splits the context $\Gamma$,
  which is a non deterministic process.
  A consequence of non determinism
  is the loss of confluence for cut elimination (but not of strong normalisation, \autoref{prop:SN});
  since splitting the context is irreversible,
  a net can have different normal forms,
  like the second net 
  of figure \ref{fig:notconflu1} (from left to right)  
  which coincides with the second net of figure \ref{fig:notconflu2}:
  this same net reduces, following the two figures,
  to two different normal forms.
\end{remark}



\begin{remark}
  A cyclic cut is a glueing cut.
  Indeed, 
  given a  cyclic cut link $\cutlink{p,q}$ in a net,
  because $p$ and $q$ belong to the target of a same link $e$
  and the only links which may have several targets are daimon links 
  it follows that $e$ is a daimon link.
\end{remark}

\begin{remark}\label{rem:acyclicCuts}
  The side condition 
  of \autoref{fig:homcutelim}
  entails that a cyclic cut is not reducible:
  for example 
  the net $\dailink{p,q} + \cutlink{p,q}$
  is a net in normal form.
\end{remark}

\begin{remark}\label{rem:invalidCuts}
  A cut link 
  which is not reducible 
  is either a clashing cut 
  or a cyclic glueing cut.
  Notice, however,
  that while clashing cuts
  never disappear during cut elimination,
  cyclic cuts may disappear 
  (see \autoref{fig:cycleParrBreak}).
\end{remark}

\begin{remark}\label{rem:cyclicReduction}
In the standard framework of $\mll$ proof structures
the cut elimination
of an axiom against a cut 
is defined as the identification of the two extreme positions,
therefore eliminating such a cut may create \emph{loops} (\autoref{sect:hypergraphs}).
To avoid loops from occurring during cut elimination
an ad hoc condition is usually added (see for example \cite{LLhandbook}).
In our framework,
this condition is 
the rather natural side condition of \autoref{fig:homcutelim}.
\end{remark}

{\begin{figure*}[!thb]
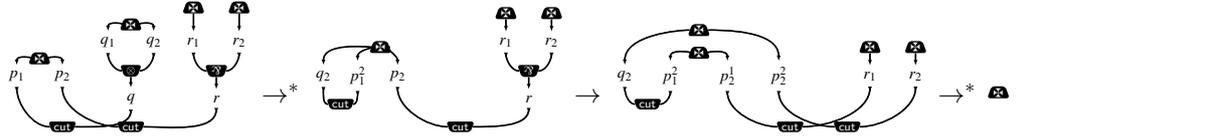
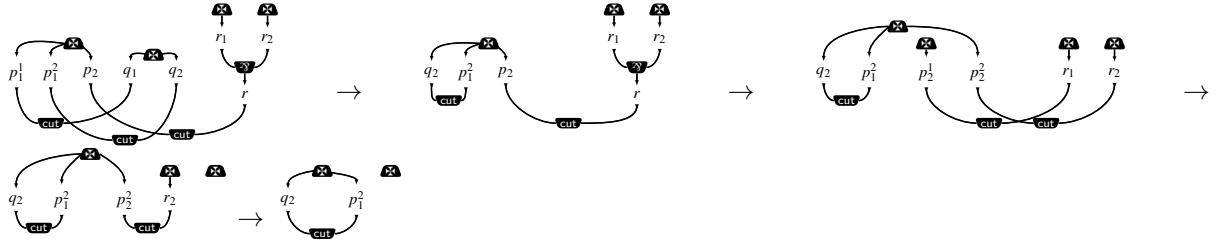

  \centering
  \begin{subfigure}{\textwidth}
    \scalebox{0.6}{\input{interactiona2.tikz}}
    $\rightarrow$
    \scalebox{0.6}{\input{interactiona3.tikz}}
    $\rightarrow$
    \scalebox{0.6}{\input{interactiona4.tikz}}
    $\rightarrow$
    \scalebox{0.6}{\input{interactiona5.tikz}}    
    $\rightarrow$
    \scalebox{0.6}{\input{interaction6.tikz}}    
    $\rightarrow$
    \scalebox{0.6}{\input{interaction7.tikz}}
    \caption{Eliminating first 
    the \emph{irreversible cut $(\maltese/\parr)$}
    produces a net\footnote{The $(\maltese/\parr)$ reduction step is not deterministic but in this very special case
    any choice yields the same net.}
     which cannot normalize in $\maltese_0$.}\label{fig:firstirr}
  \end{subfigure}
  \begin{subfigure}{\textwidth}    
    \scalebox{0.6}{\input{interactiona2.tikz}}
    $\rightarrow^*$
    \scalebox{0.6}{\input{interaction3.tikz}}
    $\rightarrow$
    \scalebox{0.6}{\input{interaction41.tikz}}
    $\rightarrow^*$
    \scalebox{0.6}{\input{interaction51.tikz}}
    \caption{Eliminating 
    the reversible cut $(\maltese/\otimes)$
    produces a cycle which can be eliminated 
    by the elimination of the $(\maltese/\parr)$
    cut remaining, hence 
    that net can normalize in $\maltese_0$.}\label{fig:notconflu1}
    \end{subfigure}
\begin{subfigure}{\textwidth}    
      \scalebox{0.6}{\input{interaction2.tikz}}
      $\rightarrow$
      \scalebox{0.6}{\input{interaction3.tikz}}
      $\rightarrow$
      \scalebox{0.6}{\input{interaction42.tikz}}
      $\rightarrow$
      \scalebox{0.6}{\input{interaction52.tikz}}          
      $\rightarrow$
      \scalebox{0.6}{\input{interaction52End.tikz}}    
      \caption{Non determinism 
      also comes from the choice 
      of how we reduce $(\parr/\maltese)$ cuts,
      different choices leading to different normal forms:
      the ``wrong'' choice results in a net which cannot normalize to $\maltese_0$.}\label{fig:notconflu2}
      \end{subfigure}
      \caption{Non homogeneous cut eliminations contains two sources of non--determinism.}
\end{figure*}



\begin{remark}\label{rem:mllCutElim}
  Notice
  that 
  whenever daimons are binary and typed by dual atomic formulas 
  the cut elimination procedure
  for $\mlldai$ defined in \autoref{def:homcut} 
  is exactly the standard cut elimination procedure 
  for $\mll$ \cite{girard_1987},\cite{LLhandbook}.
\end{remark}


The rewriting rule, denoted $\rightarrow$, 
associated with cut elimination 
is the union 
of the homogeneous and non--homogeneous cut elimination 
i.e. $\rightarrow_{h}\cup\rightarrow_{nh}$.
We write $S\reduce c S'$,
when $S'$ is obtained from $S$ by eliminating the cut $c$.
We write by $S \redmult S'$ (resp. $S \rednotmult S'$)
whenever $S\reduce c S'$
and $c$ is multiplicative (resp. not multiplicative).
Given two binary relations
  $R_1$ and $R_2$ on a set $X$
  we denote by $R_1\at R_2$
  their composition,
  i.e. for two $x,y\in X$
  $x R_1\at  R_2 y$
  if and only if there exists $z$
  such that $x R_1 z$
  and $z R_2 y$.

\begin{proposition}\label{prop:SN}
  Cut elimination is strongly normalising, furthermore:
  \begin{enumerate}
    \item \label{item:factorisation}
    $\rightarrow ^*$ can be factorised as $\redmult^*\at\rednotmult^*$.
    \item \label{item:delayParrs}
    If $c$ is a $(\parr/\maltese)$ cut in $S$;
    if $S\reduce{c}\at \rightarrow ^* S'$
    then $S\rightarrow ^*\at\reduce{c}  S'$.
    \item \label{item:earlyNonPar}
    If $c$ is not a $(\parr/\maltese)$ cut in $S$;
    if $S\rightarrow ^* \at\reduce{c} S'$
    then $S\reduce{c}\at \rightarrow ^* S'$.
  \end{enumerate}
\end{proposition}

\begin{figure*}[!thb]
  \begin{minipage}{0.3\textwidth}
  \begin{subfigure}[b]{\linewidth}
      \centering
      \scalebox{0.8}{\begin{tabular}{c @{\hskip 2pt} c@{\hskip 2pt} l}
          $A,B$ 
          & $\triangleq\,$ & 
              $X \in \pvarset $ \\
          & & $\;\mid\; A\parr B \;\mid\; A\otimes B$ \\
          $\hseq_1,\hseq_2 $  
          & $\triangleq \,$&
              $A\in \formulaset $ \\
          & & $\;\mid\; \hseq_1,\hseq_2 \;\mid\; 
          \hseq_1\parallel\hseq_2$ 
        \end{tabular}  }  
      \caption{Grammar defining $\formulaset$ (first two rows),
      and grammar defining $\hseqset$ (last two rows).
      }\label{fig:mllgrammar}
  \end{subfigure}        
\end{minipage}
  \hfill
  \begin{minipage}{0.65\textwidth}
  \begin{subfigure}{\linewidth}
      \centering
      \scalebox{0.85}{$(A\parr B)^\bot = A^\bot\otimes B^\bot \qquad 
  (A\otimes B)^\bot =  A^\bot \parr B^\bot$}
  \caption{De Morgan laws lifting the involution $(\cdot)^\bot$
  from $\pvarset$ to $\formulaset$.}\label{fig:demorgan}
  \end{subfigure}

  \vspace{8pt}

  \begin{subfigure}[b]{\linewidth}
      \centering
      \scalebox{0.75}{$\ruleone{\maltese}{\Gamma}{}
      \quad 
      \ruleone{\parr}{\Gamma,A\parr B}{\Gamma , A ,B}
      \quad 
      \ruletwo{\otimes}{\Gamma,\Delta,A\otimes B}{\Gamma,A}{\Delta,B}
      \quad 
      \ruletwo{\mathsf{cut}}{\Gamma,\Delta}{\Gamma,A}{\Delta, A^\bot}
      \quad 
      \ruleone{\mathsf{ex}}{\Gamma,B,A,\Delta}{\Gamma,A,B,\Delta}
      \quad 
      \ruleone{\mathsf{ax}}{A,A^\bot}{}$}
      \caption{Rules used for constructing the proof trees.
      The rules $(\maltese,\parr,\otimes,\mathsf{cut},\mathsf{ex})$
      define the $\mlldai$ fragment.
      Substituting the $(\maltese)$--daimon rule 
      with the $(\mathsf{ax})$--axiom rule 
      results in the fragment $\mll$,
      that is $(\mathsf{ax},\parr,\otimes,\mathsf{cut},\mathsf{ex})$.}\label{fig:mllrules}
  \end{subfigure} 
\end{minipage}
\caption{Grammar of formulas and (hyper)sequent, de Morgan laws and inference rules.}
\end{figure*}
\begin{figure*}[!htb]
  \centering 
  \small$\link[\ell]{\overline c}{\overline a , p_1 , \overline b} + \parrlink{p_1,p_2}{p}$
       \small$\switchto^l$ 
      \small$\link[\ell]{\overline c}{\overline a, p ,\overline b}$
  \qquad 
  \small$\link[\ell]{\overline c}{\overline a , p_2 , \overline b} + \parrlink{p_1,p_2}{p}$
  \small$\switchto^r$
  \small$\link[\ell]{\overline c}{\overline a, p ,\overline b}$
  \caption{The two cases (left and right) defining the switching rewriting.
  The left reduction $\switchto^l$ 
  destroys $p_1$ and makes $p_2$ a conclusion;
  while the right reduction $\switchto^r$
  destroys $p_2$ and makes $p_1$ a conclusion.}\label{fig:switch}
\end{figure*}  }

\begin{figure*}[!thb]
  \centering
  \begin{tabular}{c @{\hskip 18pt} c@{\hskip 18pt} c@{\hskip 18pt} c@{\hskip 18pt} c}
      $
    \ruleone{{\mathsf{\maltese}}}{\Gamma}{} 
    $
    &
    $\vlderivation{
  \vliin{}{{\scriptstyle\mathsf{cut}}}{ \Gamma,\Delta}
  {\vlpr{}{\pi_1}{ A,\Gamma}}
  {\vlpr{}{\pi_2}{ A^\bot,\Delta}}
  }$ 
  &
  $
  \vlderivation{
  \vliin{}{{\scriptstyle\otimes}}{ \Gamma,\Delta, A\otimes B}
  {\vlpr{}{\pi_1}{ A,\Gamma}}
  {\vlpr{}{\pi_2}{ B,\Delta}}
  }$
  &
  $
  \vlderivation{
  \vlin{}{{\scriptstyle\parr}}{ A\parr B,\Gamma}
  {\vlpr{}{\pi_0}{ A, B ,\Gamma}}
  }$
  &
  $
  \vlderivation{
  \vlin{}{{\scriptstyle \mathsf{ex}}}{ \Gamma,A,B,\Delta}
  {\vlpr{}{\pi_0}{\Gamma, B,A,\Delta}}
  }$ \\
  \scalebox{0.9}{$\dailink{p_1,\dots,p_n}$}
  & 
  \scalebox{0.9}{$S_1 + S_2 + $}
  &
  \scalebox{0.9}{$S_1 + S_2 + $}
  &
  \scalebox{0.9}{$S_0 +   $}
  &
  \scalebox{0.9}{$(S_0,a) $} \\
  &
  \scalebox{0.9}{$\cutlink{S_1(1) , S_2(1)}$}
  &
  \scalebox{0.9}{$\tenslink{S_1(1),S_2(1)}{p}$}
  &
  \scalebox{0.9}{$\parrlink{S_0(1),S_0(2)}{p}$}
  &
  \end{tabular}
  \caption{
      Induction defining the relation $\reps$.
      The proof in the first row 
      is represented by a net below it 
      in the second row.
      The position $p$ is always supposed fresh.
      In each case and for each $0\leq i \leq 2$, 
      $S_i$
      is a net which represent $\pi_i$ i.e. 
      $S_i \reps \pi_i$.
  In the case of the exchange rule 
  we explicitly mention the arrangement i.e. the order of the conclusion 
  and assume 
  $(S_0,a')\reps \pi_0$
  and $a(i) = a'(i)$
  whenever $i\leq \size \Gamma$
  or $\size\Gamma +2 < i$.
  On the other hand,
  $a'(\size \Gamma +1) = a(\size\Gamma+2)$
  and $a'(\size\Gamma +2) = a(\size\Gamma +1)$.
  }\label{fig:proofnet}
\end{figure*}

\section{Multiplicative Linear Logic and proof nets}\label{sect:MLL}

We define the 
well--known notion of proof net \cite{girard_1987}
in our setting:
in the presence of the generalised axiom $(\maltese)$,
proof nets are similar to the \emph{paraproof nets} of Curien \cite{curien:criterions}
(which come from Girard Ludics \cite{girard:ludics}).
We then formulate the Danos--Regnier criterion \cite{danos89}:
testing the acyclicity and connectedness 
of (several) graphs allows to determine 
whether a net is a (para)proof net or not \cite{curien:criterions}.


We fix a countable set 
$\pvarset$ 
of \emph{propositional variables}.
The set $\pvarset$ comes with an (explicit) involution $(\cdot)^\bot$;
for each atomic variable 
$X$ there exists its \emph{dual} atomic variable
$X^\bot$ in $\pvarset$.
    The set $\formulaset$
    of \emph{formulas} of multiplicative linear logic 
    is defined by 
    the grammar in \autoref{fig:mllgrammar}.
    The involution $(\cdot)^\bot$ is lifted 
    from $\pvarset$ to $\formulaset$
    as in \autoref{fig:demorgan}.
    The set $\hseqset$
    of \emph{hypersequents} is defined 
    by the grammar in \autoref{fig:mllgrammar},
    a \emph{sequent} is an hypersequent 
    without the \emph{parallel} `$\parallel$'
    constructor.
    The introduction of hypersequents is naturally suggested by the 
    constructions on types (\autoref{sect:RealModel}):
    indeed 
    as the interpretation of the $\parr$--connective
    is based on the interpretation 
    of the ``$,$''--connective,
    the interpretation of the $\otimes$--connective
    relies on that of the ``$\parallel$''--connective 
    (\autoref{def:constructionOnTypes} and \autoref{def:interpretmll}).
    Technically hypersequents are necessary in \emph{our} proof of the completeness theorem (\autoref{thm:completeness}).

    A \emph{proof} of $\mll$
    (resp. $\mlldai$)
    is a tree constructed 
    using the rules 
    $(\mathsf{ax},\parr,\otimes,\mathsf{cut},\mathsf{ex})$
    (resp. $(\maltese,\parr,\otimes,\mathsf{cut},\mathsf{ex})$)
    of \autoref{fig:mllrules}.

\begin{definition}\label{def:proofrep}
    A net $S$
    \emph{represents}\footnote{In the standard Linear Logic terminology $\pi$ 
    is a sequentialisation of the proof net $S$.}
    a proof $\pi$ of $\mlldai$,
    denoted $\pi \reps S$
    or $S\reps \pi$,
    whenever the relation 
    defined in \autoref{fig:proofnet} holds.
    A net represents a proof of $\mll$
    whenever it represents a proof of $\mlldai$
    where every sequent conclusion of a $(\maltese)$--rule 
    has shape $A,A^\bot$ for $A\in\formulaset$.
    A representation of a proof $\pi$
    is a net $S$ which represents $\pi$.
    A \emph{proof net} of $\mlldai$ (resp. $\mll$)
    is a net which represents 
    a proof of $\mlldai$ (resp. $\mll$):
    we say that $S$ is \emph{correct}.
    A net $S$ is \emph{correctly typeable}\footnote{Notice 
    that with the expressions ``correctly typeable'' we mean here 
    that the net is both correct (it represents a proof)
    and that we can label 
    its conclusions with the formulas of $\Gamma$.} 
    by a sequent $\Gamma$
    whenever it represents a proof of $\Gamma$
    in $\mlldai$.
\end{definition}

\begin{notation}
    Let $\mathfrak P$
    denote $\mll$ or $\mlldai$ and
    let $S$ be a net.
    We write $S\vdash_{\mathfrak P} \Gamma$
    whenever there exists a proof 
    $\pi$ in $\mathfrak P$
    such that 
    $S$ is the representation of $\pi$.  
    Furthermore 
    we denote $\proofs{\Gamma : \mathfrak P}$
    the set of all the nets $S$ 
    such that $S\vdash_{\mathfrak P} \Gamma$.
\end{notation}

    A \emph{substitution} is a map 
    $\theta : \pvarset \rightarrow \formulaset$
    such that $\theta (X^\bot) = \theta(X)^\bot$ for each $X\in \pvarset$.
    A substitution can be lifted 
    to formulas and hypersequents by induction: 
    $\theta (A\otimes B) = \theta (A) \otimes \theta (B)$ ; 
    $\theta (A\parr B) = \theta (A) \parr \theta (B)$ ; 
    $\theta (A\parallel B) = \theta (A) \parallel \theta (B)$ ;
    $\theta (A , B) = \theta (A)  , \theta (B)$.
    Given two hypersequents,
    we denote $\Delta \leq \Gamma$
    whenever there exists a substitution $\theta$
    such that $\theta\Delta = \Gamma$.    

\begin{restatable}{proposition}{propSubstiProof}\label{prop:smallproof}
    Let $\Gamma$ and $\Delta$ be two sequents 
    and suppose $\Delta\leq \Gamma$.
    For any net $S$:
    (1) if $S\vdash_{\mlldai}\Delta $  then $ S\vdash_{\mlldai}\Gamma$
    and 
    (2) 
    if $S\vdash_{\mll}\Delta $ then $ S\vdash_{\mll}\Gamma$.
  \end{restatable}

\begin{definition}
    The \emph{switching} rewriting 
    is defined on unordered nets 
    as the contextual closure of 
    the rules in \autoref{fig:switch}.
    A \emph{switching} of a net $S$
    is a normal form of $S$ for the switching rewriting: 
    we often denote it $\sigma S$.
\end{definition}

\begin{remark}
    The switching rewriting strongly normalizes 
    since every step reduces the number of links of the net.
    The rewriting is also non-deterministic 
    and non-confluent,
    every normal form is a par--free net.
    The switching rewriting can be lifted to (ordered) nets;
    with the notations of \autoref{fig:switch}
    whenever 
    an unordered net $\body S$ with $n$ conclusions 
    is such that $\body S\switchto^l \body {S'}$
    we define $(\body {S},\mathbf a) \switchto^l (\body {S'},\mathbf a')$
    where 
    $\mathbf a' (i) = \mathbf a (i)$ for each $1\leq i \leq n$
    and $\mathbf a' (n+1) = p_2$
    i.e. the new conclusion is made last conclusion
    (similarly we can define it for the case $\switchto^r$).    
\end{remark}

\begin{definition}
    The \emph{undirected multigraph}\footnote{Recall that a multigraph is a graph where 
    two vertices may be connected by several edges (not to be confused with the notion of hypergraph 
    of \autoref{def:hypergraph}). 
    The function $\border$ maps each edge to its endpoints.} induced by 
    two partitions 
    $P$ and $Q$ of a set $X$
    is
    $(V,E,\border)$ denoted $\graph (P,Q)$ where:
    (1)
    $V = \{1\}\times P\cup \{2\}\times Q$
    the vertices are the classes of $P$ and $Q$
    (as a disjoint union);
    (2) 
    $E = X$;
    (3)
    For any edge $x$ in $X$;
    $\border (x) = \{ (1,P_x) , (2,Q_x) \}$
    where $P_x\in P$ is such that $x\in P_x$
    and $Q_x\in Q$ is such $x\in Q_x$.

    Two partitions $P$
    and $Q$ of a set $X$
    are \emph{orthogonal} 
    if the multigraph $\graph(P,Q)$ is 
    acyclic and connected.
\end{definition}

\begin{definition}
    In a net $S$ 
    denote $p\geq_S q$
    the relation which holds 
    whenever there exists a link $e$
    such that $p\in\source(e)$
    and $q\in\target(e)$.
    Denote $\geq^*_S$ its reflexive and transitive closure;
    a position $p$ is
    \emph{above} a position $q$
    whenever $p\geq^*_S q$.
    Given a position $q$ 
    we denote $q\uparrow^i S$ 
    the set of \emph{initial} positions which are above $q$ in $S$.
\end{definition}

\begin{remark}
    Given a \emph{cut--free} net $S$ with conclusions $p_1,\dots,p_n$
    the sets $p_1\uparrow^i S,\dots, p_n\uparrow^i S$
    form a partition of the initial positions of $S$.
    We denote this partition $\uparrow^i S$.
\end{remark}

\begin{notation}
    Let $S$ be a net
    and let
    $\{d_1,\dots,d_n\}$ be the set of daimon links of $S$.
    The partition 
    $\{\target (d_1) ,\dots, \target (d_n)\}$
    on the set of initial positions of $S$
    is denoted by $\daipart S$.
  \end{notation}

  Reformulated in the context 
  of hypergraphs we get the following 
  theorem from \cite{danos89}.

\begin{theorem}[\!\!\cite{curien:criterions},\cite{danos89}]\label{thm:standardcrit}
    Given a cut--free net $S$,
    the following assertions are equivalent:
    \begin{enumerate}
        \item 
        $S$ is a proof net of $\mlldai$;
        \item 
        For every switching $\sigma S$ of $S$,
        the partitions 
        $\daipart S$ and $\uparrow^ i \sigma S$
        of the set of initial positions of $S$
        are orthogonal;
        \item 
        Every switching $\sigma S$ of $S$
        is acyclic and connected\footnote{
            We refer to the graph naturally induced 
            by the net $\sigma S$.
        }.
    \end{enumerate}
\end{theorem}

\section{Interaction of nets, orthogonality, and types}\label{sect:RealModel}

\def\exchange#1{\mathsf{ex}_{#1}}

We define how nets can \emph{interact} 
and if the interaction of two nets leads to the
$\maltese$-link with no outputs ($\maltese_0$) we say they are \emph{orthogonal}. 
This recalls classical realisability proposed 
by J.-L. Krivine \cite{krivine:realiclassi}, 
where (the closure by antireduction of) the set 
$\{\maltese_{0}\}$ will play the role of the \emph{pole}.
Notice, however, that our setting is fully symmetrical:
both the elements of
truth values and
falsity values
are nets.


\bigskip

The notion of ordered hypergraph 
and \emph{arrangement} introduced in \autoref{sect:hypergraphs}
will now explicitly come into play
as it is necessary for defining the interactions of nets (see \autoref{fig:interactionExample}).  
We will denote by $\# S$
the number of outputs of a net $S$.
Given a partial function $f:\mathbb N \rightarrow E$
with a finite domain of cardinality $n$
and ordered as $i_1 < i_2 < \dots < i_n$,
the \emph{collapse} of $f$, denoted $f\collapse$,
is
the total function 
with domain $[1;n]$
such that $f\collapse(m) = f(i_m)$
for any integer $1\leq m \leq n$.

\begin{definition}\label{def:interaction}
  Let  $S = (\body S , \arrange S)$ and $T = (\body T,\arrange T)$ 
  be two nets 
  and $k= min(\# S, \# T) $,
  we define their \emph{interaction}
  $S::T = (\body {S::T} , \arrange {S::T})$ as:
  \begin{center}
    \scalebox{0.8}{$\body {S:: T}
    \;\triangleq\;
    \body S + \body T +
    \sum_{1 \leq i \leq min(\# S, \# T) }
    \cutlink{ S(i) , T(i) }
    \qquad 
    \arrange{S::T} 
    \;\triangleq \;
    \left\{ \begin{array}{lr}
      \emptyset & \text{when } \# S = \# T\\
      \arrange S \restrict {[k+1;\# S]} \collapse 
            & \text{when } \# S > \# T\\
      \arrange T \restrict {[k+1;\# T]} \collapse 
            & \text{when } \# S < \# T
      \end{array}
    \right. $  }
  \end{center}
\end{definition}

\begin{figure}
\begin{subfigure}{0.3\textwidth}
  \centering
  \scalebox{0.8}{\input{interactionDef.tikz}}
  \caption{Representation of the interaction $S::T$
  of two nets $S = \dailink{a,b,c} + \parrlink{a,b}{d}$
  and $T = \dailink{p,q} + \tenslink{p,q}{r}$.}\label{fig:interactionExample}
\end{subfigure}
\hfill
\begin{subfigure}{0.67\textwidth}
  \centering
  \scalebox{0.7}{\input{orthoExample1.tikz}}
  $\rightarrow$
  \scalebox{0.7}{\input{orthoExample2.tikz}}
  $\rightarrow$
  \scalebox{0.7}{\input{orthoExample3.tikz}}
  $\rightarrow$
  \scalebox{0.7}{\input{orthoExample4.tikz}}
  \caption{The cut elimination procedure applied to $S'::T'$ leads to $\maltese_0$,
  showing that $S'\perp T'$. In this figure $S' = \dailink{q}$ and $T' = \dailink{p_1} + \dailink{p_2} + \tenslink{p_1,p_2}{p}$.}
  \label{fig:orthoExample}
\end{subfigure}
\caption{The interaction of two nets (\autoref{def:interaction})
and two orthogonal nets (\autoref{def:orthogonalityNets}).}
\end{figure}


\begin{definition}\label{def:orthogonalityNets}
  Two nets $S_1$ and $S_2$
  are \emph{orthogonal}
  if
  $S_1:: S_2\rightarrow^* \maltese_0$\footnote{Note that we require the \emph{existence} 
  of such a reduction, not all reductions 
  need to behave this way.}: when this holds we write
  $S_1\perp S_2$. For a net $S$ and a set of nets $\Lambda$, if for every $\lambda\in\Lambda$ we have $S\perp\lambda$ we write $S\perp\Lambda$.
\end{definition}

\begin{remark}
  Since cut links are asymmetric,
  namely $\cutlink{p,q}$
  and $\cutlink{q,p}$ are distinct nets,
  the interactions 
  $S::T$ and $T::S$
  are not the same net.
  However, this has no consequence on cut elimination
  because the reduction steps 
  do not depend on the order
  of the inputs of a cut link.
  Thus 
  $S::T$ reduces to $\maltese_0$
  if and only if $T::S$ does,
  and as expected the relation of orthogonality is 
  symmetric.
\end{remark}

\begin{definition}\label{def:type}
  Given a set $A$ of multiplicative nets, we define 
  the \emph{orthogonal} of 
  $A$ as 
  $A^\bot=
  \{ P \mid \forall R\in A, P\perp R\}$. 
  A \emph{type} $\mathbf{A}$ is a set of multiplicative nets
  such that $\mathbf{A}^{\bot\bot}=\mathbf{A}$.\footnote{Equivalently,
  a type is a set $\mathbf A$ such that $\mathbf{A}=B^\bot$ for some set $B$,
  see, for instance, \cite{JoinetSeiller}.}
\end{definition}


\begin{remark}\label{rem:nbconclusionTypes}  
Since cut elimination
preserves the conclusions of a net
and $\maltese_0$ has no output,
two orthogonal nets have 
the same number of conclusions.
Thus, 
for every type $\mathbf A$, 
for every $R\in\mathbf A$ and 
for every $S\in\mathbf A^{\bot}$,
the nets $R$ and $S$ have the same number of conclusions: we 
denote by $\# \mathbf A$ the number of conclusions of the nets in $\mathbf A$.
Obviously 
$\# \mathbf A=\# \mathbf A^\bot$.
\end{remark}

\begin{remark}\label{rem:clashcut}
  Clash cuts are preserved during cut elimination,
  thus a net containing such a cut cannot reduce to $\maltese_0$.
  Hence, there cannot be 
  two nets $S$ and $S'$ respectively in 
  $\mathbf A$ and $\mathbf A^\bot$ 
  such that their $i$th conclusions $S(i)$
  and $S'(i)$
  are both outputs of a $\parr$--link (or $\otimes$--link):
  their interaction $S::S'$
  contains a clash cut 
  and thus the nets cannot be orthogonal.
\end{remark}

\begin{remark}\label{rem:daimonOrthogonals}
  A net $S$ which is orthogonal to 
  the daimon link with a single output (i.e. $\maltese_1$)
  has a single conclusion 
  which can be 
  the output of a daimon link,
  a tensor link or a par link. 
  For instance the three cut--free nets
  $\dailink{p}$,
  $\dailink{p_1}+\dailink{p_2} + \tenslink{p_1,p_2}{p}$ and
  $\dailink{p_1,p_2} + \parrlink{p_1,p_2}{p}$
  are all orthogonal to $\maltese_1$ 
  (one case is proved in \autoref{fig:orthoExample}).
\end{remark}



The following proposition 
is a key step for proving propositions
\ref{prop:dualPreC} and \ref{prop:Associativity}.

\begin{restatable}{proposition}{interAction}\label{prop:interAsAction}
  Given three net $S$ and $T$ and $R$
  such that $\# S \geq \# T + \# R$:
  the interaction $S::(T\parallel R)$
  is equal to $(S::T)::R$.
\end{restatable}

In the following definition \ref{def:constructionOnTypes}
the side condition $\#S \geq \# \mathbf A$ 
ensures that
whenever a net $S$ in $A\compo B$
interacts with a net of $T\in\mathbf A^\bot$
the remaining conclusions of $S::T$
are conclusions of $S$,
this will allow to activate \autoref{prop:interAsAction}.

\begin{definition}\label{def:constructionOnTypes}
  Given two sets of nets $\mathbf A$
  and $\mathbf B$
  their \emph{functional composition}
  denoted $\mathbf A\compo \mathbf B$,
  and their \emph{parallel composition}
  denoted $\mathbf A \parallel \mathbf B$
  are defined as follows:



\scalebox{0.8}{  
  $\mathbf A \compo \mathbf B \;\triangleq\; 
\{ S \mid
      \mbox{for any }
      T \in \mathbf A^\bot, S::T \in \mathbf B
      \mbox{ and }
      \# S \geq \# \mathbf A
      \}$
      \qquad
      $\mathbf A \parallel \mathbf B \;\triangleq\;
      \{ S \parallel T \mid
        S\in \mathbf A ,
        T \in \mathbf B \}^{\bibot}$
}
\end{definition}



\begin{remark}[Density of the parallel composition]
  \label{rem:density}
  For any two types $\mathbf A$ and $\mathbf B$ we have 
  $(\mathbf A \parallel^- \mathbf B)^{\bot}=(\mathbf A \parallel\mathbf B)^{\bot}$, where 
  $\mathbf A \parallel^- \mathbf B
  = 
  \{ S \parallel T \mid S\in \mathbf A, T \in \mathbf B\}$.
\end{remark}

\begin{restatable}[Duality]{proposition}{propdualPreC}\label{prop:dualPreC}
  Given two types
  $\mathbf A$ and $\mathbf B$:
  $(\mathbf A\parallel \mathbf B)^\bot =
   \mathbf A^\bot\compo \mathbf B^\bot$
   and
   $(\mathbf A \compo \mathbf B)^\bot
   =\mathbf A^\bot \parallel \mathbf B^\bot$.
\end{restatable}


\begin{remark}
  The duality of the constructions  (\autoref{prop:dualPreC}) 
  ensures that the set of types 
  is closed under the $\parallel$
  and $\compo$ operations.
  Moreover, the intersection of two types 
  is still a type.
  This is not the case for the union
  which needs to be closed under bi--orthogonal. 
\end{remark}



\begin{remark}
  For two types $\mathbf A$
  and $\mathbf B$ 
  the unordered nets of $\mathbf A\parallel \mathbf B$
  and of $\mathbf B \parallel \mathbf A$
  are the same,
  so as the unordered nets of 
  $\mathbf A \compo \mathbf B$
  and $\mathbf B \compo \mathbf A$.
\end{remark}

\begin{restatable}{proposition}{associativity}\label{prop:Associativity}
  Given $\mathbf A,\mathbf B$ and $\mathbf C$
  three types;
  $(\mathbf A\compo \mathbf B)\compo \mathbf C
  = \mathbf A\compo (\mathbf B \compo \mathbf C)$
    and 
    $(\mathbf A\parallel \mathbf B)\parallel \mathbf C
    = \mathbf A\parallel (\mathbf B \parallel \mathbf C)$.
\end{restatable}

\begin{definition}
  Given $\mathbf A $
  and $\mathbf B$
  two types with one conclusion,
  we define
  their \emph{tensor product} (denoted $\otimes$)
  and their \emph{compositional product} (denoted $\parr$):
  \begin{center}
  \scalebox{0.85}{$\mathbf A  \otimes \mathbf B
  \,\triangleq\,
  \{ S + \tenslink{S(1),S(2)}{p} \mid 
  S\in \mathbf A\parallel \mathbf B\}^{\bibot}
  \qquad
  \mathbf A  \parr \mathbf B 
  \,\triangleq\,
  \{ S + \parrlink{S(1),S(2)}{p} \mid 
  S\in \mathbf A\compo \mathbf B\}^{\bibot}$}   
\end{center}
    \noindent 
    where $p$ denotes a fresh position.
\end{definition}

\begin{restatable}[Duality]{proposition}{propdualC}\label{prop:dualC}
  Given $\mathbf A$ 
  and $\mathbf B$ two types 
  with one conclusion,
  $(\mathbf A \otimes \mathbf B)^\bot = \mathbf A^\bot \parr \mathbf B^\bot$
  and
  $(\mathbf A \parr \mathbf B)^\bot = \mathbf A^\bot \otimes \mathbf B^\bot$.
\end{restatable}

\begin{figure}
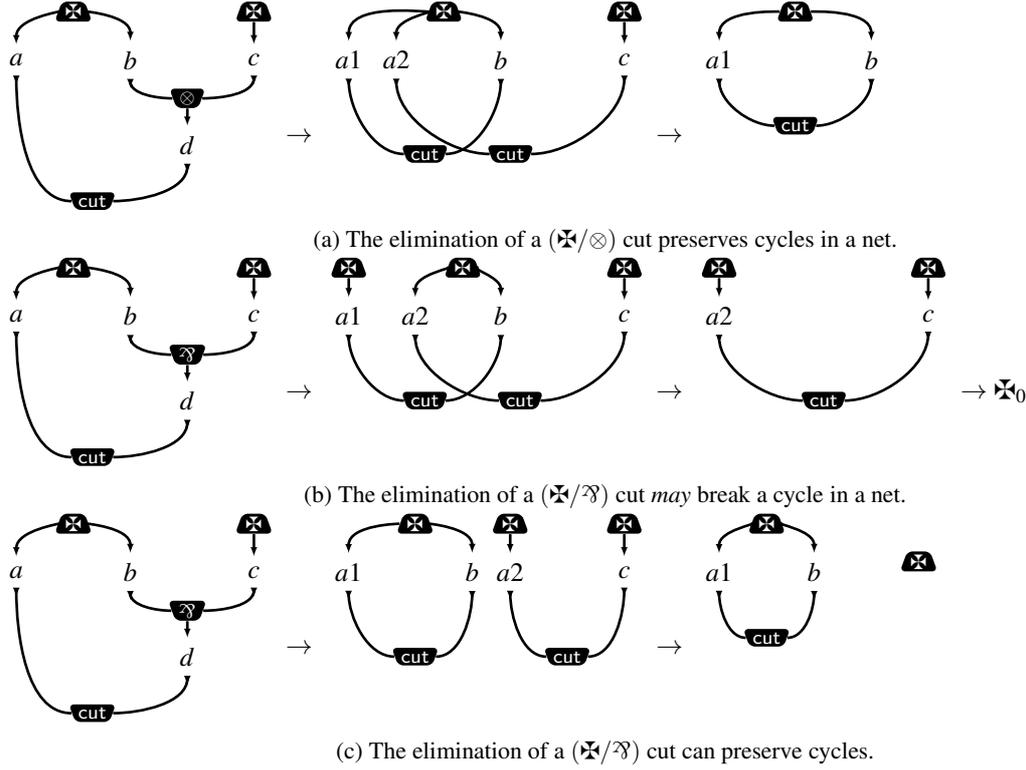

\begin{subfigure}{\textwidth}
  \input{loop1.tikz}
  $\rightarrow$
  \input{loop2.tikz}
  $\rightarrow$
  \input{loop3.tikz}
  \caption{
  The elimination of a $(\maltese/ \otimes)$ cut 
  preserves cycles in a net.}\label{cycleRemainTens}
\end{subfigure}

\begin{subfigure}{\textwidth}
  \input{loopGood1.tikz}
  $\rightarrow$
  \input{loopGoodyes2.tikz}
  $\rightarrow$
  \input{loopGoodyes3.tikz}
  $\rightarrow$
  $\maltese_0$
  \caption{
  The elimination of a $(\maltese/ \parr)$ cut 
  \emph{may} break a cycle in a net.}\label{fig:cycleParrBreak}
\end{subfigure}

\begin{subfigure}{\textwidth}
  \input{loopGood1.tikz}
  $\rightarrow$
  \input{loopGood1no1.tikz}
  $\rightarrow$
  \input{loopGood1no2.tikz}
  \caption{
  The elimination of a $(\maltese/ \parr)$ cut 
  can preserve cycles.}\label{fig:cycleRemainParr}
\end{subfigure}
\caption{The evolution of (switching) cycles and (switching) disconnections during non homogeneous cut elimination.}
\end{figure}



\section{Realisability Model: Adequacy}\label{sect:Adequacy}

We introduce our realisability 
model on untyped nets 
and prove it is adequate.
We identify 
a sufficient property 
of interpretation bases to prove adequacy (Theorem~\ref{thm:adequacy}):
for any basis 
$\ibase$ satisfying the property, a net $S$ representing an $\mlldai$ proof of a sequent $\Gamma$ 
is a realiser of $\Gamma$
i.e. it belongs to $\reali[\ibase]{\Gamma}$.
This adequacy result immediately applies to $\mll$,
since a net representing a proof of $\mll$
represents, in particular, a proof of $\mlldai$.


\bigskip

We start by giving 
an interpretation of formulas and hypersequents of multiplicative linear logic. 
We provide an interpretation of hypersequents instead of sequents 
as it turns out 
that handling hypersequents is more convenient
and proving a result on hypersequents proves it on sequents too.
However, do keep in mind 
that the proof trees we defined using \autoref{fig:mllrules}
are constructed with sequents.


\begin{definition}
  An \emph{interpretation basis} $\ibase$
  is a function
  that associates with each atomic proposition
  $X$ a type
  $\reali[\ibase]{X}$, the \emph{interpretation} of $X$,
  such that:
  \begin{itemize}[noitemsep,nolistsep]
    \item
    Each net in $\reali[\ibase]{X}$ has one conclusion.
    \item
    For any atomic proposition $X$,
    we have $\reali[\ibase]{X^\bot} \subseteq \reali[\ibase]{X}^\bot$.
  \end{itemize}
\end{definition}

\begin{definition}\label{def:interpretmll}
Given an interpretation basis $\ibase$, the \emph{interpretation} of $\mll$ formulas and of hypersequents of $\mll$ is defined 
by induction:
\begin{center}
  \scalebox{0.85}{\begin{tabular}{r @{\hskip 2pt} c@{\hskip 2pt} l}
     $\reali[\ibase]{A\otimes B}$
     & $\,\triangleq\,$&
     $\reali[\ibase]{A}\otimes \reali[\ibase]{B}$.\\
     $\reali[\ibase]{A\parr B}$
     & $\,\triangleq\,$&
     $\reali[\ibase]{A}\parr \reali[\ibase]{B}$.
  \end{tabular}}
  \qquad
  \scalebox{0.85}{\begin{tabular}{r @{\hskip 2pt} c@{\hskip 2pt} l}
     $\reali[\ibase]{\hseq_1,\hseq_2}$
     & $\,\triangleq\,$&
     $\reali[\ibase]{\hseq_1}\compo \reali[\ibase]{\hseq_2}$.\\
     $\reali[\ibase]{\hseq_1\parallel\hseq_2}$
     & $\,\triangleq\,$&
     $\reali[\ibase]{\hseq_1}\parallel \reali[\ibase]{\hseq_2}$.
    \end{tabular}}
  \end{center}
\end{definition}

\begin{remark}\label{rem:InterpretMLL}
Using duality of types (Proposition~\ref{prop:dualC}) and the properties of orthogonality one proves that for an interpretation basis $\ibase$ and an $\mll$ formula $A$ we have 
$\reali[\ibase]{A^\bot} \subseteq \reali[\ibase]{A}^\bot$.
\end{remark}

\begin{definition}
  A multiplicative net \emph{realises}
  -- with respect to an interpretation basis $\ibase$ --
  an hypersequent $\hseq$
  of $\mll$ formulas
  whenever it belongs to 
  $\reali[\ibase]{\hseq}$.
\end{definition}

\begin{notation}
For a hypersequent $\hseq$, we will often write $S\Vdash_\ibase \hseq$ instead of $S\in \reali[\ibase]{\hseq}$, and sometimes $S\Vdash \hseq$ or $S\in\reali{\hseq}$
when there is no ambiguity on the basis $\ibase$.
\end{notation}


From the point of view of cut elimination, 
a daimon link with $n$ outputs 
may be thought as the approximation
of a proof net with $n$ outputs.
More precisely, 
by iterating the process we have seen in \autoref{rem:nondet-proofsearch},
every cut--free proof $\pi$ of a formula $C$ 
can be obtained by applying the cut elimination 
procedure to the daimon link $\maltese_1$ 
(of conclusion $C$)
cut against the appropriate identities 
of $C,C^\bot$
(this generalises to a sequent $\Gamma$ and $\maltese_n$).
Furthermore
daimon links and proof nets (with the same number of conclusions)
are interchangeable 
with respect to geometrical correctness (\autoref{table:intro}):
in a correct (resp. incorrect) net $S$, 
substituting a daimon link with $n$ outputs 
by
a proof net with $n$ outputs 
produces a correct (resp. incorrect) net.
However, proof nets and daimons (with the same number of conclusions) 
differ 
on realisability:
for instance a proof net ending with a tensor link 
can never realise a formula of the form $A\parr B$
whereas a daimon link can (\autoref{thm:adequacy}).
We will thus say that 
a daimon link ``approximates'' a sequent: 
this suggests Definition~\ref{def:approximable}.




\begin{definition}\label{def:approximable}
  A type $\mathbf A$ is approximable
  if and only if 
  $\maltese_1 \in \mathbf A$.
  A basis $\ibase$ 
  is 
  \emph{approximable}
  if for each $X\in\pvarset$, 
  the type 
  $\reali[\ibase]{X}$ is approximable.
\end{definition}

\begin{remark}\label{rem:approxEquivalence}
  Because inclusion 
  is preserved by bi--orthogonal closure,
  a type $\mathbf A$
  is approximable 
  if and only if   
  $\{\maltese_1 \}^{\bibot}\subseteq \mathbf A$
  which is equivalent 
  to the inclusion
  $ \mathbf A^\bot\subseteq \{\maltese_1 \}^\bot$.
\end{remark}


\begin{restatable}[Adequacy]{theorem}{adequacy}\label{thm:adequacy}
  Let $\ibase$ be an approximable basis.
  For any net $S$ and sequent $\Gamma$
  $S\vdash_{\mlldai} \Gamma \Rightarrow S\Vdash_{\ibase} \Gamma$.
\end{restatable}
\begin{proof}
  The technique is standard in the works on realisability (see \cite{krivine:realiclassi} or  \cite{oliva:modrealiLL}):
  one proceeds by induction on the size of a proof $\pi$ represented by $S$.
  For the base case 
  one must show that $\maltese_n$
  realises any sequent $\Gamma$ with $n$ formulas.
  To do so one first checks that, for any formula $A$, $\reali[\ibase]{A}$ is approximable ($\maltese_1\in\reali[\ibase]{A}$).
\end{proof}

\begin{remark}
An approximable basis yields adequacy, in particular, for $\mll$. Notice, however, that there exist bases
yielding an interpretation that is adequate for $\mll$ but not for $\mlldai$.
%
%
\end{remark}

\section{Testability and tests}\label{sect:Tests}

\def\witness{\mathbin{\mid \mkern-3.8mu \simeq 
\mkern-13.8mu {\raisebox{0.55em}{\scalebox{0.55}{$\mathbf{\mathit{at}}$}}} }}
\def\witnesseq{\mathbin{\mid \mkern-3.8mu \simeq}}
\def\supp{\mathsf{s}}
\def\acc{\textsf{ACC}\;}

The partitions involved 
in the Danos Regnier criterion (\autoref{thm:danos})
and their orthogonality with the daimons of a net
can be translated as \emph{tests};
so that for a formula $A$, a net $S$ testable by $A$ (definition \ref{def:typeable} below) 
and orthogonal to
$\tests A$
is 
a correct net (\autoref{thm:danos}).
We will show that these tests are proofs of $\mlldai$ (\autoref{thm:cortest}).
This means that for realisers in an approximable basis, 
testability (\autoref{def:typeable}) 
and correct typeability (Definition \ref{def:proofrep}) coincide:
this is \autoref{prop:collapse}.

\begin{definition}[(Atomic) testable cut--free nets]\label{def:typeable}
  A \emph{formula labelling} 
  of a cut--free net $S$
  is a function $\tau: V_S\rightarrow \formulaset$ such that:
  \begin{itemize}
    \item \textbf{(Par)}
    When $\parrlink{p_1,p_2}{p}$ occurs in $S$: 
    if $\tau(p_1)=A$ and $\tau(p_2)=B$
    then $\tau(p)=A\parr B$.
  
    \item \textbf{(Tens)}
    When $\tenslink{p_1,p_2}{p}$ occurs in $S$: 
    if $\tau(p_1)=A$ and $\tau(p_2)=B$
    then $\tau(p)=A\otimes B$.

  \end{itemize}
  A formula labelling of a cut--free net $S$ is \emph{atomic}
  when 
  for each daimon link $\dailink{p_1,\dots,p_n}$
  in $S$ the formula $\tau(p_i)$ is a propositional variable.

  A cut--free net $S$ with $n$ conclusions 
  is
  \emph{testable} (resp. \emph{atomic testable})
  by a sequent 
  $\Gamma = A_1,\dots,A_n$, which we denote $S\witnesseq \Gamma$
  (resp. $S\witness \Gamma$),
  if there exists a formula (resp. an atomic formula) labelling $\tau$
  of $S$ such that $\tau(S(i)) = A_i$ 
  for each $1\leq i\leq n$.
\end{definition}


\begin{remark}
  $S\witnesseq \Gamma$ iff 
  $S\witness \Delta$ and $\Gamma=\theta\Delta$
  for some substitution $\theta$ and sequent $\Delta$.
\end{remark}

\begin{remark}
  $S\witness \Gamma$
  iff
  $S$ without its $\maltese$--links 
  is the syntactic forest of 
  (the formulas of) $\Gamma$.
\end{remark}

\begin{remark}
  A cut--free proof net $S\vdash_\mlldai \Gamma$ 
  is in particular testable by that sequent 
  i.e. 
  $S\witnesseq \Gamma$.
  However, a net  $S\witnesseq \Gamma$ 
  which is testable by $\Gamma$ 
  may not be a proof net
  because it could contain cycles or disconnections:
  the testability condition only provides information 
  on the multiplicative links 
  constituting the net $S$.
  When is $S$ atomic testable by $A$,
  orthogonality with the tests of $A$
  coincides with correctness (\autoref{prop:testsOrtho}).
\end{remark}

\begin{remark}\label{rem:mllIncorrectStyle}
  Let $S\witness A_1,\dots,A_n$ be a cut--free net.
  For any nets $T_1,\dots, T_n$ cut--free and
  atomically testable respectively by
  $ {A_1}^\bot,\dots,{A_n}^\bot$
  denoting $S_0$
  the normal form of $S::T_1\parallel \dots\parallel T_n$,
  $S_0$ is obtained by homogeneous cut--elimination,
  and we have
  (1) $S_0$ equals $\maltese_0$
  (2) $S_0$ is equal to 
  the sum of $k\geq 2$ daimon without conclusions ($S_0 = \sum_{1\leq i \leq k}\maltese_0$)
  or (3) 
  $S_0$ contains a cyclic cut
  ($S_0 = R + \dailink{\vec q, a, \vec r, b,\vec p}+ \cutlink{a,b}$).  
\end{remark}

\begin{remark}
  Given a net $S = (\body S , \arrange S)$ 
  we denote $S^\maltese = (\body {S^\maltese} , \arrange{S^\maltese})$
  the net such that
  $\body{S^\maltese}$ is the hypergraph consisting 
  of the daimon links occurring in $S$.
  The arrangement $\arrange{S^\maltese}$
  is induced by $\arrange S$
  because above every conclusion of $S$
  there is binary tree:
  each initial position $p$
  can be associated 
  with a sequence $\xi =\adrof p$ of $\{\pleft,\pright\}^*$
  and an integer $i = \rootof p$
  so that going up from $S(i)$ following the left/right instruction 
  of $\xi$
  one reaches the initial position $p$.
  The initial positions of $S$ 
  are then ordered by the lexicographical order 
  of $(\rootof p, \adrof p)$
  fixing $\pleft \leq \pright$. 
\end{remark}

\begin{notation}
  Given a net $S$ with $n$ initial positions,
  and $P =\{ C_1,\dots,C_k\}$ a partition of the initial positions of $S$
  we denote by $\nat[S] P$
  the partition 
  $\{ \arrange{S^\maltese}^{-1} (C_1),\dots,\arrange{S^\maltese}^{-1}(C_k)\}$ 
  of $\{1,\dots, n\}$.
  We might abusively write $\nat P$
  for $\nat[S]P$.
\end{notation}

\begin{restatable}{proposition}{homogeneousOrtho}\label{prop:orthoperfect}
  Let $A$ be a formula, 
  given two 
  cut free nets 
  $S\witness A$
  and $T\witness A^\bot$
  the assertions are equivalent:
  \begin{enumerate}
    \item The nets $S$ and $T$ are orthogonal.
    \item The nets $S^\maltese$ and $T^\maltese$ are orthogonal.
    \item 
    The partition 
    $\nat[S]{\daipart S}$ and $\nat[T]{\daipart T}$ are orthogonal.
  \end{enumerate}
\end{restatable}





\begin{definition}\label{def:test}
  A cut-free net $T$
  is a \emph{test} of a formula $A$ if 
  $T\witness A^\bot$
  and there exists a net $S\witness A$
  and a switching $\sigma S$ 
  such that $\nat[T]{\daipart T} = \nat[S]{\uparrow^i \sigma S}$.
  We denote by
  $\tests A$
  the 
  set $\{ S \mid S \mbox{ is a test of } A \}$.
\end{definition}

 \begin{restatable}{proposition}{orthoTest}\label{prop:testsOrtho}
  For $S$ cut--free, $S\witness A$, we have:
  $S\vdash_{\mlldai} A \Leftrightarrow S\perp \tests A.$
 \end{restatable}

 A net $S$ with $n$
 conclusion can always be transformed 
 in a net with $1$ conclusion by 
 putting a bunch of par--links below its conclusions; 
 this allows to generalise the previous proposition.


\begin{restatable}[Danos--Regnier Tests]{theorem}{danosregnier}\label{thm:danos}
  Given a cut--free net $S \witness A_1,\dots,A_n$;
  $S \vdash_{\mlldai} A_1 ,\dots, A_n$
  if and only if 
    $S$ is orthogonal to 
    $\tests {A_1}\parallel \dots \parallel \tests {A_n}$.
\end{restatable}



\begin{restatable}{theorem}{thmcortest}\label{thm:cortest}
  Any  test $T$ of a formula $A$
  is correctly typeable by $A^\bot$,
  $T\vdash_{\mlldai} A^\bot$.
\end{restatable}
 \begin{proof}
  Consider a test $T$ of $A$
  then by \autoref{thm:danos}
  any net $S\vdash_{\mlldai} A$
  is orthogonal to $T$.
  By the counter--proof criterion \cite{curien:criterions} 
  a net $N\witness A^\bot$ orthogonal to each proof of $A$
  is a proof; therefore it follows that 
  $T$ is a proof of $A^\bot$.
 \end{proof}

\begin{remark}
 \autoref{thm:danos} 
 is a refinement of the counter--proof criterion 
 of P.L. Curien \cite{curien:criterions}:
 if $S\witness A$ and $S\perp\tests{A}$ 
 then $S\vdash_{\mlldai} A$
 -- and every element 
 of 
 $\tests A$ are proofs of $A^\bot$ (\autoref{thm:cortest}),
 but the converse does not hold.   
\end{remark}
 From \autoref{thm:danos} and \autoref{thm:cortest}
one obtains an ``interactive'' criterion 
for the nets of multiplicative linear logic ($\mll$).
One takes a net of $S$ of $\mll$
(i.e. a net with binary daimons)
and confronts it with the tests of 
the according formulas (\autoref{def:test}).
A straightforward consequence of the \autoref{thm:danos}
is the reformulation of B\'echet's theorem in our framework.

\begin{corollary}\label{cor:mllIncorrectStyle}
  Let $S\witness A_1,\dots,A_n$ be a cut--free net.
  If $S$ is not correct 
  then there exists
  nets $T_1\in\tests{A_1},\dots, T_n\in\tests{A_n}$
  such that
  the normal form of $S::T_1\parallel \dots\parallel T_n$
  is not correct:
  we are in case (2)
  or (3) of \autoref{rem:mllIncorrectStyle}.
\end{corollary}


  \begin{remark}\label{rem:bechet}
  The \autoref{cor:mllIncorrectStyle}
  obviously applies to $\mll$ nets, 
  the main difference with B\'echet's original 
  result 
  is that his opponents are $\mll$ proof nets
  (in our framework they are $\mlldai$ proof nets).
  However it is not difficult to adapt our techniques to obtain B\'echet's
  result.  
\end{remark}


\begin{remark}\label{rem:testinbase}
  Consider an approximable basis $\ibase$
  and a sequent $\Gamma = A_1,\dots,A_n$
  we have
  $\reali[\ibase]{\Gamma} = 
  \smash{(\reali[\ibase]{A_1}^\bot \parallel \dots \parallel \reali[\ibase]{A_n}^\bot)^\bot}$.
  By \autoref{thm:adequacy},
  for any $A_i^\bot$ we have $\smash{\proofs{A_i^\bot : \mlldai} \subseteq \reali[\ibase]{A_i^\bot}}$ 
  while  $\tests {A_i} \subseteq \proofs {A_i^\bot : \mlldai }$ (\autoref{thm:cortest})
  thus $\tests{A_i}\subseteq \reali[\ibase]{A_i^\bot} \subseteq \reali[\ibase]{A_i}^\bot $ 
  (\autoref{rem:InterpretMLL}).
  Because the $\parallel$--construction preserves inclusions
  and orthogonality inverts inclusions
  we derive that
  $\reali[\ibase]{\Gamma} \subseteq 
  \smash{(\tests {A_1} \parallel \dots \parallel \tests{A_n})^\bot}$.
\end{remark}

Remark \ref{rem:testinbase} combined with the previous theorem (\autoref{thm:danos})
means that 
for realisers in an approximable basis, testability and (correct) typeability 
collapse.

\begin{restatable}{proposition}{propcollapse}\label{prop:collapse}
  Given $\ibase$ an approximable basis\footnote{The \autoref{prop:collapse} actually holds  for any ``adequate''
  basis $\ibase$.} 
  and a sequent $\Gamma$
  for any cut--free net $S\in\reali[\ibase]{\Gamma}$ the assertions are equivalent:
  \begin{enumerate}
    \item $S\witnesseq \Gamma$ i.e. $S\witness \Delta$ for some sequent $\Delta \leq \Gamma$.
    \item $S\vdash_{\mlldai} \Gamma$.
  \end{enumerate}
\end{restatable}

\section{Completeness}\label{sect:Completeness}

\def\baseone{\mathbf 1}

Using \autoref{prop:collapse}
we provide a completeness result;
we exhibit an approximable basis for which a net $S$ realising a sequent $\Gamma$
is testable, and so equivalently $S\vdash_{\mlldai}\Gamma$.
This basis, denoted $\baseone$,
maps each atomic formula
to $\{\maltese_1\}^{\bibot}$.

\begin{restatable}{proposition}{baseOne}\label{prop:baseOneTestable}
  For any sequent $\Gamma$
  and any cut--free net $S$;
  if $S\in\reali[\baseone]{\Gamma}$
  then $S\witnesseq \Gamma$.   
\end{restatable}

\begin{remark}
  By the \autoref{prop:baseOneTestable}
  and the \autoref{thm:adequacy}
  we have that 
  $S\in\reali[\baseone]{\Gamma}$
  iff $S\vdash_{\mlldai}\Gamma$.
\end{remark}

Since the base $\baseone$ is approximable, 
\autoref{prop:collapse}
allows to prove:

\begin{restatable}[$\mlldai$ completeness]{theorem}{mlldaiCompletion}\label{thm:mlldaiNewComplete}
  Given a cut--free net $S$
  and a sequent $\Gamma$;
  \begin{itemize}
    \item
    If for all basis $\ibase$ we have $S\in\reali[\ibase]{\Gamma}$,
    then $S\vdash_{\mlldai}\Gamma$.
    \item 
    $S\in \reali[\ibase]{\Gamma}$ for any approximable basis $\ibase$
    iff $S\vdash_{\mlldai}\Gamma$.
  \end{itemize}
\end{restatable}
\def\parbase{{\ibase}\langle\parr\rangle}
\begin{remark}\label{rem:nhomcutelimWHY}
  The non homogeneous cut elimination 
  allows to distinguish the types
  $\reali[\ibase]{X,X^\bot}$
  and $\reali[\ibase]{X,Y}$
  for a well chosen basis:
  for instance for the basis, that we will denote $\parbase$,
  which maps positive propositional variables 
  to $\{ \maltese_\parr\}^{\bot}$
  and negative propositional variables 
  to $\{ \maltese_\parr\}^{\bibot}$,
  where $\maltese_\parr$ denotes the geometrically incorrect net 
  $\dailink{a} +\dailink{b} +\parrlink{a,b}{c}$.  

  In that case,
  (1) because $\maltese_2$ is not orthogonal to 
  $\maltese_\parr\parallel \maltese_\parr$ (\autoref{fig:dai2fail})
  it follows that $\maltese_2\notin \reali[\parbase]{X,X}$
  and more generally $\maltese_2 \notin \reali[\parbase]{X,Y}$;
  (2) by the property expressed in \autoref{rem:completenessMLLatomic} (and illustrated in \autoref{fig:localDuality}), 
  $\maltese_2 \in \reali[\parbase]{X,X^\bot}$; 
  (3) point (1) above is not in contradiction with the theorem of adequacy 
  (\autoref{thm:adequacy}) because,
  even though 
  $\maltese_2\vdash_{\mlldai} X,Y$, the basis $\parbase$ is not approximable.
\end{remark}

\begin{remark}\label{rem:provabilityCorrectNeedOrthoIncor}
  The ability to distinguish 
  realisers of the 
  sequents $X,X^\bot$ and $X,Y$ (\autoref{rem:nhomcutelimWHY})
  allows us to 
  derive 
  the completeness result for $\mll$
  (\autoref{thm:completeness})
  from the completeness result for $\mlldai$ (\autoref{thm:mlldaiNewComplete}).
  In \autoref{rem:nhomcutelimWHY},
  to show that 
  $\maltese_2 \notin \reali[\parbase]{X,Y}$
  we have used incorrect nets (specifically $\maltese_\parr$),
  which explains that the completeness theorem for $\mll$ (\autoref{thm:completeness}) 
  refers to \emph{any} basis $\ibase$ (and not only to approximable basis).  
  In the terms of \autoref{table:intro},
  we retrieve provability correctness
  by using interactions with geometrically incorrect nets.
\end{remark}

\begin{restatable}[$\mll$ completeness]{theorem}{thmcomp}\label{thm:completeness}
  Let $S$ be a cut--free net
  such that each of its daimon link has exactly two outputs, 
  $\Gamma$ be a sequent such that $S\witness \Gamma$;
  if $S \in \reali[\ibase]{\Gamma}$ for any basis $\ibase$
  then, $S\vdash_{\mll}\Gamma$.
\end{restatable}

\begin{figure}
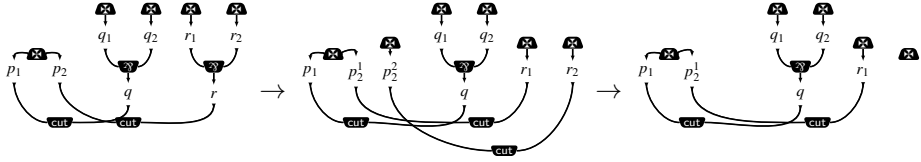

  \centering
  \scalebox{0.6}{\input{interactionb.tikz}}
  $\rightarrow$
  \scalebox{0.6}{\input{interactionb1.tikz}}
  $\rightarrow$
  \scalebox{0.6}{\input{interactionb2.tikz}}
  \caption{
    The daimon link $\maltese_2$
    is not orthogonal to $\maltese_\parr\parallel\maltese_\parr$:
    a disconnected net never reduces to a connected one 
    (and $\maltese_0$ is connected).
  }\label{fig:dai2fail}
\end{figure}
\begin{figure}
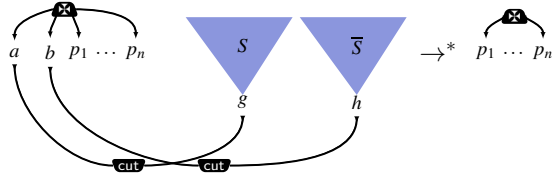

  \centering
  \scalebox{0.75}{\input{daiOrthoPairElim.tikz}}
  \raisebox{60pt}{$\rightarrow^*$}
  \scalebox{0.75}{\input{daiOrthoPairElim2.tikz}}
  \caption{The interaction of two orthogonal nets $S$ and $\overline S$
  with a daimon reduces to a daimon (with two less outputs).
  }\label{fig:localDuality}
\end{figure}

\begin{remark}\label{rem:mllAdequacy}
  A result of adequacy for $\mll$ can also be stated:
  given an interpretation basis $\ibase$ (not necessarily approximable)
  such that for each propositional variable $X$
  we have $\reali[\ibase]{X^\bot} = \reali[\ibase]{X}^\bot$,
  for any net $S$,
  if $S\vdash_{\mll} \Gamma$ then $S\in\reali[\ibase]{\Gamma}$.
\end{remark}


\begin{remark}\label{rem:completenessMLLatomic}
  The completeness result for $\mll$ (\autoref{thm:completeness})
  only identifies cut--free and \emph{atomic} proofs 
  (i.e. where axioms introduce sequents of the form $X,X^\bot$).
  This is because 
  for any atomic formulas $X$ and $Y$, 
  and for any basis $\ibase$ such that $\reali[\ibase]{X^\bot}=\reali[\ibase]{X}^\bot$,
  $\maltese_2 \in \reali[\ibase]{X\parr X^\bot , Y\parr Y^\bot }$
  while 
  $X\parr X^\bot$ and $Y\parr Y^\bot$ 
  are not dual formulas: 
  contrary to the atomic case 
  we cannot use $\maltese_2$ to distinguish $\reali[\parbase]{X\parr X^\bot , Y\parr Y^\bot}$ 
  from $\reali[\parbase]{X\parr X^\bot , X^\bot\otimes X}$.

  The fact that $\maltese_2 \in \reali[\parbase]{X\parr X^\bot , Y\parr Y^\bot }$
  (and more generally for any basis $\ibase$ such that $\reali[\ibase]{X^\bot}=\reali[\ibase]{X}^\bot$)
  is derived from the fact that,
  for any integer $k$
  and for  
  any two orthogonal nets $S_1$ and $S_2$ with one conclusion,
  the interaction $\maltese_{k+2} ::  (S_1 \parallel S_2)$ 
  has \emph{at least one} reduction to $\maltese_k$ 
  by cut elimination (\autoref{fig:localDuality}).
  We use this property for $k=2$ and $k=4$
  to show that $\maltese_2 \in \reali[\ibase]{X\parr X^\bot , Y\parr Y^\bot}$.
  More precisely, we prove that,
  $\maltese_2 \perp \reali[\ibase]{X\parr X^\bot}^\bot \parallel \reali[\ibase]{Y\parr Y^\bot}^\bot$:
  given $S,\overline S$ and $R,\overline R$
  two pairs of orthogonal nets (with one conclusion),
  when all nets $S,\overline S, R ,\overline R$ have disjoint
  sets of vertices, we can derive the following:

$$    
\begin{array}{r l r}
    & \dailink{a,b} :: S + \overline S 
    +\tenslink{S(1) , \overline S(1)}{q}  
    + 
    R + \overline R 
    +\tenslink{R(1) , \overline R(1)}{r}  \\
    \rightarrow\cdot \rightarrow &
    \dailink{a_1,a_2,b_1,b_2} :: 
    S + \overline S + 
    R + \overline R \\
    \rightarrow^* &
    \dailink{b_1,b_2} :: 
    R + \overline R \\
    \rightarrow^* &
    \maltese_0 
\end{array}
$$
  
\end{remark}



\newpage



\typeout{}
\bibliography{maktaba}{}

\newpage
\tableofcontents

\newpage

\appendix

\section{Additional Figures}

\begin{figure}[H]
    \begin{subfigure}{\textwidth}
    \centering
        \scalebox{0.48}{\input{cut_daitens1.tikz}}
        \small{$\rightarrow$}
        \scalebox{0.48}{\input{reduct_daitens1.tikz}}
        \scalebox{0.48}{\input{cut_daitens2.tikz}}
        \small{$\rightarrow$}
        \scalebox{0.48}{\input{reduct_daitens2.tikz}}
        \caption{
            Extra cases for the elimination 
        of $(\maltese / \otimes)$ cuts,
        on the left the elimination step 
        when one of the inputs belongs to the daimon above the cut,
        on the right the elimination step 
        when both inputs belong to the daimon above the cut.
        }
\end{subfigure}

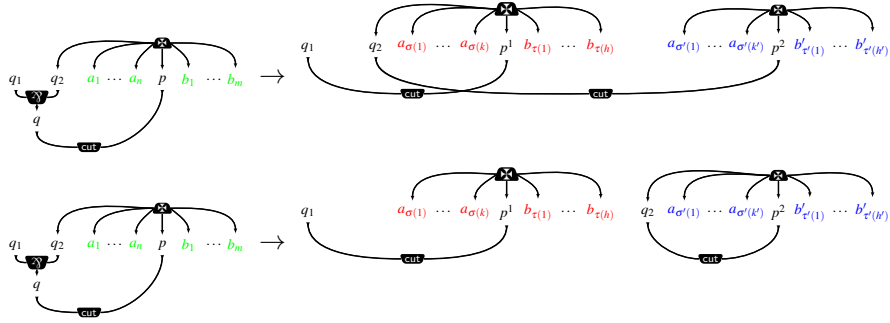
\begin{subfigure}{\textwidth}
    \centering
    \scalebox{0.55}{\input{cut_daiparr1.tikz}}
    $\rightarrow$
    \scalebox{0.55}{\input{reduct_daiparr1A.tikz}}

\scalebox{0.55}{\input{cut_daiparr1.tikz}}
$\rightarrow$
\scalebox{0.55}{\input{reduct_daiparr1B.tikz}}
    \caption{Extra cases for the elimination 
    of $(\maltese / \parr)$ cuts:
    when 
    one of the inputs belongs to the daimon above the cut.}
\end{subfigure}

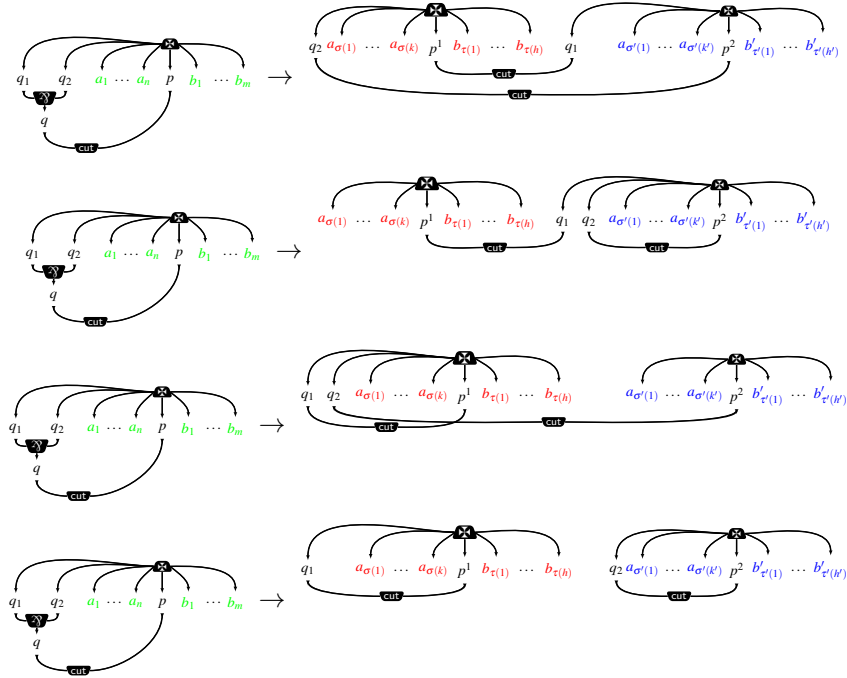
\begin{subfigure}{\textwidth}
    \centering
    \scalebox{0.55}{\input{cut_daiparr2.tikz}}
    $\rightarrow$
    \scalebox{0.55}{\input{reduct_daiparr2AA.tikz}}

\scalebox{0.55}{\input{cut_daiparr2.tikz}}
$\rightarrow$
\scalebox{0.55}{\input{reduct_daiparr2AB.tikz}}

\scalebox{0.55}{\input{cut_daiparr2.tikz}}
    $\rightarrow$
    \scalebox{0.55}{\input{reduct_daiparr2BA.tikz}}

\scalebox{0.55}{\input{cut_daiparr2.tikz}}
$\rightarrow$
\scalebox{0.55}{\input{reduct_daiparr2BB.tikz}}
    \caption{Extra cases for the elimination 
    of $(\maltese / \parr)$ cuts:
    when both 
    inputs belong to the daimon above the cut.}
\end{subfigure}
\caption{Complements to \autoref{fig:nhomcutelim}
for defining non homogeneous cut elimination (\autoref{def:nohomcut}).
}\label{fig:cutelimcomplement}
\end{figure}



\begin{figure}[H]
  \centering
  \scalebox{0.6}{\input{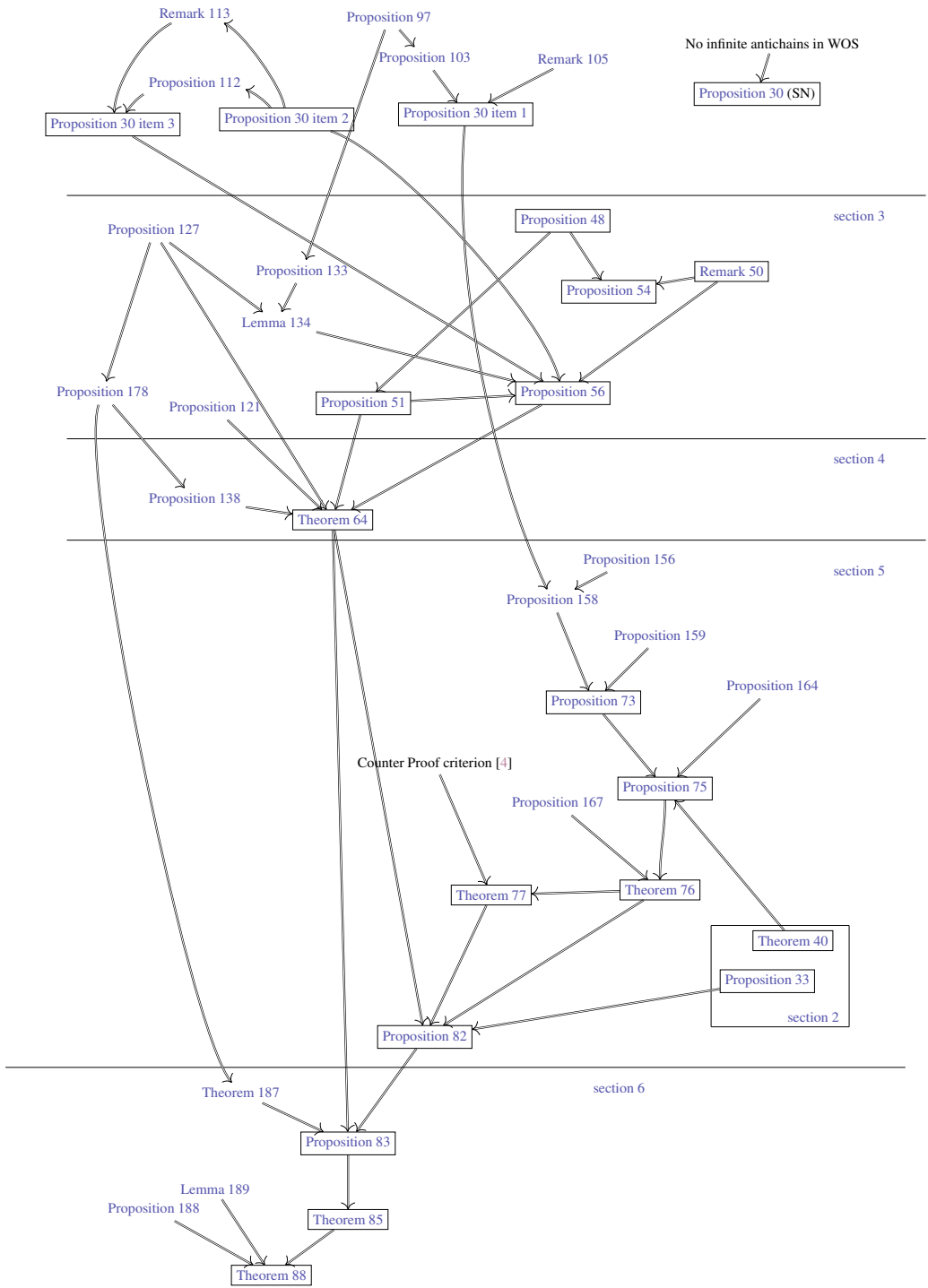}}
  \caption{    
  Relation of ``dependency'' between propositions
  in this document (one can observe that this is an acyclic graph!).
    The nodes of the graph are (some of) the propositions
    in the document, 
    while an edge from a proposition $P$ to proposition $Q$
    means that in order to prove $Q$ we have used the proposition $P$.
    Nodes inside a rectangle 
    are propositions which occur in the main text of the paper,
    while nodes without rectangle 
    are proposition occuring only in the appendix.
  }
\end{figure}

\newpage

\renewcommand*{\merge}[1][]{\bowtie^{#1}}
\def\dais#1{\mathcal D_{#1}}

\newcommand*{\formurep}{\equiv_{\mathcal F}}
\newcommand*{\formunet}[2]{\underline {#1}(#2) }
\newcommand*{\placehold}{\mathbf V }
\newcommand*{\formuat}[2]{\mathbf{F}_{#1}(#2) }
\newcommand*{\formlink}[3]{#1\langle #2 \triangleright #3 \rangle}
\newcommand*{\copos}[2][]{\mathsf{co}_{#1}({#2})}
\newcommand*{\setimg}[1]{\mathsf{set}({#1})}

\section{Complements to \autoref{sect:hypergraphs}}


We recall the notion of homomorphism and isomorphism between hypergraphs.

\begin{definition}[Homomorphism and isomorphisms of hypergraphs]\label{def:isohgraph}
  An \emph{homomorphism}
  between two (labelled) hypergraphs
  $\hgraph_1 =(V_1,E_1,\source_1,\target_1,\labl_1)$
  and $\hgraph_2 = (V_2,E_2,\source_2,\target_2,\labl_2)$
  is a pair of bijective maps $\tuple{f^V,f^E}$,
  such that
  \begin{itemize}
    \item ({\scshape commutation})
    For any edge $e_1$ of $\hgraph_1$, $\source_2(f^E(e_1)) = f^V(\source_1 (e_1))$,
    i.e. the following diagram commutes:
    \begin{center}
\begin{tikzcd}
e_1 \arrow[r, "\source_1"] \arrow[d, "f^E"'] & \source_1(e_1) \arrow[d, "f^V"] &                                                \\
e_2 \arrow[r, "\source_2"]                   & \source_2(e_2)                  &
\end{tikzcd}
\end{center}
    \item ({\scshape label-preserving})
    The function $f^E$ preserves the labels
    of the links,
    i.e. for any link $e$ in $\hgraph_1$;
    $\labl(e) = \labl (f^E(e))$
  \end{itemize}
The homomorphism is an isomorphism if both $f^V$ and $f^E$ are bijective.
\end{definition}

\subsection{Rewriting properties of nets: Strong Normalisation}

\begin{proposition}[Strong normalisation \autoref{prop:SN}]
  The cut elimination rewriting $\reduce{}$
  is strongly normalising.
\end{proposition}
\begin{proof}
  This is obtained by observing that an 
  appropriate measure on nets always decrease with cut elimination,
  i.e. whenever $S\rightarrow S'$
  then $m(S)> m(S')$.
  Such a measure $m$ may be 
  the pairs $(\mathsf{connective}(S),\mathsf{cut}(S))$
  where $\mathsf{connective}(S)$ is the number 
  of $\otimes$-- and $\parr$-- links in $S$,
  while $\mathsf{cut}(S)$ is the number of cut links in $S$.
  The order involved is then the lexicographical order.

  Indeed we have in that case that $S\rightarrow S'$
  implies $m(S)> m(S')$:
  the $(\otimes/\maltese)$, $(\parr/\maltese)$
  and multiplicative steps 
  always decrease the number $\mathsf{connective}(S)$,
  finally the glueing steps do not change $\mathsf{connective}(S)$ 
  however they make the numbe of cuts $\mathsf{cut}(S)$ decrease.
\end{proof}

\subsection{Rewriting properties of nets: unrelated cuts}

\begin{notation}
  Given two binary relations
  $R_1$ and $R_2$ on a set $X$
  we denote by $R_1\at R_2$
  the composition of the two relation,
  i.e. for two $x,y\in X$
  $x R_1\at  R_2 y$
  if and only if there exists $z$
  such that $x R_1 z$
  and $z R_2 y$.
\end{notation}

\begin{notation}
  Given a net $S$ and cut link $c$
  we denote by $S\reduce c S'$,
  whenever $S\rightarrow S'$
  by eliminating the cut $c$. 
  Note that whenever $S\rightarrow S'$
  there exists a cut $c$
  such that $S\rightarrow c S'$,
  namely the unique cut which belong to $S$
  but not $S'$.
\end{notation}

\begin{notation}
  We denote by $\redmult$
  the relation between nets such 
  that $S\redmult S'$
  whenever $S\reduce c S'$
  and $c$ is multiplicative,
  on the other hand we denote 
  $S\rednotmult S'$
  whenever 
  $S\reduce c S'$
  and $c$ is not a multiplicative cut.
  Similarly one defined 
  $\redglue$, $\rednotglue$,
  $\redparrdai$ , $\rednotparrdai$,
   $\redtensdai$ and $\rednottensdai$.
\end{notation}

\begin{definition}
  Given a net $S$
  two cuts $c_1$ and $c_2$
  in $S$ are \emph{unrelated}
  whenever:
  \begin{itemize}
    \item $S$ can be written as $C + R_1 + R_2$
    where $R_1$ is the redex of the cut $c_1$
    and $R_2$ is the redex of the cut $c_2$.
    \item By eliminating $c_1$, $S$ reduces to $C + R_1' + R_2$.
    \item By eliminating $c_2$, $S$ reduces to $C + R_1 + R_2'$.
  \end{itemize}
\end{definition}

\begin{proposition}[Strong Confluence of unrelated cuts]\label{prop:unrelatedConflu}
  Given $S$ some net 
  containing two reductible and unrelated cuts $c_1$ and $c_2$.
  The diagram below commutes:
  \begin{center}
\begin{tikzcd}
S \arrow[r, "c_1"] \arrow[d ,"c_2"] & S_1 \arrow[d, "c_1",dotted] \\
S_1' \arrow[r ,"c_2",dotted]             & S_2                
\end{tikzcd}
\end{center}
Where the dotted arrows are the existence of a reduction.
\end{proposition}
\begin{proof}
  Let us fix a net $S$ 
  containing (at least) two unrelated cuts 
  that we denote $c_1$ and $c_2$.
  This means that 
  one can write 
  $S$ has $C + R_1 + R_2$
  where the $R_i$ are the corresponding redexes.
  Assuming (without loss of generality) 
  that $R_i\rightarrow R_i'$:
  by contextual closure, the elimination 
  of $c_1$ rewrites $S$ 
  in $S_1 = C + R_1' + R_2$
  and the elimination of $c_2$
  rewrites $S$ in $S_{2} = C + R_1 + R_2'$

  Using contextual closure once more we conclude that
  $S_1$ reduces to $C + R_1' + R_2'$ 
  by eliminating $c_2$
  while $S_2$ reduces to $C + R_1' + R_2'$  
  by eliminating $c_1$.
\end{proof}

\begin{remark}\label{rem:stableUnrelatedCuts}
  Consider a cut $c$  which is not multiplicative in a net $S$
  if $c$ is not related with a cut $c'$ in $S$
  then if $S\reduce c S'$
  the cut produced by $c$ 
  are also unrelated with $c'$.
\end{remark}

\begin{definition}[Strongly unrelated cuts]
We say that two cuts 
$c_1 =\cutlink{p_1,q_1}$
and $c_2 = \cutlink{p_2,q_2}$
are \emph{strongly unrelated} 
whenever they are distinct 
and the set $D_1$ of daimon links 
above $p_1$ or $q_1$
does not intersect with 
the set $D_2$ of daimon links 
above $p_2$ or $q_2$.

\end{definition}

\begin{remark}\label{rem:strongUnrelateMult}
The notion of strongly unrelated 
cuts 
is to obtain the stability of unrelated cuts (as in \autoref{rem:stableUnrelatedCuts})
also for multiplicative cuts.
Given a multiplicative cut $c$ of a net $S$
and a cut $c'$ of $S$ unrelated with $c$,
if $S\reduce c S'$
then the cut produced by the elimination of $c$
remain unrelated with $c$ in $S'$. 
\end{remark}

\begin{proposition}\label{prop:unrelatedFactorization}
  Given a net $S$
  and a partition of its set of cut links 
  $\mathsf{cut}(S) = C_1\uplus C_2$
  such that each cuts of $C_1$ and $C_2$
  are strongly unrelated
  then the assertions are equivalent:
  \begin{enumerate}
    \item 
    $S\rightarrow^* T$
    \item 
    $S\rightarrow^*_1 S' \rightarrow^*_2 T$,
    where $\rightarrow_1$ eliminates cut from 
    $C_1$ or the cut it produces 
    and $\rightarrow_2$ eliminates cuts from $C_2$
    or the cut it produces.
  \end{enumerate}
\end{proposition}
\begin{proof}
  $1\Rightarrow 2.$
  If two cut $c$ and $d$ are unrelated: 
  if $c$ is not multiplicative, any cut created by the elimination of $c$ 
  will still be unrelated with $d$ (\autoref{rem:stableUnrelatedCuts}).
  If $c$ is multiplicative since $C_1$ and $C_2$ are strongly unrelated 
  and thus $c$ and $d$ are strongly unrelated cuts,
  the non multiplicative cuts say $c_1$ and $c_2$ created by the elimination of $c$
  involve a different daimon link than that of $d$
  thus $c_1$ and $c_2$ are still unrelated with $d$ (\autoref{rem:strongUnrelateMult}).
  Therefore, it follows that for each step of cut elimination 
  $\reduce{c}_1$ 
  and $\reduce{c'}_2$ 
  the cuts $c$ and $c'$ are unrelated 
  thus using the property of confluence of unrelated cuts 
  (\autoref{prop:unrelatedConflu})
  one concludes.

  $2\Rightarrow 1.$
  indeed because $\rightarrow_1$
  and $\rightarrow_2$
  are special case of the cut elimination $\rightarrow$.
\end{proof}

\subsection{Multiplicative Factorization of Cut Elimination 
-- Proof of \autoref{prop:SN} \autoref{item:factorisation}}

\begin{remark}\label{rem:multiUnrelation}
Given a cut $c$ in a net $S$
if $c$ is a multiplicative cut 
then for any other cut $c'$ of $S$:
$c$ and $c'$ are unrelated,
because the redexes of cuts only involve three links
made of the cut link and the two links above the input of that cut.
Therefore two cuts of two distinct multiplicative redexes cannot involve a same connective link
since otherwise it means 
that these two cuts 
have for input a same position 
i.e. the target of the connective link
(this violates the property of \emph{targget--disjointness} of nets).
\end{remark}

\begin{proposition}[General commutation of the Multiplicative cuts]\label{prop:multiCommutes}
For any net $S$
containing a multiplicative cut $c$,
for any sequence of cuts $\alpha^*$ in $S$
then the following diagrams commutes;
\begin{center}
  \begin{tikzcd}
  S \arrow[r, "c"] \arrow[d, dashed ,"\alpha^*"] & S_1 \arrow[d, "\alpha^*"] \\
  S_1' \arrow[r, dashed ,"c"]             & S_2                
  \end{tikzcd}
  \end{center}
    Where the dotted arrows are the existence of a reduction.

\end{proposition}
\begin{proof}
Because the multiplicative cut $c$
is unrelated with any cut $d$ of $S$ (\autoref{rem:multiUnrelation})
we conclude using the confluence of unrelated cuts (\autoref{prop:unrelatedConflu}).
\end{proof}

\begin{remark}\label{rem:multRedexPreserved}
If $S$ reduces to $S'$
by eliminating a non multiplicative cut 
then any multiplicative cut $c$ occuring in $S'$ 
is a multiplicative cut occuring in $S$.
\end{remark}

\begin{remark}\label{rem:multCreatesMult}
  Given a net $S$ containing a cut $c$
  and one of its reduction $S\rightarrow_{c} S'$,
  if a multiplicative or clash cut $c'$ occurs in $S'$ and not $S$:
  then $c$ was a multiplicative cut.

  This is because (1) glueing cuts 
  make cuts disappear 
  (2) $(\otimes/\maltese)$ and $(\parr/\maltese)$ cuts 
  create 
  two new cuts which can be of types $(\otimes/\maltese)$,
  $(\parr/\maltese)$ or $(\maltese/\maltese)$
  but these new cuts are never multiplicative cuts.
\end{remark}

\begin{proposition}[Cut Elimination Factorisation]\label{prop:factorisation}
The relation induced by the cut--elimination rewriting 
$\rightarrow^*$ of \autoref{def:nohomcut} and \autoref{def:homcut}
is equal to 
$\rightarrow_{\mathsf{mult}}^* \at \rightarrow_{\neg {\mathsf{mult}}}^*$.
\end{proposition}
\begin{proof}
The inclusion
of $\rightarrow_{\mathsf{mult}}^* \at \rightarrow_{\neg {\mathsf{mult}}}^*$
in $\rightarrow^*$ is obvious.

Let us now show the converse inclusion.
We do so by induction on the length of the reduction sequence.
If there is only one reduction this is trivial.

Let us treat the general case when $S\rightarrow^* S'$
doing a sequence of $n>1$ steps.
Then this can be decomposed 
as $S\rightarrow^* S_0 \reduce {c} S'$;
apply the induction hypothesis 
yielding that there exists a net $S_1$ such that
$S\redmult^* S_1 \rednotmult^* S_0$.
Two case may occur depending on the type of the cut $c$.
\begin{itemize}
  \item 
  If $c$ is not multiplicative;
  then $S_0\rednotmult S'$
  and from $S\redmult^* \at \rednotmult^* S_0$.
  we derive 
   $S\redmult^* \at \rednotmult^* S'$.
  \item 
  If $c$ is multiplicative;
  note that whenever 
  $T\rednotmult^* T'$
  a multiplicative cut $c$ occuring in $T'$
  will aslo occur in $T$ and be also a multiplicative cut (\autoref{rem:multRedexPreserved}).
  To conclude however it will be necessary to do an induction on the reduction sequence 
  $ S_1 \rednotmult^* S_0$:
  \begin{itemize}
    \item 
    If it contains only one step of reduction;
    then the situation is the following 
    $S_1 \reduce {c'} S_0 \reduce c S'$
    where $c'$ is not multiplicative 
    and $c$ is a multiplicative cut.
    Applying \autoref{rem:multRedexPreserved}
    non multiplicative cuts cannot create multiplicative cuts, 
    therefore the cut $c$ also appears in $S_1$
    and thus both cuts $c$ and $c'$ occur in $S_1$.
    As a consequence since multiplicative 
    commute (\autoref{prop:multiCommutes})
    it follows that 
    $S_1 \reduce {c} S_0' \reduce {c'} S'$.
    Therefore $S_1\redmult ^* \rednotmult^* S'$.

    \item
    If the sequence of non multiplicative cut--elimination 
    contains $n$ element.
    Note that it may be factorized as 
    $S_1 \rednotmult^* S_2 \rednotmult S_0$
    while $S_0\redmult S'$.
    Applying the previous result (the case when the $\rednotmult$ sequence is of length $1$)
    we derive that 
    $S_1\rednotmult^* S_2 \redmult S_0' \rednotmult S'$.
    Since the sequence $S_1\rednotmult^* S_2$
    has a smaller length than that of 
    $S_1 \rednotmult^* S_2 \rednotmult S_0$
    we call the induction hypothesis 
    and ensure 
    $S_1\redmult^* \at \rednotmult^* S_0'$
    Finally since $S_0'\rednotmult S'$
    we derive 
    $S_1\redmult^* \at \rednotmult^* S'$.
  \end{itemize}

  Now we can conclude;
  the situtation was 
  $S\redmult^* S_1 \rednotmult^* S_0\redmult S'$
  But we have shown that 
  $S_1 \rednotmult^* S_0\redmult S'$
  yields 
  $S_1 \redmult^* \at \rednotmult^* S'$.
  This yields 
  $S\redmult^* S_1 \redmult^* \at \rednotmult^* S'$
  which means 
  $S\redmult^* \at \rednotmult^* S'$.

\end{itemize}

\end{proof}

\begin{remark}
The previous proposition 
is a novelty of the cut elimination with generalised axiom we have introduced.
The standard cut--elimination procedure of $\mll$ proof structure (\cite{girard_1987})
does not enjoy such a factorisation:
mainly because the \autoref{rem:multCreatesMult}
fails, 
a non multiplicative cut (that is in that case an $\mathsf{(ax/cut)}$ cut)
can be eliminated creating a single cut 
that may be a multiplicative cut.
\end{remark}

\begin{remark}
If $S\rightarrow_{\mathsf{mult}}^* S_0$
i.e. $S$ reduces to $S_0$
by eliminating multiplicative cuts only,
then $S$ and $S_0$ have the same orthogonal.
\end{remark}

\subsection{Delaying Irreversible cuts elimination -- proof of \autoref{prop:SN} \autoref{item:delayParrs}}

\begin{proposition}[Non--Deterministic non homogeneous cut elimination commutes to the right]\label{prop:rightcom}
Given $S$ some net 
containing a non homogeneous cut link $c$
of a daimon link against a $\parr$--link.
The diagram below commutes,
for any cut kind of cut link $c'$;
\begin{center}
\begin{tikzcd}
S \arrow[r, "c"] \arrow[d, dashed ,"c'"] & S_1 \arrow[d, "c'"] \\
S_1' \arrow[r, dashed ,"c"]             & S_2                
\end{tikzcd}
\end{center}
Where the dotted arrows are the existence of a reduction.
\end{proposition}
\begin{proof}
If $c'$ is a multiplicative cut
the commutation holds (\autoref{prop:multiCommutes}).
And more generally if $c'$ and $c$ are unrelated then 
the commutation also hold (\autoref{prop:unrelatedConflu}).
Thus $c'$ is a glueing or non homogeneous cut 
involving the same daimon link as $c$.

Let us treat each cases
depending on the type of the cut $c'$;
\begin{itemize}
  \item
  Assuming that $c'$ is a glueing cut,
  and assuming (without loss of generality) 
  that the two first conclusions 
  of the daimon link are the premisses of the cut links $c$ and $c'$,
  the reduction will be of the form:
  \begin{center}
    \scalebox{0.9}{\begin{tabular}{l l}
    &  $\dailink{p_1,\dots,p_n} + \parrlink{r_1,r_2}{r} 
      +\dailink{q_1,\dots,q_k}
    + \cutlink{p_1,r} + \cutlink{p_2,q_1}$ \\
     $\rightarrow$ & 
    $\dailink{p_1^1,p_2,A} + \dailink{p_1^2,B} 
    +\dailink{q_1,\dots,q_k}
  + \cutlink{p_1^1,r_2} + \cutlink{p_1^2,r_2} + \cutlink{p_2,q_1}$ \\
    $\rightarrow$ &
    $\dailink{p_1^1,q_2,\dots,q_k,A} + \dailink{p_1^2,B} 
    + \cutlink{p_1^1,r_2} + \cutlink{p_1^2,r_2} $
    \end{tabular}}
  \end{center}
  By consistently 
  chosing the partition during the elimination 
  of the cut $c$ 
  this can be matched by first starting with the elimination 
  of the glueing cut $c'$.
  \begin{center}
    \scalebox{0.9}{\begin{tabular}{l l}
    &  $\dailink{p_1,\dots,p_n} + \parrlink{r_1,r_2}{r} 
      +\dailink{q_1,\dots,q_k}
    + \cutlink{p_1,r} + \cutlink{p_2,q_1}$ \\
     $\rightarrow$ & 
     $\dailink{p_1, q_2,\dots,q_k,\dots,p_n} + \parrlink{r_1,r_2}{r} 
   + \cutlink{p_1,r} $\\
    $\rightarrow$ &
    $\dailink{p_1^1,q_2,\dots,q_k,A} + \dailink{p_1^2,B} 
   + \cutlink{p_1^1,r_1}+ \cutlink{p_1^2,r_2} $
    \end{tabular}}
  \end{center}
  \item In the first non homogeneous case the redex is of the following form 
  $$\dailink{p_1,\dots,p_n} + \parrlink{r_1,r_2}{r} +\tenslink{q_1,q_2}{q}
  + \cutlink{p_1,r} + \cutlink{p_2,q}.$$
  After one step of cut elimination this becomes 
  in all generality, given that $A= \{ a_1,\dots,a_k\}$
  and $B= \{b_1,\dots,b_l\}$ partition $p_2,\dots,p_n$,
  $$\dailink{p_1^1,p_2, A'} + \dailink{p_1^2,B} + 
  +\tenslink{q_1,q_2}{q}
  + \cutlink{p_1^1,r_1}  + \cutlink{p_1^2,r_2} + \cutlink{p_2,q}.$$
  Without loss of generality assume that 
  $p_2$ occurs in the class $A = \{p_2\}\cup A'$,
  then after one step of cut elimination this becomes;
  $$\dailink{p_1^1,p_2^1,p_2^2,A'} + \dailink{p_1^2,B} + 
  + \cutlink{p_1^1,r_1}  + \cutlink{p_1^2,r_2}
  + \cutlink{p_2^1,q_1}  + \cutlink{p_2^2,q_2}.$$

  Indeed one can obtain the same redex 
  by first eliminating $c'$,
  when eliminating the $\parr$--link
  we need to make a consistent choice 
  e.g. 
  the partition of $p_2^1,p_2^2,\dots,p_n$
  made of the two classes $B$
  and $A'\cup\{ p_2^1,p_2^2\}$:
  \begin{center}
  \scalebox{0.9}{\begin{tabular}{l l}
  &  $\dailink{p_1,\dots,p_n} + \parrlink{r_1,r_2}{r} +\tenslink{q_1,q_2}{q}
  + \cutlink{p_1,r} + \cutlink{p_2,q}$ \\
   $\rightarrow$ & 
  $\dailink{p_1,p_2^1,p_2^2\dots,p_n} + \parrlink{r_1,r_2}{r} 
  + \cutlink{p_1,r} + \cutlink{p_2^1,q_1}
  + \cutlink{p_2^2,q_2}$ \\
  $\rightarrow$ &
  $\dailink{p_1,p_2^1,p_2^2,\dots,p_n} + \parrlink{r_1,r_2}{r} 
  + \cutlink{p_1,r} + \cutlink{p_2^1,q_1}
  + \cutlink{p_2^2,q_2}$ \\
  $\rightarrow $ &
  $\dailink{p_1^1,p_2^1,p_2^2,A'} + \dailink{p_1^2,B} + 
  + \cutlink{p_1^1,r_1}  + \cutlink{p_1^2,r_2}
  + \cutlink{p_2^1,q_1}  + \cutlink{p_2^2,q_2}.$
  \end{tabular}}
\end{center}
  \item In the non homogeneous second case 
  both cuts are made of $\parr$ links 
  against a daimon thus 
  in all generality a reduction is of the 
  following form form;
  \begin{center}
  \scalebox{0.9}{\begin{tabular}{l l}
  & $\dailink{p_1,\dots,p_n} + 
    \parrlink{r_1,r_2}{r} +\parrlink{q_1,q_2}{q}
  + \cutlink{p_1,r} + \cutlink{p_2,q}$ \\
   $\rightarrow $&
  $\dailink{p_1^1, p_2, A} + \dailink{p_1^2,B} 
  + \parrlink{q_1,q_2}{q} 
  + \cutlink{p_2,q} 
  + \cutlink{p_1^1,r_1} + \cutlink{p_1^2,r_2}$ \\
   $\rightarrow $&
  $\dailink{p_1^1, p_2^1, A_1} + \dailink{p_2^2, A_2^2} + \dailink{p_1^2,B} 
  + \cutlink{p_2^1,q_1} + \cutlink{p_2^2,q_2}
  + \cutlink{p_1^1,r_1} + \cutlink{p_1^2,r_2}$
  \end{tabular}}
\end{center}
  Indeed starting by eliminating the cut $c'$
  we can obtain the same redex
  by making consistent choices in the partitions:
  \begin{center}
  \scalebox{0.9}{\begin{tabular}{l l}
   & $\dailink{p_1,\dots,p_n} + 
    \parrlink{r_1,r_2}{r} +\parrlink{q_1,q_2}{q}
  + \cutlink{p_1,r} + \cutlink{p_2,q}$ \\
   $\rightarrow$ &
  $\dailink{p_2^1, p_1, A_1 \cup B} + \dailink{p_2^2,A_2} 
  + \parrlink{r_1,r_2}{r} + \cutlink{p_1,r} 
  + \cutlink{p_2^1,q_1} + \cutlink{p_2^2,q_2}$ \\
   $\rightarrow$ &
  $\dailink{p_1^1, p_2^1, A_1} + \dailink{p_2^2, A_2} + \dailink{p_1^2,B} 
  + \cutlink{p_2^1,q_1} + \cutlink{p_2^2,q_2}
  + \cutlink{p_1^1,r_1} + \cutlink{p_1^2,r_2}$
  \end{tabular}}
\end{center}
\end{itemize}
\end{proof}

\begin{proposition}[General right commutation, \autoref{prop:SN}\autoref{item:delayParrs}]\label{prop:delayingThePar}
For any net $S$
containing an irreversible cut $c$,
and given $\alpha^*$
a series of cut elimination 
of cuts occuring in $S$
then the following diagram commutes;
\begin{center}
  \begin{tikzcd}
  S \arrow[r, "c"] \arrow[d, dashed ,"\alpha^*"] & S_1 \arrow[d, "\alpha^*"] \\
  S_1' \arrow[r, dashed ,"c"]             & S_2                
  \end{tikzcd}
  \end{center}
\end{proposition}
\begin{proof}
This follows from a simple induction
on $\alpha^*$.
If $\alpha^*$ is made of only one cut 
the proposition \ref{prop:rightcom}
gets us the conclusion.
Otherwise we decompose the sequence of reductions 
as $\alpha^* = \beta\at\beta^*$ and 
by applying the proposition \ref{prop:rightcom} we obtain 
the following diagram.
\begin{center}
  \begin{tikzcd}
  S \arrow[r, "c"] \arrow[d, dashed ,"\beta"] & S_1 \arrow[d, "\beta"]  \\
  S_1' \arrow[r, dashed ,"c"]           & S_2   \arrow[d, "\beta^*"]  \\
  {}  & S_3
  \end{tikzcd}
\end{center}
Calling the induction hypothesis on $\beta^*$
we can complete the diagram as follows
and conclude.
\begin{center}
  \begin{tikzcd}
  S \arrow[r, "c"] \arrow[d, dashed ,"\beta"] & S_1 \arrow[d, "\beta"]  \\
  S_1' \arrow[r, dashed ,"c"] \arrow[d,  dashed, red ,"\beta^*"]  & S_2  \arrow[d, "\beta^*"] \\ 
  S_2' \arrow[r, dashed, red, "c"]  & S_3           
  \end{tikzcd}
\end{center}
\end{proof}

\subsection{Reversible cuts can be anticipated -- Proof of \autoref{prop:SN} \autoref{item:earlyNonPar}}
\begin{proposition}[Deterministic non homogeneous cut elimination commutes to the left]\label{prop:leftcom}
Given $S$ some net 
containing a non homogeneous cut link $c$
of a daimon link against a $\otimes$--link.
The diagram below commutes,
for any cut kind of cut link $c'$;
\begin{center}
\begin{tikzcd}
S \arrow[r, "c'"] \arrow[d, dashed ,"c"] & S_1 \arrow[d, "c"] \\
S_1' \arrow[r, dashed ,"c'"]             & S_2                
\end{tikzcd}
\end{center}
Where the dotted arrows are the existence of a reduction.
\end{proposition}
\begin{proof}
We reason as in the proof for proposition \ref{prop:rightcom},
if the two cuts involved different link of the net we 
conclude.
Furthermore if $c'$ is an homogeneous cut the proposition hold.
Thus we assume that $c$ and $c'$ 
are both non--homogeneous and involve the same daimon link.
\begin{itemize}
  \item If $c'$ is also a tensor link 
  its clear that the elimination of the two cuts commute.
  \item If $c'$ is a $\parr$ link 
  we can call the previous proposition \ref{prop:rightcom}
  on the cut $c'$, 
  claiming that its elimination commutes on the right 
  with any step of cut elimination.
  In particular any elimination of $c'$ followed by any elimination of $c$
  can be matched by an elimination of $c$ followed by an elimination of $c'$.
  Thus we conclude.
\end{itemize}
\end{proof}

\begin{proposition}[General left commutation]\label{prop:anticipateTensor}
For any net $S$
containing an reversible cut $c$,
and given $\alpha^*$
a series of cut elimination 
of cuts occuring in $S$
then the following diagram commutes;
\begin{center}
  \begin{tikzcd}
  S \arrow[r, "\alpha^*"] \arrow[d, dashed ,"c"] & S_1 \arrow[d, "c"] \\
  S_1' \arrow[r, dashed ,"\alpha^*"]             & S_2                
  \end{tikzcd}
  \end{center}
\end{proposition}
\begin{proof}
As for the previous proposition 
it follows from a simple induction
on $\alpha^*$.
If $\alpha^*$ is made of only one cut 
the proposition \ref{prop:leftcom}
gets us to conclude.
Otherwise we decompose the sequence of reductions 
as $\alpha^* = \beta^*\at\beta$ and 
by applying the proposition \ref{prop:leftcom} 
together with the induction hypothesis 
we obtain the following diagram.
\begin{center}
  \begin{tikzcd}
  S \arrow[d, red, dashed, "c"] \arrow[r ,"\beta^*"] 
  & S_1 \arrow[r, "\beta"] \arrow[d, dashed, "c"] 
  & S_2 \arrow[d, "c"]\\
  S_1' \arrow[r , dashed ,red,"\beta^*"]  
  & S_2  \arrow[r, dashed,"\beta"]  
  & S_3
  \end{tikzcd}
\end{center}
\end{proof}

\begin{remark}\label{rem:anticipateGlue}
Moreover, following a similar method, we can easily establish 
that glueing cuts commute to the left 
with any kind of cut.
The case with a $(\parr/ \maltese)$ cut 
is taken care of 
since we know these cuts can be delayed (\autoref{prop:SN} \autoref{item:delayParrs}).
For the other cases note that two related glueing cuts indeed commute,
furthermore by some simple calculations,
one can prove that a glueing cut and $(\otimes/\maltese)$ cut 
will also commute.
Indeed the result of commutation of two cuts may then be generalised to sequences 
of cut elimination as in \autoref{prop:anticipateTensor}.
\end{remark}

\begin{proposition}[\autoref{prop:SN} \autoref{item:earlyNonPar}]
Given a net $S$ containing a cut $c$ 
that is not a $(\parr/\maltese)$ cut.
If $S\rightarrow^*\reduce{c} S'$
then $S\reduce{c}\rightarrow^* S'$.
\end{proposition}
\begin{proof}
This is a consequence 
of \autoref{prop:anticipateTensor}
and \autoref{rem:anticipateGlue}.
\end{proof}

\section{Complements to \autoref{sect:MLL}}

\subsection{Multigraph}

The notion of multigraph 
is used several times:
(1) in order to define orthogonal partitions 
one must compute a multigraph
and 
(2) in all generality 
the graph underlying an hypergraph 
should also be a multigraph 
-- this will be defined in the next section of this appendix (\autoref{app:underGraph}).

\begin{definition}
  An \emph{(undirected) multigraph} is a tuple 
  $G = (V_G,E_G,\border[G])$
  such that:
  \begin{itemize}
    \item $V_G$ and $E_G$
    are two disjoint sets, 
    whose elements are respectively called \emph{nodes}
    and \emph{edges}.
    \item 
    The border function $\border[G]:E_G \rightarrow \mathcal P{(V_G)}$
    maps edges to sets of nodes of cardinal $1$ or $2$.
    The intuition is that 
    whenever $\border[G](e) = \{x ,y\}$
    then $x$ and $y$ are the two endpoints of the edge $e$,
    while if $\border[G](e) = \{x\}$
    the edge $e$ forms a \emph{loop}.
  \end{itemize}
\end{definition}


\subsection{Graph underlying an hypergraph}\label{app:underGraph}

In \autoref{sect:MLL}
  we informally refer to the notion of graph
  underlying an hypergraph.
  This induced graph can be formally defined 
  as in the following \autoref{def:underGraph}

  \begin{definition}[Underlying graph]\label{def:underGraph}
    Given an hypergraph $\hgraph = (V_\hgraph,E_\hgraph,\target_\hgraph,\source_\hgraph)$ 
    its \emph{underlying directed graph}
    is
    a directed graph denoted 
    $\undergraph \hgraph = (V_G , E_G, \target_G , \source_G)$
    obtained from $\hgraph$
    as follow;
    \begin{itemize}
      \item A vertex of $\undergraph \hgraph$
      is either a position of the hypergraph $\hgraph$,
      or a link of $\hgraph$,
      i.e. $V_G = V_\hgraph \uplus E_\hgraph$.
      \item 
      An edge of $\undergraph \hgraph$
      is a tuple $(e,p,\mathsf{in})$
      of a link $e\in E_\hgraph$
      with $p\in\source(e)$
      or a tuple $(e,p,\mathsf{out})$
      with $p\in\target(e)$.
      \item 
      For an edge 
      $(e,p,\mathsf{in})$
      we have 
      $\target(e,p,\mathsf{in}) = e$
      and $\source (e,p,\mathsf{in}) = p$.
      For an edge 
      $(e,p,\mathsf{out})$
      we have 
      $\target(e,p,\mathsf{out}) = p$
      and $\source (e,p,\mathsf{out}) = e$.
    \end{itemize}
  \end{definition}


  \begin{definition}[Induced undirected graph]
    Given a directed graph $G=(V_G,E_G,\target_G,\source_G)$
    its induced undirected graph
    is the graph 
    $G^* = (V_G^*,E_G^*,\border[G^*])$
    such that:
    \begin{itemize}
      \item 
      The set of vertices $V_G$ and $V_G^*$
      are the same;
      \item 
      The set of edges $E_G$ and $E_G^*$
      are the same;
      \item 
      For any edge $e$  
      its border in $G^*$
      is given by 
      $\border[G^*](e) = \{\source_G(e),\target_G(e)\}$.
    \end{itemize}
  \end{definition}

  \begin{remark}
    The paths of 
    the induced undirected graph $G^*$
    of a directed (multi)graph
    $G=(V_G,E_G,\target_G,\source_G)$ 
    are exactly the paths in $G$ 
    where we allow edges to be travelled 
    through in both direction 
    (source--target or target--source).
  \end{remark}  

  \begin{definition}
    Given an hypergraph $\hgraph$
    its \emph{underlying graph} $\undergraph\hgraph ^*$
    is the induced undirected graph of $\undergraph \hgraph$.
  \end{definition}

  The previous definition is the one involved in 
  in the theorem of Danos Regnier (\autoref{thm:standardcrit}):
  given a switching $\sigma S$ of a net $S$
  the acyclicity and connectedness of $\sigma  S$ 
  is that of its underlying undirected graph $\undergraph{\sigma S} ^*$.

\subsection{On substitutions -- proof of \autoref{prop:smallproof}}

\propSubstiProof*
\begin{proof}
  This is because 
  if a sequent $\Delta$ can be introduced 
  by a daimon rule (resp. an axiom rule for $\mll$)
  then for any substitution $\theta$
  the sequent $\theta \Delta$
  can be introduced by a daimon rule (resp. an axiom rule for $\mll$).
  Thus a proof by induction on the represented proof tree
  allows us to conclude.
\end{proof}

\begin{remark}\label{rem:subproof}
  The \autoref{prop:smallproof} implies that
  whenever $A\leq B$
  the set 
  $\proofs{A:\mlldai}$
  is contained in $\proofs{B:\mlldai}$
  but also that 
  $\proofs{A:\mll}$
  is contained in $\proofs{B:\mll}$
\end{remark}

  \section{Complements to \autoref{sect:RealModel}}

\subsection{On the notion of Interaction}
  \subsubsection{Proof of \autoref{prop:interAsAction}}

  We have presented 
  a notion of interaction $S::T$
  between two net $S$ and $T$ in \autoref{def:interaction},
  this definition makes direct use of the order on the conclusions 
  of the nets (the arrangement of \autoref{sect:hypergraphs}).
  An alternative can be to define 
  the interaction between two nets $S$ and $T$
  using a notion of \emph{interface} $\sigma$,
  a functional and injective relation between the conclusions 
  of $S$ and $T$,
  then $S::_{\sigma} T$
  is equal to 
  $S + T + \sum_{(p,q)\in\sigma} \cutlink{S(p),T(p)}$.
  This makes orthogonality more subtle to define 
  so we chose to work with ordered conclusions
  which determines the interaction $S::T$,
  in some sense the order on the conclusions 
  is used to define a canonical interface.
  
  \interAction*
  \begin{proof}
    Consider three net $S,T$ and $R$
    such that $\#S \geq \#T+\#R$.
    Note that $\# T\parallel \# R = \#T+\#R$
    thus $\#S\geq \# T\parallel \# R$
    and so 
    $min(\#S, \# T\parallel \# R) = \# T\parallel \# R$.
    As a consequence $S::(T\parallel R)$
    is equal to 
    $$
    S::(T\parallel R) = S + T + R + \sum_{1\leq i \leq \# T\parallel \# R} \cutlink{S(i),(T\parallel R)(i)}
    $$

    Note that 
    for any $1\leq i \leq \#T$
    we have $(T\parallel R)(i) = T(i)$
    thus:
    $$
    S::(T\parallel R) = 
    \left( S +T + \sum_{1\leq i\leq \#T} \cutlink{S(i),T(i)} \right)+ R + 
    \sum_{\# T +1 \leq i \leq \# T\parallel R} \cutlink{S(i),(T\parallel R)(i)}
    $$

    Furthermore 
    note that for each $\# T +1 \leq i \leq \# S$
    we have that 
    the conclusion $S(i)$
    is the conclusion $(S::T)(i-\#T)$ of $S$.
    Also for any $\# T +1 \leq i \leq \# T+\#R$
    we have $(T\parallel R)(i)$
    equal to $R(i-\# T)$.
    Thus we can rewrite the interaction $S::(T\parallel R)$
    as follows:

    \scalebox{0.9}{\begin{minipage}{\linewidth}
    \begin{align*}      
    & S::(T\parallel R) \\
    = & 
    \left( 
      S +T + \sum_{1\leq i\leq \#T} \cutlink{S(i),T(i)} \right)+ R + 
    \sum_{\# T +1 \leq i \leq \# T\parallel R} \cutlink{S(i ),(T\parallel R)(i)}\\
    = & 
    \left( 
      S +T + \sum_{1\leq i\leq \#T} \cutlink{S(i),T(i)} \right)+ R + 
    \sum_{\# T +1 \leq i \leq \# T\parallel R} \cutlink{(S::T)(i- \#T),R(i - \# T)}\\
    = &
    \left( S +T + \sum_{1\leq i\leq \#T} \cutlink{S(i),T(i)} \right)+ R + 
    \sum_{\#T + 1 -\# T \leq i \leq \# T\parallel R -\#T} \cutlink{(S::T)(i ),R(i)}\\
    = &
    \left( S +T + \sum_{1\leq i\leq \#T} \cutlink{S(i),T(i)} \right)+ R + 
    \sum_{  1  \leq i \leq  R } \cutlink{(S::T)(i ),R(i)} \\
    = &
    S::T 
    + R + 
    \sum_{  1  \leq i \leq  R } \cutlink{(S::T)(i ),R(i)} \\
    = &
    (S::T)::R
    \end{align*}      
  \end{minipage}}

  \end{proof}

  \subsubsection{Interaction of nets and unrelated connective links}

      The following proposition 
      will be used in the proof of the result of adequacy 
      (\autoref{thm:adequacy}) in the inductive cases.
      In particular when one interprets $S$ as the image of a proof 
      $\pi_0$ of $\Gamma,A,B$,
      the link $e$ represents a $\parr$-link (say a link $\parrlink{p_1,p_2}{p}$)
      and $T$ is an opponent from $\reali[\ibase]{\Gamma}^\bot$. 

    \begin{proposition}\label{prop:commutationInterSumlink}
      Given two nets $S$ and $T$
      with $\# S \geq \# T$
      and a link $e$ such that 
      the positions of $\source(e)$
      are included in the outputs of $S::T$
      and are the only conclusions of $S$.
      Then we have:
      $$ (S:: T)  + e = (S+e)::T.$$
    \end{proposition}
    \begin{proof}
      Recall 
      the definition of the interaction
      $$S::T
      =
      S +\sum_{ 1\leq i \leq \#T} \cutlink{S(i),T(i)}
      +T
      $$

      Thus we have that 
      $$(S::T) +e
      =
      \left(S +\sum_{ 1\leq i \leq \#T} \cutlink{S(i),T(i)}
      +T\right) +e
      $$

      Because each input of $e$ is a conclusion of $S$
      in $S::T$
      then each input of $e$ is a conclusions 
      $S(j)$ of $S$ with $j\geq\# T$.
      Furthermore say $S = (\body S ,\arrange S)$
      then $S+e = (\body S + e , (\arrange S \setminus{\source(e)} )\at \target(e))$
      and for each $1\leq i \leq \# T$
      we have $S(i) = (S+e) (i)$.
      It follows that: 

      $$
      \left(S +\sum_{ 1\leq i \leq \#T} \cutlink{S(i),T(i)}
      +T\right) +e
      =
      \left(S +\sum_{ 1\leq i \leq \#T} \cutlink{(S+e)(i),T(i)}
      +T\right) +e
      $$

      Which itself is equal to 

      $
      (S + e) +\sum_{ 1\leq i \leq \#T} \cutlink{(S+e)(i),T(i)}
      +T 
      = (S+e) :: T
      $

    \end{proof}

    \subsection{On types, proving duality results and associativity}

  \subsubsection{The Merge Operator}\label{ann:merge}

  \begin{definition}[Merge of two nets]\label{def:mergeDaimon}
    The \emph{merge} of two daimon links $d = \dailink{\overline p}$
    and $d'=\dailink{\overline q}$
    is the daimon link $d\merge d' = \dailink{\overline p ,\overline q}$.
    Given 
    a net $S$ containinng a daimon link $d$ 
    i.e. $S= d + S_0$
    and a net $S'$ containing a daimon link $d'$
    i.e. $S' = d' + S'_0$.
    We denote by $S\merge[(d,d')]S'$
    the net $S_0 + S_0' + (d\merge d')$.
  \end{definition}

  \begin{proposition}\label{prop:mergeDecompo}
    Given two nets $S$ and $S'$
    respectively containing daimons $d$ and $d'$
    so that $S=S_0+d$ and $S' = S_0'+d'$:
    $S\merge[(d,d')] S'$
    is equal to $S\merge[(d,d')] d' + S_0'$.
  \end{proposition}
  \begin{proof}
    One merely needs 
    to write the nets 
    $S\merge[(d,d')] S'$
    and $(S\merge[(d,d')] d') + S_0'$
    as sum of links 
    to conclude.
  \end{proof}
  
  \begin{notation}
    We fix a convention 
    to denote the daimon links resulting from cut elimination steps
    involving daimons (glueing and non homogeneous):
    \begin{itemize}
      \item 
      In a glueing step
      say $d = \dailink{\overline p,a}$
      and $d' = \dailink{\overline q,b}$
      the glueing redex $d +\cutlink{a,b} + d'$
      reducing to $\dailink{\overline p , \overline q}$
      we denote that resulting daimon link by $g_{a,b}(d,d')$.
      \item 
      In a $(\otimes/\maltese)$ step
      say $d = \dailink{\overline p,a}$
      the reddex $d +\cutlink{a,b} + \tenslink{b_1,b_2}{b}$
      reducing to $\dailink{\overline p , a_1,a_2} + \cutlink{ a_1,b_1} + \cutlink{a_2,b_2}$
      we denote that resulting daimon link by $p_{a}(d)$.
      \item 
      In a $(\parr/\maltese)$ step
      say $d = \dailink{\overline p,a}$
      the reddex $d +\cutlink{a,b} + \parrlink{b_1,b_2}{b}$
      reducing to $\dailink{\overline p ^A , a_1}  + \dailink{\overline p ^B , a_2 }+ 
      \cutlink{ a_1,b_1} + \cutlink{a_2,b_2}$
      for a partition $\overline p = \overline p ^ A \uplus \overline p ^ B$
      we denote that resulting daimon links by $t^1_{a, A}(d)$
      and $t^2_{a,B}(d)$.
    \end{itemize}

    Given a reduction step $S \rightarrow S'$
    eliminating a cut $c$
    we denote $S\rightarrow^c S' $.
    Given an elimination step $S\rightarrow^c S'$
    that isn't multiplicative,
    we denote by 
    $d_c$ the daimon involve in the redex 
    if it is unique,
    otherwise we denote 
    the unordered pair of daimons $g(d_c)$
    in the case of a glueing cut.

  \end{notation}

  \begin{proposition}\label{prop:cutelimMerge}
    Given a net $S$ containing a cut $c=\cutlink{a,b}$
    and daimon $d$ of $S$
    if $S\rightarrow^c S'$:
    \begin{itemize}
      \item 
      If $c$ is multiplicative 
      $S\merge[d]\maltese_1 \rightarrow^c S'\merge[d] \maltese_1$.
      \item 
      If $d$ isn't $d_c$ or contained in $g(d_c)$
      then 
      $S\merge[d]\maltese_1 \rightarrow^c S'\merge[d] \maltese_1$.
      \item 
      If $c$ is a glueing cut 
      and $d$ belongs to $g(d_c)$
      then $S\merge[d]\maltese_1 \rightarrow S'\merge[g_{\source(c)}(g(d_c))]\maltese_1$
      \item 
      If $c$ is a $(\otimes/\maltese)$ cut 
      and $d=d_c$
      then $S\merge[d]\maltese_1 \rightarrow S'\merge[p_a(d)]\maltese_1$
      \item 
      If $c$ is a $(\parr/\maltese)$ cut 
      and $d=d_c$
      and $S\rightarrow^c S'$ by chosing a partition $A\uplus B$ of the daimon $d = \dailink{\overline p}$
      then $S\merge[d]\maltese_1 \rightarrow S'\merge[t^1_{a,A}(d)]\maltese_1$
      and $S\merge[d]\maltese_1 \rightarrow S'\merge[t^2_{a,B}(d)]\maltese_1$.
    \end{itemize}
  \end{proposition}
  \begin{proof}
    If the reduction $S\rightarrow S'$
    is a multiplicative step,
    then it does not affects the daimon of $S$
    and therefore we conclude.
    In the other cases one (or two for the glueing cut) daimon
    are involved in the cut--elimination,
    if these daimons don't contain $d$ 
    then by contextual closure one can ensure the proposition.

    Now let us treat the three cases where the daimon $d$ 
    is involved.
    \begin{itemize}
      \item Say the 
      reduction in $S$ 
      is that of a glueing cut:
      $$\dailink{\overline p , a} + \cutlink{a,b} +
        \dailink{\overline q, b} \rightarrow \dailink{\overline p,\overline q}.$$

      Say (without loss of generality) that $d = \dailink{\overline p, a}$
      then in $S \merge [d] \maltese_1$
      the daimon $d$ is isomorphic to $\dailink{\overline p, a , u}$
      and the redex becomes:
      $$\dailink{\overline p , a , u} + \cutlink{a,b} +
      \dailink{\overline q, b} \rightarrow \dailink{\overline p, u, \overline q}.$$

      Indeed the resulting daimon 
      is $g _{a,b}(d_1,d_2) \merge \maltese_1$
      and therefore the redex of $S\merge[d]\maltese_1$
      is $S'\merge[g_{a,b}(d_1,d_2)]\maltese_1$.

      \item 
      Say the reduction is that of 
      a $(\otimes,\maltese)$ cut:
      $$
      \dailink{\overline p, a} + \cutlink{a,b} +
      \tenslink{b_1,b_2}{b} \rightarrow 
      \dailink{\overline p, a_1, a_2} + \cutlink{a_1,b_1} +\cutlink{a_2,b_2} .$$

      It follows that:
      \begin{align*}
        & \dailink{\overline p, a} + \cutlink{a,b} +
      \tenslink{b_1,b_2}{b} \merge \maltese_1  \\
         = \quad &
        \dailink{\overline p, a ,u} + \cutlink{a,b} +
      \tenslink{b_1,b_2}{b} \\
         \rightarrow \quad &
        \dailink{\overline p, a_1, a_2 , u} + \cutlink{a_1,b_1} +\cutlink{a_2,b_2}   \\
         = \quad &
        \dailink{\overline p, a_1, a_2} + \cutlink{a_1,b_1} +\cutlink{a_2,b_2} \merge\maltese_1 
      \end{align*}

      Thus we have shown that 
      in $S\rightarrow S'$
      the daimon $d$ became 
      $p_{a}(d)$
      and we have that 
      $S\merge[d] \maltese_1$
      rewrites to 
      $S' \merge[p_a(d)] \maltese_1$.

      \item 
      Say the reduction is that of 
      a $(\parr,\maltese)$ cut:
      $$
      \dailink{\overline p, a} + \cutlink{a,b} +
      \parrlink{b_1,b_2}{b} \rightarrow 
      \dailink{\overline p_A, a_1} + \dailink{\overline p_B,a_2} +
       \cutlink{a_1,b_1} +\cutlink{a_2,b_2} .$$

      It follows that:
      \begin{align*}
        & \dailink{\overline p, a} + \cutlink{a,b} +
      \parrlink{b_1,b_2}{b}  \merge \maltese_1  \\
        = \quad &
        \dailink{\overline p, a ,u} + \cutlink{a,b} +
      \parrlink{b_1,b_2}{b} \\
        \rightarrow \quad &
        \dailink{\overline p_A, a_1 ,u} + \dailink{\overline p_B,a_2} +
       \cutlink{a_1,b_1} +\cutlink{a_2,b_2}   \\
        = \quad &
        \dailink{\overline p_A, a_1 } + \dailink{\overline p_B,a_2} +
       \cutlink{a_1,b_1} +\cutlink{a_2,b_2}  \merge\maltese_1 
      \end{align*}

      Thus we have shown that 
      in $S\rightarrow S'$
      the daimon $d$ can rewrites in two daimons 
      $t^1_{a,A}(d)$
      or $t^2_{b,B}(d)$
      for any choice of partition $A$ and $B$ of $\overline p$.
      As a consequence
      $S\merge[d] \maltese_1$
      rewrites to 
      $S' \merge[t^1_{a,A}(d)] \maltese_1$
      or to $S'\merge [t^2_{a,B}(d)]\maltese_1$
      still for any choice of partition.

    \end{itemize}
  \end{proof}
    
\begin{remark}\label{rem:mergeAndCutelim}
  The previous proposition 
  ensures the (less precise) statement:
  given a net $S$ with a daimon link $d$ 
  if $S\rightarrow S'$
  then there exists a daimon $d'$
  in $S'$ so that $S\merge[d]\maltese_1 \rightarrow S'\merge[d']\maltese_1$.
  By a simple induction, this can be generalised to multiple steps of cut elimination.
\end{remark}

  \begin{proposition}\label{prop:mergeCompute}
    Given $S$ and $T$ two nets such that $S::T\rightarrow \maltese_k$,
    $d$ a daimon of $S$, 
    and $d_0=\maltese_n$ some daimon link:
    the net $(S\merge[d] \maltese_n)::T$
    reduces to $\maltese_{n+k}$.
  \end{proposition}
  \begin{proof}
    Assume that $S$ and $T$ are orthogonal,
    then $S::T \rightarrow \maltese_k$.
    Consider $d$ some daimon of $S$
    then in particular it is a daimon of $S::T$
    and $S\merge[d]d_0 ::T$
    is equal to $S::T \merge[d] d_0$.
    Because $S::T\rightarrow \maltese_k$
    it follows that 
    $S::T\merge[d]\maltese_n\rightarrow \maltese_k\merge[d']d_0$
    where $d'$ is a daimon link of $\maltese_k$
    (\autoref{rem:mergeAndCutelim}).
    Since $\maltese_k$ is a net made of a single daimon link,
    necessarily $d' = \maltese_k$
    and thus 
    since $d_0=\maltese_n$
    $\maltese_k\merge[d']d_0=\maltese_{n+k}$
    Therefore we conclude 
    $S::T\merge[d]d_0\rightarrow \maltese_{n+k}$.
  \end{proof}

\begin{proposition}\label{prop:addConcluReduction}
  Given two nets $S$ and $T$ are orthogonal.
  $S'$ a net obtained from $S$
  by adding $n$ conclusions to some (eventually distinct) daimons  of $S$
  then $S'::T$ reduces to $\maltese_n$.
\end{proposition}
\begin{proof}
  One inductively uses the previous proposition(\autoref{prop:mergeCompute}).
\end{proof}

\subsubsection{Behavior of Identity cut--nets}

\begin{definition}
  We call a cut $c$ 
  occuring in a net $S$ 
  \emph{irreversible} whenever it is of type $(\maltese/\parr)$
  (\autoref{def:cuttypes}).
  An \emph{identity cut--net}
  whenever it is of the following form 
  $$\dailink{\vec a , p_1 , \vec b , p_2, \vec c}
  + \parrlink{p_1,p_2}{p}
  + \cutlink{p,q} 
  + \dailink{\vec d , q , \vec e}$$

  Namely, the $\parr$--involved above the input $p$ of the 
  cut link
  is such that both its inputs are outputs of the same daimon .
\end{definition}

\begin{proposition}[Normal Form of Identity cut--nets]\label{prop:cutnetNF}
  An identity cut--net 
  $S$ with $n$ conclusions
  has a unique normal form 
  which is (up to isomorphism)
  $\maltese_n$.
\end{proposition}
\begin{proof}
  Note that the conclusions of an identity 
  cut--net are all outputs of 
  (one of the two) daimon links.
  Writing $S$ as the following:
  $$S ::= \dailink{\vec a , p_1 , \vec b , p_2, \vec c}
  + \parrlink{p_1,p_2}{p}
  + \cutlink{p,q} 
  + \dailink{\vec d , q , \vec e}$$
  The $n$ conclusions of $S$
  correspond 
  to the set obtained from the vector 
  $\vec a, \vec b , \vec c , \vec d , \vec e $
  (note that this vetor contains no repetition).

  Fixing a choice for the partition 
  of the context $\vec d , \vec e$
  of the daimon containing the position $q$ 
  we obtain a reduction:
  \begin{align*}
    &
  \dailink{\vec a , p_1 , \vec b , p_2, \vec c}
  + \parrlink{p_1,p_2}{p}
  + \cutlink{p,q} 
  + \dailink{\vec d , q , \vec e} \\ 
  \rightarrow \quad &
  \dailink{\vec a , p_1 , \vec b , p_2, \vec c}
  + \cutlink{p_1,q_1}
  + \cutlink{p_2,q_2}
  + \dailink{\vec d_{\pleft} , q_1 , \vec e_{\pleft}}
  + \dailink{\vec d_{\pright} , q_2 , \vec e_{\pright}}  
  \end{align*}

  At this point two cuts occur thus two reductions are possible,
  both lead to the same net.
  Eliminating first $\cutlink{p_1,q_1}$:
  \begin{align*}
    &
  \dailink{\vec a , p_1 , \vec b , p_2, \vec c}
  + \cutlink{p_1,q_1}
  + \cutlink{p_2,q_2}
  + \dailink{\vec d_{\pleft} , q_1 , \vec e_{\pleft}}
  + \dailink{\vec d_{\pright} , q_2 , \vec e_{\pright}} \\ 
  \rightarrow  \quad &
  \dailink{\vec a  , \vec b , p_2, \vec c , \vec d_{\pleft}  , \vec e_{\pleft}}
  + \cutlink{p_2,q_2}
  + \dailink{\vec d_{\pright} , q_2 , \vec e_{\pright}} \\ 
  \rightarrow  \quad &
  \dailink{\vec a  , \vec b , \vec d_{\pright} , \vec e_{\pright} , \vec c , \vec d_{\pleft}  , \vec e_{\pleft}}\\
  \equiv \quad & 
  \dailink{\vec a , \vec b , \vec c , \vec d , \vec e}\\ 
  \equiv \quad & 
  \dailink{S(1),\dots, S(n)}
  \end{align*}

  Eliminating first $\cutlink{p_2,q_2}$:
  \begin{align*}
    &
  \dailink{\vec a , p_1 , \vec b , p_2, \vec c}
  + \cutlink{p_1,q_1}
  + \cutlink{p_2,q_2}
  + \dailink{\vec d_{\pleft} , q_1 , \vec e_{\pleft}}
  + \dailink{\vec d_{\pright} , q_2 , \vec e_{\pright}} \\
  \rightarrow \quad &
  \dailink{\vec a , p_1 , \vec b , \vec d_{\pright}  , \vec e_{\pright} , \vec c}
  + \cutlink{p_1,q_1}
  + \dailink{\vec d_{\pleft} , q_1 , \vec e_{\pleft}} \\ 
  \rightarrow \quad &
  \dailink{\vec a , \vec d_{\pleft}  , \vec e_{\pleft} , \vec b , \vec d_{\pright}  , \vec e_{\pright} , \vec c} \\ 
  \equiv \quad & 
  \dailink{\vec a , \vec b , \vec c , \vec d , \vec e}\\ 
  \equiv \quad & 
  \dailink{S(1),\dots, S(n)}
  \end{align*}
  This is true for any choice of partition 
  during the first $(\maltese/\parr)$
  step,
  therefore 
  $S$ has a unique normal
  which is (upto iso) $\dailink{S(1),\dots,S(n)}$.
\end{proof}

\begin{remark}\label{rem:cutnet0conclu}
  From \autoref{prop:cutnetNF}
  one derive that an identity cut--net without 
  conclusions always reduces to 
  $\maltese_0$.
\end{remark}

\begin{remark}\label{rem:cutfreeNoOutputNets}
  A net which is cut--free and has no conclusions 
  is a sum of daimon links without conclusions:
  such a net cannot contain connective 
  otherwise to have no conclusion 
  a cut needs to be added,
  similarly its daimon links cannot have outputs 
  otherwise to have no conclusions a cut must be present in the net.
\end{remark}

\begin{proposition}[Identity cut--net behavior]\label{prop:cutnetBehavior}
  Consider an identity cut $S$ without conclusions
  and denote its two daimons $d_1$ and $d_2$.
  Let $T_1$ and $T_2$ be two nets without conclusions,
  and $d_1'$ a daimon of $T_1$ while $d_2'$ denotes a daimon of $T_2$,
  the assertions are equivalent:
  \begin{enumerate}
    \item 
    $T_1\rightarrow ^* \maltese_0$ and $T_2\rightarrow^* \maltese_0$.
    \item 
    $S \merge[d_1,d_1'] T_1 \merge[d_2,d_2'] T_2$ 
    reduces to $\maltese_0$
  \end{enumerate}
\end{proposition}
\begin{proof}
 $1\Rightarrow 2.$ This is the easy implication.
 $T_2\rightarrow^* \maltese_0$ thus by applying \autoref{prop:mergeCompute}
 $S \merge[d_1,d_1'] T_1 \merge[d_2,d_2'] T_2$ 
 reduces to 
 $S \merge[d_1,d_1'] T_1 \merge[d_2,f (d_2')] \maltese_0$
 that is $S \merge[d_1,d_1'] T_1$
 again applying \autoref{prop:mergeCompute}
 with $T_1\rightarrow ^* \maltese_0$ 
 yields that 
 $S \merge[d_1,d_1'] T_1$ 
 reduces to 
 $S\merge[d_1] \maltese_0$
 that is $S$.
 We conclude since $S\rightarrow^*\maltese_0$ 
 (\autoref{rem:cutnet0conclu}).

 $2\Rightarrow 1.$
 The other direction requires the use of \autoref{prop:SN} (\autoref{item:delayParrs})
 and \autoref{prop:SN} (\autoref{item:earlyNonPar}).
 The cut that we will denote $c$ occuring in $S$
 is a $(\parr/\maltese)$ cut 
 thus it can be performed last (\autoref{prop:SN}\autoref{item:delayParrs}).
 Thus if     $S \merge[d_1,d_1'] T_1 \merge[d_2,d_2'] T_2$ 
reduces to $\maltese_0$
one can factorize the reduction as 
$S \merge[d_1,d_1'] T_1 \merge[d_2,d_2'] T_2 
\rightarrow^* U$
with $U\rightarrow_c \maltese_0$,
one can verify that this implies that $U$ is (isomorphic to) the following net 
$$
\dailink{p_1,p_2}{p} +\cutlink{p,q} + \dailink{q}
$$
That is, $U$ and $S$ are isomorphic nets.

In that case note that 
the cuts of $T_1$ and $T_2$
in the net $S \merge[d_1,d_1'] T_1 \merge[d_2,d_2'] T_2 $
are unrelated
and thus applying \autoref{prop:unrelatedConflu},
one can factorize the reduction leading to $U$ as follows:

\begin{equation}\label{equa:theReduction}
  S \merge[d_1,d_1'] T_1 \merge[d_2,d_2'] T_2 
\rightarrow^* _2
S_1
\rightarrow^*_1
S
\rightarrow_c \rightarrow^*
\maltese_0
\end{equation}

where $\rightarrow_1$ eliminates cuts occuring in $T_1$ 
or its reducts only 
and $\rightarrow_2$ eliminates cuts occuring in $T_2$ or its reducts only.
As a consequence using \autoref{prop:mergeCompute}
one may rewrite the first step of the reduction:
$(S \merge[d_1,d_1'] T_1) \merge[d_2,d_2'] T_2 
\rightarrow^* _2
(S \merge[d_1,d_1'] T_1)\merge[d_2,f(d_2')]T_2'$

Necessarily because $c$ (and the cut it produces) is a cut that is not in $T_2$
or any of its redexes,
and because the cuts of $T_1$ and $T_2$ also are disjoint
the reduction \autoref{equa:theReduction}
implies that $T_2'$ must be cut--free
and thus in particular a normal form of $T_2$. 

Similarly 
because 
$(S \merge[d_1,d_1'] T_1)\merge[d_2,f(d_2')]T_2'$ 
reduces to 
$S$ by $\rightarrow^*_1 $
applying \autoref{prop:mergeCompute}
one obtains that 
$S_1 = (S \merge[d_1,d_1'] T_1)\merge[d_2,f(d_2')]T_2'$ 
reduces to $S$
in the following way:
there is a reduction $T_1\rightarrow T_1'$
and $S$ equals 
$(S \merge[d_1,d_1'] T_1')\merge[d_2,f(d_2')]T_2'$.
Again
because $c$ (and the cut it produces) is a cut that is not in $T_1$
or any of its redexes,
and because the cuts of $T_1$ and $T_2$ also are disjoint
the reduction \autoref{equa:theReduction}
implies that $T_1'$ must be cut--free
and thus in particular a normal form of $T_1$. 

From this argument one shows that 
\begin{equation}\label{equa:TheEqual}
  S = (S \merge[d_1,d_1'] T_1')\merge[d_2,f(d_2')]T_2'.
\end{equation}
Furthermore $T_1'$ and $T_2'$
are both cut--free and without conclusions 
they are therefore sums of $\maltese_0$ links \autoref{rem:cutfreeNoOutputNets}.
If $T_1'$ (resp $T_2'$) 
is $\sum_{1\leq i \leq n}\maltese_0$ (resp. $\sum_{1\leq i \leq k}\maltese_0$)
then 
$S = (S \merge[d_1,d_1'] \sum_{1\leq i \leq n}\maltese_0 )\merge[d_2,f(d_2')]T_2'$
equals 
$(S + \sum_{1\leq i \leq n-1}\maltese_0 )\merge[d_2,f(d_2')]T_2'$
that will be 
$(S + \sum_{1\leq i \leq n-1}\maltese_0 )\merge[d_2,f(d_2')]\sum_{1\leq i \leq k}\maltese_0$
that is 
$S + \sum_{1\leq i \leq n-1}\maltese_0 + \sum_{1\leq i \leq k-1} \maltese_0$.

From \autoref{equa:TheEqual}
this means that necessarily 
$n-1 =0$ and $k-1 = 0$
i.e. 
$T_1'$ and $T_2'$ equal $\maltese_0$.
Therefore we conclude that 
$T_1\rightarrow ^* \maltese_0$
and 
$T_2\rightarrow^*  \maltese_0$

\end{proof}

  \subsubsection{Proofs of section \autoref{sect:RealModel} --
  \autoref{prop:dualPreC} , \autoref{prop:dualC} and \autoref{prop:Associativity}}

\propdualPreC*
  \begin{proof} 
    Consider a net $S$
    orthogonal to $\mathbf A \parallel \mathbf B$. For any $a\in\mathbf A$ and $b\in\mathbf B$
    the net $S::(a\parallel b)$
    reduces to $\maltese_0$. Since $S::(a\parallel b)=S::a::b$ (\autoref{prop:interAsAction})
    and since the orthogonality holds for any pair $(a,b)$, in particular for any net $b\in\mathbf B$ the net 
    $(S::a)::b$ reduces to $\maltese_0$.
    This means that, for any $a\in\mathbf A$, $S::a$ is orthogonal to $\mathbf B$ i.e.\ $S::a\in\mathbf B ^\bot$.
    Since $\mathbf A = (\mathbf A^{\bot})^\bot$, this means that $S\in\mathbf A^\bot \compo \mathbf B ^\bot$.
  
    On the other hand, consider a net $S\in\mathbf A^\bot\compo\mathbf B^\bot$:  
    for any $a\in\mathbf A = (\mathbf A^{\bot})^\bot$ the net $S::a$ belongs to $\mathbf B^\bot$ and so for any $b\in\mathbf B$
    the net $S::a::b$ reduces to $\maltese_0$.
    Since $S::a::b=S::(a\parallel b)$ (\autoref{prop:interAsAction}), we have that $S$ is orthogonal to $\mathbf A \parallel \mathbf B$ (using Remark~\ref{rem:density}). 
    
    Hence we showed $(\mathbf A\parallel \mathbf B) ^\bot=\mathbf A^\bot \compo \mathbf B ^\bot$. As a consequence $(\mathbf A \compo \mathbf B)^\bot=(\mathbf A^{\bibot} \compo \mathbf B^{\bibot})^\bot=(\mathbf A ^\bot \parallel \mathbf B^\bot )^{\bibot}=\mathbf A ^\bot \parallel \mathbf B^\bot$.
  \end{proof}

  \begin{lemma}\label{lem:parrDualityHardCase}
    Given a net with no conclusions $R$
    and two types with one conclusion $\mathbf A$
    and $\mathbf B$:
    If 
    $R\merge[d]\dailink{p_1,p_2} + \parrlink{p_1,p_2}{p} $belongs to $\mathbf A^\bot \parr \mathbf B^\bot$
    then 
    $R\merge[d]\dailink{p} $ belongs to $\mathbf A^\bot \parr \mathbf B^\bot$
  \end{lemma}
  \begin{proof}
    To do so one shows that 
    $R\merge[d]\dailink{p} $ is orthogonal to any net 
    orthogonal to $\mathbf A^\bot \parr \mathbf B^\bot$.
    Consider therefore an opponent 
    $U$
    orthogonal to $\mathbf A^\bot \parr \mathbf B^\bot$.
    We distinguish two cases:
    \begin{itemize}
      \item 
      If $U = U'+ \tenslink{U'(1),U'(2)}{q}$
      has a terminal tensor link.
      Note that, by assumption: 
      $$R\merge[d]\dailink{p_1,p_2} + \parrlink{p_1,p_2}{p} :: U \rightarrow ^* \maltese_0$$
      Since multiplicative cuts can be performed first
      and 
      $R\merge[d]\dailink{p_1,p_2} + \parrlink{p_1,p_2}{p} :: U$ 
      reduces to 
      $R\merge[d]\dailink{p_1,p_2} :: U'$ by eliminating a multiplicative cut 
      hence we derive the following (\autoref{prop:factorisation})
      $$ R\merge[d]\dailink{p_1,p_2} :: U' \rightarrow ^* \maltese_0.$$

      In the meanwhile we have the following reduction:
      \begin{align*}
        &R\merge[d]\dailink{p}  :: U \\
        =\quad &
        R\merge[d]\dailink{p}  :: U'+ \tenslink{U'(1),U'(2)}{q}  \\
        \rightarrow \quad & 
        R\merge[d]\dailink{p_1,p_2}  :: U'   \\
        \rightarrow \quad & 
        \maltese_0 \mbox{(previous argument)}
      \end{align*}

      \item 
      If $U$ has a terminal daimon link 
      outputing its only conclusion $q$.
      Then $U$ may written as
      as $\dailink{q} \merge[d'] R'$.
      Therefore one can identify an identity cut--net $V$ in the interaction:
      \begin{align*}
        & R\merge[d]\dailink{p_1,p_2} + \parrlink{p_1,p_2}{p} :: U \\
        =\quad &
        R\merge[d]\dailink{p_1,p_2} + \parrlink{p_1,p_2}{p} :: \dailink{q} \merge[d'] R'\\
        =\quad &
        (R\merge[d]\dailink{p_1,p_2}) + \parrlink{p_1,p_2}{p} +\cutlink{p,q} + (\dailink{q} \merge[d'] R') \\
        =\quad &
        [\dailink{p_1,p_2} + \parrlink{p_1,p_2}{p} +\cutlink{p,q} + (\dailink{q} \merge[d'] R')] 
        \merge[d]R \\
        =\quad &
        [ \dailink{p_1,p_2}+ \parrlink{p_1,p_2}{p} +\cutlink{p,q} + \dailink{q} ]\merge[d'] R' \merge[d] R \\
        =\quad &
        [ V ]\merge[d'] R' \merge[d] R
      \end{align*}

      Now since $V$ is a cut net without conclusions
      while $R'$ and $R$ both have no conclusion 
      while 
      $[ V ]\merge[d'] R' \merge[d] R \rightarrow^* \maltese_0$ 
      (because by assumption $R\merge[d]\dailink{p} :: U \rightarrow \maltese_0$)
      we derive that $R\rightarrow^* \maltese_0$
      and $R'\rightarrow\maltese_0$ (\autoref{prop:cutnetBehavior}).

      We then derive the following 
      \begin{align*}
        & R\merge[d]\dailink{p} :: U \\
        = \quad 
        & R\merge[d]\dailink{p} ::  \dailink{q} \merge[d']R' \\
        \rightarrow^* \quad
        & \maltese_0\merge[d]\dailink{p} ::  \dailink{q} \merge[d']R' \\
        = \quad
        & \dailink{p} ::  \dailink{q} \merge[d']R' \\
        \rightarrow^* \quad
        & \dailink{p} ::  \dailink{q} \merge[d'] \maltese_0 \\
        = \quad
        & \dailink{p} ::  \dailink{q} \\
        \rightarrow \quad
        & \maltese_0 
      \end{align*}

      in the end we have shown that 
      $R\merge[d]\dailink{p} $ belongs to $(\mathbf A^\bot \parr \mathbf B^\bot)^{\bibot}$
      which is $\mathbf A^\bot \parr \mathbf B^\bot$ and conclude.
    \end{itemize}
  \end{proof}

\propdualC*
  \begin{proof} 
  One equality implies the other: say that $(\mathbf A \otimes \mathbf B)^\bot = \mathbf A^\bot \parr \mathbf B^\bot$ 
    holds for any pair of types. Then $\mathbf A^\bot \otimes \mathbf B^\bot=(\mathbf A^\bot \otimes \mathbf B^\bot)^{\bibot}=(\mathbf A^{\bibot} \parr \mathbf B^{\bibot})^\bot=(\mathbf A\parr \mathbf B)^\bot$.
  
  Let us now prove $\mathbf A^\bot \parr \mathbf B^\bot=(\mathbf A \otimes \mathbf B)^\bot$. As for 
    $\mathbf A^\bot \parr \mathbf B^\bot\subseteq (\mathbf A \otimes \mathbf B)^\bot$, we prove 
     $$
    \{ S + \tenslink{S(1),S(2)}{p} \mid S \in \mathbf A\parallel \mathbf B \}
    \subseteq \{ T + \parrlink{T(1),T(2)}{p} \mid T \in \mathbf A^\bot\compo \mathbf B^\bot \}^\bot
    $$
    which is enough to conclude because the previous inclusion yields (since bi--orthogonality preserves inclusion)
    $$
    \mathbf A \otimes \mathbf B
    \subseteq \{ T + \parrlink{T(1),T(2)}{p} \mid T \in \mathbf A^\bot\compo \mathbf B^\bot \}^\bot
    $$
    that is (since the tri--orthogonal of a set is the orthogonal of this set):
    $$
    \mathbf A \otimes \mathbf B
    \subseteq
    (\mathbf A^\bot \parr \mathbf B^\bot )^\bot.
    $$
    Which 
    implies $(\mathbf A \otimes \mathbf B)^\bot
    \supseteq
    (\mathbf A^\bot \parr \mathbf B^\bot )$ (Orthogonality inverts inclusion).

    So let 
    $S_{0}\in\{ S + \tenslink{S(1),S(2)}{p} \mid S \in \mathbf A\parallel \mathbf B \}$ 
    and 
    $T_{0}\in\{ T + \parrlink{T(1),T(2)}{p} \mid T \in \mathbf A^\bot\compo \mathbf B^\bot \}$. 
    One easily sees that $S_{0}\bot T_{0}$: after eliminating the multiplicative cut we obtain a net $S::T$, 
    where $S\in\mathbf A^\bot \compo \mathbf B^\bot$ and $T\in\mathbf A\parallel \mathbf B$: 
    these constructions are orthogonal (Proposition~\ref{prop:dualPreC}).

    \bigskip

As for $(\mathbf A \otimes \mathbf B)^\bot\subseteq\mathbf A^\bot \parr \mathbf B^\bot$, 
let $T\in (\mathbf A \otimes \mathbf B)^\bot$ i.e.
$T\bot\{ S + \tenslink{S(1),S(2)}{p} \mid S \in \mathbf A\parallel \mathbf B \}^{\bot\bot}$; 
we can deduce that the unique terminal link of $T$ is either a $\parr$--link or a $\maltese$--link: 
\begin{itemize}
\item
if $T=T_0 + l$ 
has a terminal $\parr$--link  $l$ then, 
for every $S_0\in\mathbf A\parallel \mathbf B$
the net $S::=S_0+\tenslink{S(1),S(2)}{p}$ belongs $\mathbf A \otimes \mathbf B$. 
In a single multiplicative step of cut elimination $T::S$ reduces to
$T_0::S_0$ and since  $S::T \rightarrow^* \maltese_0$
it follows that $T_0::S_0 \rightarrow^* \maltese_0$, using \autoref{prop:SN} \autoref{item:earlyNonPar}. 
As a consequence $T_0\bot\ \mathbf A\parallel \mathbf B$ 
i.e. $T_{0}\in\mathbf A ^\bot\compo \mathbf B^\bot$ and $T\in\mathbf A^\bot \parr \mathbf B^\bot$;

    \item
if $T$ has for terminal link a daimon link then
$T$ may be written as $R \merge [d] \dailink{p}$
where $R$ is a net without conclusions.
Let us start an argument 
from the assumption that $T$ belongs to $(\mathbf A \otimes \mathbf B)^\bot$:
\begin{align*}
  &
  T \perp \mathbf A \otimes \mathbf B \\
  \Leftrightarrow \quad &
  T \perp  \{ S_a \parallel S_b \mid S_a \in \mathbf A , S_b \in \mathbf B\}^{\bibot} \\
  \Leftrightarrow \quad &
  T \perp  \{ S_a \parallel S_b + \tenslink{S_a(1),S_b(1)}{p}\mid S_a \in \mathbf A , S_b \in \mathbf B\} \\
  \Leftrightarrow \quad &
  T ::  S_a \parallel S_b + \tenslink{S_a(1),S_b(1)}{p} \rightarrow ^* \maltese_0 
  \quad (\mbox{for any } S_a\in\mathbf A , S_b\in \mathbf B)\\
  \Leftrightarrow \quad &
  R \merge [d] \dailink{p} ::  S_a \parallel S_b + \tenslink{S_a(1),S_b(1)}{p} 
  \rightarrow ^* \maltese_0 
  \quad (\mbox{for any } S_a\in\mathbf A , S_b\in \mathbf B) \\
  \Leftrightarrow \quad &
  R \merge [d] \dailink{p_1,p_2} ::  S_a \parallel S_b
  \rightarrow ^* \maltese_0 
  & (\mbox{\autoref{prop:SN}\autoref{item:earlyNonPar}})\\
  \Leftrightarrow \quad &
  R \merge [d] \dailink{p_1,p_2} \perp \{ S_a \parallel S_b \mid S_a \in \mathbf A , S_b \in \mathbf B \} \\
  \Leftrightarrow \quad &
  R \merge [d] \dailink{p_1,p_2} \perp \mathbf A \parallel \mathbf B \\
  \Leftrightarrow \quad &
  R \merge [d] \dailink{p_1,p_2} \in \mathbf A^\bot \compo \mathbf B^\bot \\
  \Rightarrow \quad &
  R \merge [d] \dailink{p_1,p_2} +\parrlink{p_1,p_2}{p} \in \mathbf A^\bot \parr \mathbf B^\bot \\
  \Rightarrow \quad &
  R \merge [d] \dailink{p} \in \mathbf A^\bot \parr \mathbf B^\bot 
  & (\mbox  {\autoref{lem:parrDualityHardCase}})\\
  \Leftrightarrow \quad &
  T \in \mathbf A^\bot \parr \mathbf B^\bot 
  & 
\end{align*}


\end{itemize}

    \end{proof}
  
    \associativity* 
    \begin{proof}
      The result is a consequence of the two following properties:
      \begin{enumerate}
      \item\label{item:density}
      the density of the parallel composition (\autoref{rem:density});
      \item\label{item:parallelNets}
      \autoref{prop:interAsAction}: 
      more precisely,
      for any nets $S_1,S_2,S_3,S_4$ such that $\#S_1 \geq \# S_2 + \#S_3 + \# S_4$, we have
       $$(S_1::(S_2 \parallel S_3))::S_4=((S_1 :: S_2) ::S_3 )::S_4=(S_1 :: S_2) :: (S_3\parallel S_4).$$
      \end{enumerate}
      By \autoref{prop:dualPreC} it is enough 
        to show that one of the two constructions 
        is associative. Let us do it for the parallel composition: more precisely, using property \ref{item:density} we prove that $(A \parallel (B\parallel C))^{\bot}=(A \parallel^{-} (B^{-}\parallel C^{-}))^{\bot}$ and $((A\parallel B) \parallel C)^{\bot}=((A\parallel^{-} B) \parallel^{-} C)^{\bot}$, and using property \ref{item:parallelNets} that $(A \parallel^{-} (B^{-}\parallel C^{-}))^{\bot}=((A\parallel^{-} B) \parallel^{-} C)^{\bot}$. From the previous equalities one deduces 
        $A \parallel (B\parallel C)=(A \parallel (B\parallel C))^{\bot\bot}=((A\parallel B) \parallel C)^{\bot\bot}=(A\parallel B) \parallel C$.
        
        To prove the equality 
        $(A \parallel (B\parallel C))^{\bot}=(A \parallel^{-} (B \parallel ^{-} C))^{\bot}$ one proves that, for every net $S$, the following equivalence holds: $S\perp A \parallel (B\parallel C)\iff S\perp A \parallel^{-} (B^{-}\parallel C^{-})$. Indeed: 
      
        \begin{align*}
          & S \perp A \parallel (B\parallel C) & \\
          \Leftrightarrow & S \perp A \parallel^- (B\parallel C) & \mbox{(by property \ref{item:density})}\\
          \Leftrightarrow & \forall a \in A , \forall x \in B\parallel C ,  
          S :: (a\parallel x)\rightarrow \maltese_0  & \mbox{(by definition)} \\
          \Leftrightarrow & \forall a \in A , \forall x \in B\parallel C ,  
            (S :: a) :: x \rightarrow \maltese_0  
            &\mbox{(by \autoref{prop:interAsAction})} \\
          \Leftrightarrow & \forall a \in A , S::a \perp (B\parallel C)  &\mbox{(by definition)}\\
          \Leftrightarrow & \forall a \in A ,S :: a \perp (B\parallel^- C) & \mbox{(by property \ref{item:density})}\\
          \Leftrightarrow & \forall a \in A ,\forall x\in (B\parallel^- C), (S :: a) :: x \rightarrow \maltese_0 & \mbox{(by definition)}\\
          \Leftrightarrow & \forall a \in A ,\forall x\in (B\parallel^- C), S :: (a\parallel x) \rightarrow \maltese_0 
          & \mbox{(by \autoref{prop:interAsAction})}\\
          \Leftrightarrow & S \perp A\parallel^- (B\parallel^- C) &\mbox{(by definition).}
          \end{align*}
      
        To prove the equality $((A\parallel B) \parallel C)^{\bot}=((A\parallel^{-} B) \parallel^{-} C)^{\bot}$, we need to express the (obvious) fact that, given three nets $S,T,U$ such that $\# S=\# T+\# U$, once the list of $S$'s conclusions that one decides to cut with the list of conclusions of $T$ (resp.\ $U$) is chosen, the interaction between $S$ on the one hand and $T$ and $U$ on the other hand is uniquely determined. With the conventions of definition \ref{def:interaction}, when we write $(S :: T) :: U$ we cut the first part (``the head'') of the sequence $a(S)$ with $a(T)$ and the last part (``the tail'') of the sequence $a(S)$ with $a(U)$. We denote by $ex_{U}(S)$ the net $S$ where we have swapped the tail and the head of $a(S)$: the previous fact can then be expressed by the equality $(S :: T) :: U=(ex_{U}(S) :: U) :: T$.\\
        Like for the previous equality, we prove that, for every net $S$, the following equivalence holds: $S\perp (A\parallel B) \parallel C\iff S\perp (A\parallel^{-} B) \parallel^{-} C$. Indeed: 
      
           \begin{align*}    
            & S\perp (A\parallel B) \parallel C &\\
            \Leftrightarrow & S \perp (A\parallel B) \parallel^- C &\mbox{(by property \ref{item:density})}\\
            \Leftrightarrow & \forall x\in A\parallel B,\forall c\in C,
                  S :: (x\parallel c)\rightarrow \maltese_0  & \mbox{(by definition)} \\
             \Leftrightarrow & \forall x\in A\parallel B,\forall c\in C,
            (S :: x) :: c\rightarrow \maltese_0  &\mbox{(by \autoref{prop:interAsAction})}\\
            \Leftrightarrow & \forall x\in A\parallel B,\forall c\in C,
            (ex_{c}(S) :: c) :: x\rightarrow \maltese_0  &\mbox{(by the obvious fact above)}\\
           \Leftrightarrow & \forall c\in C,
            ex_{c}(S) :: c\in (A\parallel B)^{\bot}  &\mbox{(by definition)}\\       
            \Leftrightarrow & \forall c\in C,
            ex_{c}(S) :: c\in (A\parallel^{-} B)^{\bot}  &\mbox{(by property \ref{item:density})}\\ 
            \Leftrightarrow & \forall a\in A,\forall b\in B,\forall c\in C,
            (ex_{c}(S) :: c) :: (a\parallel b)\rightarrow \maltese_0  &\mbox{(by definition)}\\
             \Leftrightarrow & \forall a\in A,\forall b\in B,\forall c\in C,
            ((ex_{c}(S) :: c) :: a) :: b\rightarrow \maltese_0  &\mbox{(by \autoref{prop:interAsAction})}\\   
             \Leftrightarrow & \forall a\in A,\forall b\in B,\forall c\in C,
            ((S :: a) :: b) :: c\rightarrow \maltese_0  &\mbox{(by the obvious fact above)}\\        
            \Leftrightarrow & \forall a\in A,\forall b\in B,\forall c\in C,
            (S :: (a\parallel b)) :: c\rightarrow \maltese_0  &\mbox{(by \autoref{prop:interAsAction})}\\     
              \Leftrightarrow & \forall a\in A,\forall b\in B,\forall c\in C,
            S :: ((a\parallel b)\parallel c)\rightarrow \maltese_0  &\mbox{(by \autoref{prop:interAsAction})}\\        
             \Leftrightarrow & S\perp (A\parallel^{-} B) \parallel^{-} C &\mbox{(by definition).}\\
         \end{align*}
         
         To prove the last equality $(A \parallel^{-} (B^{-}\parallel C^{-}))^{\bot}=((A\parallel^{-} B) \parallel^{-} C)^{\bot}$, we proceed like before and we show that, for every net $S$, the following equivalence holds: $S\perp A \parallel^{-} (B^{-}\parallel C^{-})\iff S\perp (A\parallel^{-} B) \parallel^{-} C$. Indeed: 
         
          \begin{align*}
          & S\perp A \parallel^{-} (B^{-}\parallel C^{-}) & \\
          \Leftrightarrow & \forall a \in A , \forall b \in B,\forall c\in C ,  
          S :: (a \parallel (b\parallel c))\rightarrow \maltese_0  & \mbox{(by definition)} \\
          \Leftrightarrow & \forall a \in A , \forall b \in B,\forall c\in C ,  
          (S :: a) :: (b\parallel c)\rightarrow \maltese_0  & \mbox{(by \autoref{prop:interAsAction})} \\
          \Leftrightarrow & \forall a \in A , \forall b \in B,\forall c\in C ,  
          ((S :: a) :: b) :: c\rightarrow \maltese_0  & \mbox{(by \autoref{prop:interAsAction})} \\
          \Leftrightarrow & \forall a \in A , \forall b \in B,\forall c\in C ,  
          (S :: (a\parallel b)) :: c\rightarrow \maltese_0  & \mbox{(by property \ref{item:parallelNets})} \\
          \Leftrightarrow & \forall a \in A , \forall b \in B,\forall c\in C ,  
          S :: ((a\parallel b) \parallel c)\rightarrow \maltese_0  &\mbox{(by \autoref{prop:interAsAction})} \\
          \Leftrightarrow & S \perp (A\parallel^- B)\parallel^- C &\mbox{(by definition).}
          \end{align*}
      \end{proof}

  \section{Complements to \autoref{sect:Adequacy}}

\subsection{Proving \autoref{rem:InterpretMLL}}

\begin{proposition}[Constructions preserve inclusion]\label{prop:inclusionPreserve}
  Consider four types 
  $\mathbf A_0\subseteq \mathbf A$
  and $\mathbf B_0 \subseteq \mathbf B$
  then:
  \begin{enumerate}
    \item 
    $\mathbf A_0 \parallel \mathbf B_0 \subseteq \mathbf A\parallel\mathbf B$
    \item 
    $\mathbf A_0 \compo \mathbf B_0 \subseteq \mathbf A\compo\mathbf B$
    \item 
    $\mathbf A_0 \otimes \mathbf B_0 \subseteq \mathbf A\otimes\mathbf B$
    \item 
    $\mathbf A_0 \parr \mathbf B_0 \subseteq \mathbf A\parr\mathbf B$

  \end{enumerate}
\end{proposition}
\begin{proof}
  We treat each point independently.
  \begin{enumerate}
    \item 
  Consider $x$ an element of 
  $\mathbf A_0 \parallel^- \mathbf B_0$
  then $x$ is of the form 
  $a_0\parallel b_0$
  with $a_0 \in \mathbf A_0$
  and $b_0\in \mathbf B_0$.
  Because we have the inclusion 
  $\mathbf A_0\subseteq \mathbf A$
  and $\mathbf B_0 \subseteq \mathbf B$
  it follow then that 
  $x=a_0\parallel b_0$
  belongs to $\mathbf A \parallel^- \mathbf B$
  and thus to $\mathbf A \parallel \mathbf B$.

  As a consequence 
  $\mathbf A_0 \parallel^- \mathbf B_0 \subseteq \mathbf A\parallel\mathbf B$
  thus, because bi orthogonality preserves inclusion, 
  it follows that 
  $(\mathbf A_0 \parallel^- \mathbf B_0) ^{\bibot} \subseteq (\mathbf A\parallel\mathbf B)^{\bibot}$
  therefore 
  $\mathbf A_0 \parallel \mathbf B_0  
  \subseteq \mathbf A\parallel\mathbf B$
  by density (\autoref{rem:density})

  \item 
  If $\mathbf A_0\subseteq \mathbf A$
  and $\mathbf B_0\subseteq \mathbf B$
  we equivalently have the inclusions
  $\mathbf A_0^\bot\supseteq \mathbf A^\bot$
  and $\mathbf B_0^\bot\supseteq \mathbf B^\bot$.
  Using the previous demonstrated fact 
  it follows that 
  $\mathbf A_0^\bot\parallel  \mathbf B_0^\bot$
  contains $\mathbf A^\bot\parallel \mathbf B^\bot$.
  Again using the fact that orthoganility invert inclusions 
  we then derive that
  $(\mathbf A_0^\bot\parallel  \mathbf B_0^\bot)^\bot$
  is included in  $(\mathbf A^\bot\parallel \mathbf B^\bot)^\bot$.
  By duality 
  this means 
  $\mathbf A_0 \compo \mathbf B_0 \subseteq \mathbf A\compo\mathbf B$ (\autoref{prop:dualPreC}).

  \item 
  In the $\otimes$--case 
  we reason similarly to the $\parallel$ case.

  \item 
  For the $\parr$--case 
  we can reason by duality 
  using the $\otimes$--case.
  
\end{enumerate}
\end{proof}

\begin{proposition}[\autoref{rem:InterpretMLL}]
  Given a formula $A$
  and a basis $\ibase$,
  we have $\reali[\ibase]{A^\bot}\subseteq\reali[\ibase]{A}^\bot$.
\end{proposition}
\begin{proof}
  By induction on the formula.
  If $A=X$ is an atomic formula this is trivial.
  If $A = B\otimes C$
  then:
  \begin{align*}
    \reali[\ibase]{(B\otimes C) ^\bot} 
    & = \reali[\ibase]{B^\bot \parr C^\bot} \\
    & = \quad 
    \reali[\ibase]{B^\bot}\parr \reali[\ibase]{C^\bot} \\ 
    & \subseteq \quad 
    \reali[\ibase]{B}^\bot\parr \reali[\ibase]{C}^\bot 
    & (\mbox{\autoref{prop:inclusionPreserve} and induction hypothesis})\\ 
    &=\quad 
    (\reali[\ibase]{B}\otimes \reali[\ibase]{C} )^\bot
    & (\mbox{\autoref{prop:dualC}})\\ 
    &=\quad 
    (\reali[\ibase]{B\otimes C} )^\bot
    & (\mbox{\autoref{prop:dualC}})
  \end{align*}

  Similarly, if $A=B\parr C$:
  \begin{align*}
    \reali[\ibase]{(B\parr C) ^\bot} 
    & = \reali[\ibase]{B^\bot \otimes C^\bot} \\
    & = \quad 
    \reali[\ibase]{B^\bot}\otimes \reali[\ibase]{C^\bot} \\ 
    & \subseteq \quad 
    \reali[\ibase]{B}^\bot\otimes \reali[\ibase]{C}^\bot 
    & (\mbox{\autoref{prop:inclusionPreserve} and induction hypothesis})\\ 
    &=\quad 
    (\reali[\ibase]{B}\parr \reali[\ibase]{C} )^\bot
    & (\mbox{\autoref{prop:dualC}})\\ 
    &=\quad 
    (\reali[\ibase]{B\parr C} )^\bot
    & (\mbox{\autoref{prop:dualC}})
  \end{align*}
\end{proof}

\subsection{On approximable basis}

    \begin{remark}\label{rem:interCommutes}
      Given two nets $S$ and $T$
      the net
      $S::T$ is equal to the net $T::S$,
      that is, up to ismorphism of the cut links 
      mapping each $\cutlink{S(i),T(i)}$
      to a cut link $\cutlink{T(i),S(i)}$.
    \end{remark}

    \begin{proposition}\label{prop:approxBase}
      Given $\ibase$ an approximable basis 
      for any formula $A$ of $\mll$
      the type $\reali[\ibase]{A}$
      is approximable.
    \end{proposition}
    \begin{proof}
      By induction on $A$.
      In the case where $A$ is atomic this follows from the definition.
      \begin{itemize}
        \item For a formula $A\otimes B$. 
        Let us show $\maltese_1$ belongs to $\reali[\ibase]{A\otimes B}$ 
        by showing that 
        $\maltese_1$
        is orthogonal to $\reali[\ibase]{A\otimes B} ^\bot$.

        The interpretation $\reali[\ibase]{A\otimes B} ^\bot$
        is equal to $(\reali[\ibase]{A}\otimes\reali[\ibase]{B})^\bot$
        and so to $\reali[\ibase]{A}^\bot\parr\reali[\ibase]{B}^\bot$.
        To be orthogonal to $\reali[\ibase]{A}^\bot\parr\reali[\ibase]{B}^\bot $
        is to be orthogonal 
        to the nets $S+\parrlink{p_1,p_2}{p}$
        where $S$ belongs to $\reali[\ibase]{A}^\bot \compo \reali[\ibase]{B}^\bot$.
        Consider $S$ such a net, then:
        \begin{align*}
          & \maltese_1 :: S + \parrlink{p_1,p_2}{p} & \\
          \rightarrow \quad &
          \maltese_1 \parallel \maltese_1 :: S 
            & \mbox{(\autoref{def:nohomcut})} \\
          = \quad &
          S:: (\maltese_1 \parallel \maltese_1 ) 
          & \mbox{(\autoref{rem:interCommutes})} \\
          = \quad &
          (S:: \maltese_1) :: \maltese_1  
          & \mbox{(\autoref{prop:interAsAction})} 
        \end{align*}
        
        By induction
        $\maltese_1$ belongs to $\reali[\ibase]{A}^{\bot\bot}$
        and $S$ belongs to $\reali[\ibase]{A}^\bot \compo \reali[\ibase]{B}^\bot$
        we conclude that
        $S::\maltese_1$ belongs to $\reali[\ibase]{B^\bot}$.
        Again by induction 
        $\maltese_1$ belongs to 
        $\reali[\ibase]{B}^{\bibot}$
        and so in particular it belongs to 
        in $\reali[\ibase]{B^\bot}^\bot$.
        To conclude: 
        we have shown that 
        $S::\maltese_1$ belongs to $\reali[\ibase]{B^\bot}$,
        hence it is orthogonal to $\reali[\ibase]{B^\bot}^\bot$
        so in particular 
        it is orthogonal to $\maltese_1$.
        As a consequence $S::\maltese_1::\maltese_1$
        reduces to $\maltese_0$.

        \item Now consider a formula of the form $A\parr B$.
        Again let us show that 
        $\maltese_1$
        is orthogonal to $\reali[\ibase]{A\parr B} ^\bot$
        i.e. to $\reali[\ibase] A ^\bot \otimes \reali[\ibase]{B}^\bot$.
        To be orthogonal to that type 
        is to be orthogonal to the nets 
        $\overline a \parallel \overline b +\tenslink{p_1,p_2}{p}$.
        Consider such a net, then:
        \begin{align*}
          & \maltese_1 :: \overline a \parallel \overline b +\tenslink{p_1,p_2}{p} & \\
          \rightarrow \quad &
          \maltese_2 :: \overline a \parallel \overline b
            & \mbox{(\autoref{def:nohomcut})} \\
          = \quad &
          \maltese_1\merge\maltese_1 :: \overline a \parallel \overline b
            & \mbox{(\autoref{def:mergeDaimon})} 
        \end{align*}
        By induction both
        $\reali[\ibase] A$ and $\reali[\ibase] B$ contains $\maltese_1$
        thus $\reali[\ibase] A \compo \reali[\ibase] B$
        contains $\maltese_2$ (\autoref{prop:mergeInCompo}).
        Equivalently this means that $\maltese_2$
        is orthogonal to $\reali[\ibase] A^\bot \parallel \reali[\ibase] B^\bot$,
        hence $\maltese_2 \perp \overline a \parallel \overline b$.
        This concludes to show that $\maltese_1$ belongs to 
        $\reali[\ibase]{A\parr B}$.
      \end{itemize}
    \end{proof}

    \subsection{Adequacy -- Proof of \autoref{thm:adequacy}}

    \adequacy*
    \begin{proof} 
      We proceed by induction 
      on the proof $\pi$
      which is represented by the net $S$.

      \emph{Base case.}
      In the base the proof $\pi$ consists of 
      a proof tree using a single $\maltese$--rule introducing 
      any sequent $\Gamma = A_1,\dots,A_n$ of length $n$.
      Then any net $S$ representing $\pi$
      is isomorphic to $\maltese_n$,
      we must therefore show that $\maltese_n$
      belongs to $\reali[\ibase]{\Gamma}$.

      Because 
      $\ibase$ is approximable 
      any interpretation $\reali[\ibase]{A}$
      is approximable (\autoref{prop:approxBase}): for every $1\leq i\leq n$ we have $\maltese_{1}\in\reali[\ibase]{A_{i}}$ i.e.\ $u\bot\maltese_{1}$ for every $u\in\reali[\ibase]{A_{i}}^{\bot}$.
We want to show that 
      $\maltese_n \in \reali[\ibase]{A_1,\dots,A_n}$
      or equivalently that 
      $\maltese_n$ is orthogonal to $\reali[\ibase]{A_1,\dots,A_n}^\bot=(\reali[\ibase]{A_1}\compo\ldots\compo\reali[\ibase]{A_n})^{\bot}=\reali[\ibase]{A_1}^\bot\parallel\ldots\parallel\reali[\ibase]{A_n}^\bot$.
      By density (Remark~\ref{rem:density}) to prove $\maltese_n\bot\ \reali[\ibase]{A_1}^\bot\parallel\ldots\parallel\reali[\ibase]{A_n}^\bot$ it is enough to prove that $u_1\parallel \dots \parallel u_n\bot\maltese_{n}$
      for every $u_i\in\reali[\ibase]{A_i}^\bot$. We proceed by induction on $n$: 
      \begin{align*}
        & \maltese_n :: (u_1\parallel \dots \parallel u_n) & \\
        =\quad 
        & (\ldots(\maltese_n ::u_1) ::\dots)::u_n 
        & (\mbox{\autoref{prop:interAsAction}}) \\
        \rightarrow^*\quad 
        & (\ldots(\maltese_{n-1} :: u_2) ::\dots)::u_n 
        & (\mbox{\autoref{prop:addConcluReduction} and $u_{1}\bot\maltese_{1}$}) \\
        =\quad 
        & \maltese_{n-1} :: (u_2\parallel\dots\parallel u_n) 
        & (\mbox{\autoref{prop:interAsAction}}) \\
        \rightarrow^*\quad & \maltese_{0} 
        & (\mbox{Induction hypothesis}) 
      \end{align*}
      As a consequence
      $\maltese_n$
      is contained in the type 
      $\reali[\ibase]{A_1,\dots,A_n}$,
      i.e. $\maltese_n$
      realises the sequent $\Gamma$.

      To conclude observe that 
      any proof $\pi$ of conclusion $\Gamma=A_1,\dots,A_n$ 
      which uses only one inference rule 
      in $\mlldai$, must use a daimon inference rule.
      Then any such proof $\pi$ 
      is mapped by $\deseq\pi$
      to the daimon link 
      with $n$ outputs $\maltese_n$.
      Hence we have showed the base case of the induction,
      where the induction is 
      performed on the size (its number of inference rules) of the proof $\pi$.

        \emph{Inductive cases.}
        We look at the last rule of $\pi$
        which may be a $\otimes$,$\parr$
        or $\mathsf{cut}$--rule:
        \begin{itemize}[leftmargin=35pt]
          \item [\small \scshape ind$\cdot 1$]
          Assume that the last rule in the
          represented proof is a $\parr$--rule with main conclusion $A\parr B$.
          Thus the sequent is of the form $\Gamma, A\parr B$,
          and by assumption
          $S\vdash_{\mlldai} \Gamma, A\parr B$.
          Say the conclusions of $S$
          are ordered as $q_1,\dots,q_n,p$
          then since $p$ is given the type $A\parr B$ the net $S$ is of the form
          $S' + \parrlink{p_1,p_2}{p}$  where  $S'\vdash_{\mlldai} \Gamma,A,B$.
        
          Calling the induction hypothesis we can deduce $S'\Vdash_\ibase\Gamma , A , B$.
          By definition, this menas that 
          for any $\gamma$ in $\reali[\ibase]{\Gamma}^\bot$,
          we have
          $S' :: \gamma \Vdash_\ibase A,B$, and thus
          $S'::\gamma + \parrlink{p_1,p_2}{p}
          \Vdash_\ibase A\parr B$.
        
          Furthermore interaction and sum commute (\autoref{prop:commutationInterSumlink}) 
          hence,
          $S'::\gamma + \parrlink{p_1,p_2}{p} = S'+ \parrlink{p_1,p_2}{p} ::\gamma $ so that $S::\gamma\Vdash_\ibase A\parr B$ for any $\gamma \in \reali[\ibase]{\Gamma}^\bot$:
          this allows us to conclude that $S\Vdash_\ibase \Gamma, A\parr B$.
        
          \item [\small \scshape ind$\cdot 2$]
          Assume that the last rule in the
          represented proof is a $\otimes$--rule with main conclusion $A\otimes B$.
          Thus the sequent is of the form $\Gamma, \Delta, A\otimes B$,
          and $S$ is of the form $S_1\otimes S_2$
          where 
          $S_1\vdash_{\mlldai} \Gamma,A$
          and
          $S_2\vdash_{\mlldai} \Delta , B$.
                  
          Calling the induction
          hypothesis
          we obtain
          $S_1\Vdash_\ibase \Gamma,A$
          and
          $S_2\Vdash_\ibase \Delta , B$.
          Thus for any $\gamma\in\reali[\ibase]{\Gamma}^\bot$
          and $\delta\in\reali[\ibase]{\Delta}^\bot$
          $S_1:: \gamma\Vdash_\ibase A$
          and
          $S_2:: \delta\Vdash_\ibase B$.
          In particular this means that
          the tensor union of the two nets
          $S_1:: \gamma +
          S_2:: \delta  + \tenslink{p_1,p_2}{p}$
          is in $\reali[\ibase]{A}\otimes \reali[\ibase]{B}=\reali[\ibase]{A\otimes B}$.
        
          Note that
          $S_1:: \gamma +
          S_2:: \delta  + \tenslink{p_1,p_2}{p}=S_1+ S_2 + \tenslink{p_1,p_2}{p} ::\gamma\parallel\delta$ (\autoref{prop:interAsAction} and \autoref{prop:commutationInterSumlink}).
          Thus $S_1+ S_2 + \tenslink{p_1,p_2}{p} ::\gamma\parallel \delta \Vdash_\ibase A\otimes B$.
          Since this hold for any $\gamma\in\reali[\ibase]{\Gamma}^\bot$
          and $\delta\in\reali[\ibase]{\Delta}^\bot$
	 we conclude that 
          $S_1+ S_2 + \tenslink{p_1,p_2}{p} \Vdash_\ibase \Gamma ,\Delta , A\otimes B$
          i.e. 
          $S \Vdash_\ibase \Gamma ,\Delta , A\otimes B$.

          \item[\scshape{ind}$\cdot 3$]
          Assume that the last rule in $\pi$
          is a cut rule 
          between two subproofs $\pi_1$ and $\pi_2$
          of respective conclusions $\Gamma, A$
          and $\Delta, A^\bot $.
          Then a net $S$ which represents 
          $\pi$ is of the form 
          $S = S_1 + S_2 + \cutlink{S_1(\#S_1),S_2(\# S_2)}$
          where
          $S_1$ (resp. $S_2$)
          is a representation of the proof $\pi_1$ (resp. $\pi_2$).

          Applying the induction hypothesis 
          we derive 
          $S_1$ belongs to $\reali[\ibase]{\Gamma,A}$
          and $S_2$ belongs to $\reali[\ibase]{\Delta}$.
          Hence for any $u\in\reali[\ibase]{\Gamma}^\bot$
          and $v\in\reali[\ibase]{\Delta}^\bot$
          we have 
          $S_1:: u \in \reali[\ibase]{A}$
          and 
          $S_2::v \in \reali[\ibase]{A^\bot}$.
          Because 
          $\reali[\ibase]{A^\bot}\subset \reali[\ibase]{A}^\bot$ (Remark~\ref{rem:InterpretMLL})
          it follows that 
          $(S_1::u) \bot (S_2::v)$.
          Let us rewrite this net to conclude:
          \begin{align*}
            &S_1::u \quad :: \quad S_2::v \\
        =\quad &
        S_1::u \quad +\cutlink{(S_1::u)(1),(S_2::v)(1)}+ \quad S_2::v 
        & \mbox{(\autoref{def:interaction})} \\
        =\quad &
        S_1::u \quad +\cutlink{S_1(\# S_1),S_2(\#S_2)}+ \quad S_2::v 
        & \mbox{(Identity)} \\
            =\quad & 
            (S_1 +\cutlink{S_1(\# S_1),S_2(\#S_2)}) :: u + S_2::v 
            & \mbox{(\autoref{prop:commutationInterSumlink})} \\
            =\quad & 
            (S_1 +\cutlink{S_1(\# S_1),S_2(\#S_2)} + S_2) :: u :: v 
            & \mbox{(\autoref{prop:commutationInterSumlink})} \\
            =\quad & 
            S:: u :: v 
            & (\mbox{Identify } S)\\
            =\quad & 
            S:: (u \parallel v )
            & (\mbox{\autoref{prop:interAsAction}})\\
          \end{align*}

          This shows that $S$
          is orthogonal to any 
          $u \parallel v$
          when $u\in\reali[\ibase]{\Gamma}^\bot$
          and $v\in\reali[\ibase]{\Delta}^\bot$.
          In other words 
          $S$ is orthogonal to 
          $\reali[\ibase]{\Gamma}^\bot \parallel^{-} \reali[\ibase]{\Delta}^\bot$
          and therefore to
          $\reali[\ibase]{\Gamma}^\bot \parallel \reali[\ibase]{\Delta}^\bot$ (\autoref{rem:density})
          thus by duality it belongs to 
          $\reali[\ibase]{\Gamma} \compo \reali[\ibase]{\Delta}$ (\autoref{prop:dualPreC})
          i.e. 
          by definition to 
          $\reali[\ibase]{\Gamma,\Delta} $.
        \end{itemize}
    \end{proof}

    \section{Complements to \autoref{sect:Tests}}
    
\subsection{On orthogonality of Paritions}

    \subsubsection{Orthogonal Partitions and bijections}

    \begin{definition}[image of a partition]
      Given a partition $P=\{C_1,\dots,C_n\}$ of a set $X$
      and a function $f:X\rightarrow Y$
      the \emph{image} 
      of $P$ by $f$ 
      is the set $\{ f(C_1),\dots,f(C_n)\}$, 
      it is denoted $f(P)$.
    \end{definition} 

    \begin{proposition}[Bijections preserve partitions]
      Given $f:X\rightarrow Y$ a bijection between two sets.
      For any partition $P$ of $X$
      the image $f(P)$ is a partition of $Y$.
    \end{proposition}

\begin{definition}
  An isomoprhisms 
  $f : G\rightarrow H$ between two undirected multigraph 
  $G=(V_G,E_G,\border[G])$
  and $H=(V_H,E_H,\border[H])$
  is a pair of functions $(f_V,f_E)$ such that:
  \begin{itemize}
    \item 
    $f_V : V_G\rightarrow V_H$ is a bijection.
    \item 
    $f_E : E_G\rightarrow E_H$ is a bijection.
    \item 
    For any edge $e$ of $E_G$
    we have that $f_V(\border[G]{e})$
    equals $\border[H]{f(e)}$.
  \end{itemize}
\end{definition}

    \begin{proposition}
      Given $f:X\rightarrow Y$ a bijection 
      $P$ and $Q$ two partitions of $X$.
      The assertions are equivalent:
      \begin{itemize}
        \item 
        The partitions $P$ and $Q$ are orthogonal.
        \item 
        The partitions $f(P)$ and $f(Q)$ are orthogonal.
      \end{itemize}
    \end{proposition}
    \begin{proof}
      We show that the induced graph 
      by $P$ and $Q$ 
      and by $f(P)$ and $f(Q)$ are isomorphic,
      i.e. that $G_1 ::= \mathsf G(P,Q)$ and $G_2::=\mathsf{G(f(P),f(Q))}$
      are isomorphic.
      The bijection between the set of edges is given by $f:X\rightarrow Y$
      while the bijection between the vertex--sets 
      ${V}_{G_1}$ and ${V}_{G_2}$
      is given by the map 
      associating with a subset of $X$ its image under $f$ 
      i.e. $F:\{x_1,\dots,x_n\}\rightarrow \{f(x_1),\dots,f(x_n)\}$.

      Finally, one can check that the 
      border function 
      are coherent: for $A\in P$
      the image $f(\border[\mathsf G(P,Q)](A))$
      is $\border[\mathsf G (f(P),f(Q)) ](f(A))$.
    \end{proof}

  \subsubsection{Orthogonal Partitions and their representations as nets}
  
  \begin{definition}
    A \emph{natural partition} is a partition 
    of a subset of $\mathbb N$.
    A natural partition of size $n$ 
    is a partition of the set of integers $\{1,\dots, n \}$.
  \end{definition}

  \begin{remark}
    Given a natural partition $P$ of size $n$,
    any partition to $Q$ which is orthogonal to $P$
    is also a natural partition of size $n$.
  \end{remark}

  \begin{definition}\label{def:partitionRepr}
    A net $S=(\body S, \arrange S)$ 
    containing only daimon links
    \emph{represents a natural partition} $P$ of size $n$
    (which we denote $S\equiv P$)
    whenever:
    \begin{itemize}
      \item 
      For each class $\{ i_1,\dots,i_k \}$ in $P$
      there exists a daimon link 
      $\dailink{S(i_1),\dots,S(i_k)}$ in $S$.
      \item 
      For each daimon link $\dailink{S(a_1),\dots,S(a_m)}$ 
      contained in $S$
      the set of integers $\{a_1,\dots,a_n\}$
      is a class of the partition $P$.
    \end{itemize}
  \end{definition}

  \begin{remark}
    Given a natural partition $P$ of size $n$
    there exists (infinitely) many nets $S$
    which represent $P$,
    however all these nets are isomorphic.
  \end{remark}

  \begin{notation}
    Given a natural partition $P$ of size $n$
    we abusively denote by $P^\maltese$
    a representation of $P$.
  \end{notation}
  
  \begin{proposition}
    Given $P$ and $Q$ 
    two natural partition of size $n$,
    and two nets $S$ and $T$ 
    such that 
    $S\equiv P$ while $T\equiv Q$,
    the assertions are equivalent:
    \begin{itemize}
      \item 
      The partitions $P$ and $Q$ are orthogonal.
      \item 
      The nets $S$ and $T$ are orthogonal.
    \end{itemize}
  \end{proposition}
  \begin{proof}
    Observe that $\mathsf G(P,Q)$ is acyclic and connected 
    if and only if $S::T$ is acyclic and connected.
    Then we can prove each implication:
    \begin{itemize}[leftmargin=40pt]
      \item[$1\Rightarrow 2.$]
      Assuming that $\mathsf G(P,Q)$ 
      then $S::T$ is acyclic and connected.
      Because $S$ and $T$ contain only daimon links,  
      $S::T$ contains only glueing cuts.
      Furthermore the elimination of glueing cuts preserve 
      acyclicity and connectedness
      while $S::T$ is a net without conclusion 
      (and cut--elimination preserve the conclusion of a net).
      The only cut free, acyclic--connected
      net with no conclusions is $\maltese_0$
      thus we conclude $S::T\rightarrow \maltese_0$
      meaning that $S$ and $T$ are orthogonal.
      \item[$2\Rightarrow 2.$]  
      Assume that $S::T\rightarrow \maltese_0$
      note that the reduction step of glueing cuts 
      preserve acyclicity and connectedness 
      \emph{in both directions},
      as a consequence since $S::T$ reduces to $\maltese_0$
      by eliminating only glueing cuts 
      and since $\maltese_0$ is acyclic and connected 
      we conclude that $S::T$ is also acyclic and connected.
      This means equivalently that $\mathsf G(P,Q)$ is acyclic and connected.
    \end{itemize}
  \end{proof}



  \subsection{Proving \autoref{prop:testsOrtho}}

  \subsubsection{Elimination of Multiplicative cuts and addresses}

\begin{definition}
  An \emph{address} is a sequence 
  of elements of $\{ \pleft , \pright\}$.
  The element at position $\pleft$ (resp. $\pright$)
  relatively to $p$ in a net $S$
  is as follows:
  \begin{center}
    \scalebox{0.8}{
      \begin{tabular}{r @{\hskip 2pt } c @{\hskip 2pt }l}
        $\findpos  {\pleft }  $ 
          & $\triangleq $ &
          $\left\{
            \begin{array}{c}
                p_1 \mbox{ if }  p \in \link[\square]{p_1,p_2}{p} \mbox{ belongs to } S \\
                \mbox{undefined otherwise.}
             \end{array}
        \right.$.\\
        $\findpos  {\pright }  $ 
          & $\triangleq $ &
          $\left\{
            \begin{array}{c}
                p_2 \mbox{ if }  p \in \link[\square]{p_1,p_2}{p} \mbox{ belongs to } S \\
                \mbox{undefined otherwise.}
             \end{array}
        \right.$.
    \end{tabular}}
  \end{center}

    The element at address $\xi$ 
  relatively to $p$ in a net $S$
  is defined inductively:
    $$\findpos  {\epsilon } = p   
    \quad ; \quad 
    \findpos {\pleft \at \xi} = 
    \findpos [\findpos{\pleft} ]{\xi}
    \quad ; \quad 
    \findpos {\pright \at \xi} = 
    \findpos [\findpos{\pright}] {\xi}$$
\end{definition}

\begin{remark}
  For each position $p$
  of a net $S = (\body S, \arrange S)$
  there exists an address $\xi$
  and an integer $i$ 
  so that 
  $p = \findpos [S(i)] {\xi}$ 
  (This can be shown by induction on the number of links in $S$).
  For a position 
  $p$ we denote that index $i$
  by $\rootof p$
  and $\xi$ by $\adrof p$.
\end{remark}

\begin{remark}
  Adresses are ordered 
  in the following way: 
  we set $\pleft \leq \pright$
  and lift to adresses using the lexicographical order.
\end{remark}

  \begin{definition}[cut--free initial order]
    Given a cut--free net $S=(\body S , \arrange{ S})$
    we can order the initial positions of $S$:
      A position $p$ is smaller than $q$ 
      if 
      $\rootof p < \rootof q$
      or if 
      $\rootof p =\rootof q$
      while $\adrof p \leq \adrof q$
      (for the lexicographical order).
    This induced order is denoted $\leq ^{\arrange S}$.

    Given a conclusion $p$ of a net $S$
    an address $\xi$ is \emph{defined} for $p$ 
    is $\findpos [p] \xi$ is defined in $S$.
    An address is \emph{maximal} for a position $p$ 
    if $\xi$ is defined for $p$ 
    and no address that is defined for $p$ 
    has $\xi$ for strict prefix.
    We denote $\maxadr p$
    the maximal addresses of a position $p$.
  \end{definition}

  \begin{proposition}
    Given a net $S$, $p$ a conclusion of $S$
    and $\xi$ an address.
    The assertions are equivalent:
    \begin{itemize}
      \item 
      $\xi$ is maximal for $p$.
      \item 
      $\findpos [p] \xi$ is an initial position of $S$.
    \end{itemize}
  \end{proposition}

  \begin{corollary}
    Given a net $S = (\body S, \arrange S)$ 
    the set of initial positions of $S$
    is equal to 
    $$\bigcup_{1\leq i \leq \# S} \findpos[S(i)]{\maxadr {S(i)}}.$$
  \end{corollary}

  \begin{notation}
    Given a position $p$ in a net $S$
    we denote 
    $\connects p$
    the set of connective links above $p$ 
    in $S$.

    Given an (unordered) net 
    $S = (V , E, \source,\target)$ and a set of link $E_0$ of $S$
    we denote 
    by $S\setminus E_0$ the net 
    $ (V\cap \bigcup_{e\in E\setminus E_0} \target(e)\cup\source(e) , E\setminus E_0 , \source,\target)$. 
  \end{notation}

  \begin{proposition}\label{prop:orthoDaiDescription}
    Given $S$ and $T$ 
    two nets such that 
    $S\witness A_1,\dots,A_n$
    while $T\witness A_1^\bot,\dots,A_n^\bot$
    then by eliminating only multiplicative cuts
    the interaction $S::T$ reduces 
    to 
    $$S^\maltese + T^\maltese + 
    \sum_{1\leq i \leq n}\sum_{\xi \in \{\pleft,\pright\}^{N_i}} 
    \cutlink{\findpos[S(i)]{\xi}_S , \findpos[T(i)]{\xi}_T}$$
  \end{proposition}
  \begin{proof}
    Assume that 
    $S$ and $T$ contain connective links 
    (i.e. $\otimes$-- and $\parr$--links)
    and more precisely, 
    without loss of generality assume that their first conclusion 
    is the output of a connective link.
    As a consequence 
    since $S\witness A_1,\dots,A_n$ and $T\witness A_1^\bot,\dots, A_n^\bot$ witness dual sequents 
    it follows that $A_1 = B_1\otimes B_2$
    while $A_1^\bot = B_1\parr B_2$.

    First note that 
    $S::T$
    is equal to 
    $$S + T + \sum_{1\leq i \leq n} \cutlink{S(i),T(i)} =
    S+T + \cutlink{S(1), T(1)}+ \sum_{2\leq i \leq n} \cutlink{S(i),T(i)}
    $$
    Furthermore since 
    $A_1 = B_1\otimes B_2$
    while $A_1^\bot = B_1\parr B_2$
    we can ensure that 
    a link $\tenslink{p_1,p_2}{S(1)}$
    occurs in $S$ and is terminal 
    while a link 
    $\parrlink{q_1,q_2}{S(1)}$
    occurs in $T$ and is terminal.

    As a consequence $\cutlink{S(1), T(1)}$ 
    is a multiplicative cut in $S::T$
    and its elimination produces the following reduction:

    \scalebox{0.85}{\begin{minipage}{\linewidth}
      \begin{align*}
      & S+T + \cutlink{S(1), T(1)} \\
       \rightarrow \; &
      S\setminus\tenslink{p_1,p_2}{S(1)} 
      + 
      T\setminus\parrlink{q_1,q_2}{T(1)} 
      +
      \cutlink{p_1,q_1}
      + 
      \cutlink{p_2,q_2} \\
      = \; &
      S\setminus\tenslink{p_1,p_2}{S(1)} 
      + 
      T\setminus\parrlink{q_1,q_2}{T(1)} \\
      &\quad +
      \cutlink{\findpos[S(1)]{\pleft},\findpos[T(1)]{\pleft}}
      + 
      \cutlink{\findpos[S(1)]{\pright},\findpos[T(1)]{\pright}}
    \end{align*}       
    \end{minipage}}

    By contextual closure of the cut elimination procedure 
    we derive the following reduction:
    \newline 

    \scalebox{0.85}{\begin{minipage}{\linewidth}    
    \begin{align*}
      S::T &= S + T + \sum_{1\leq i \leq n} \cutlink{S(i),T(i)}\\
       &=
      S+T + \cutlink{S(1), T(1)}+ \sum_{2\leq i \leq n} \cutlink{S(i),T(i)}\\
      &\rightarrow
      S\setminus\tenslink{p_1,p_2}{S(1)} 
      + 
      T\setminus\parrlink{q_1,q_2}{T(1)}  \\
      & \quad +
      \cutlink{\findpos[S(1)]{\pleft},\findpos[T(1)]{\pleft}}
      + 
      \cutlink{\findpos[S(1)]{\pright},\findpos[T(1)]{\pright}}
      + 
      \sum_{2\leq i \leq n} \cutlink{S(i),T(i)}
    \end{align*}
    \end{minipage}}

    Now consider the two nets 
    $S_0 = (S\setminus\tenslink{p_1,p_2}{S(1)}  ,
       \findpos[S(1)]{\pleft} \at \findpos[S(1)]{\pright} \at \arrange{S})$
    and 
    $T_0 = (T\setminus\parrlink{q_1,q_2}{T(1)}  ,
      \findpos[T(1)]{\pleft} \at \findpos[T(1)]{\pright} \at \arrange{2})$.
    Observe that 
    because 
    $S\witness A_1,\dots, A_n$
    it follows that 
    $S_0\witness B_1,B_2,A_2,\dots,A_n$
    while because
    $T\witness A_1^\bot,\dots, A_n^\bot$
    it follows that 
    $T_0\witness B_1^\bot,B_2^\bot,A_2^\bot,\dots,A_n^\bot$.
    Because $S_0$ and $T_0$ have less connective links than $S$ and $T$ 
    we can apply the induction hypothesis.
    Furthermore we observe that 
 
    \scalebox{0.8}{\begin{minipage}{\linewidth}
      \begin{align*}
      & S\setminus\tenslink{p_1,p_2}{S(1)} 
      + 
      T\setminus\parrlink{q_1,q_2}{T(1)} \\
      & \quad +
      \cutlink{\findpos[S(1)]{\pleft},\findpos[T(1)]{\pleft}}
      + 
      \cutlink{\findpos[S(1)]{\pright},\findpos[T(1)]{\pright}}
      + 
      \sum_{2\leq i \leq n} \cutlink{S(i),T(i)} \\
       = \;&
      S_0 
      + 
      T_0 
      +
      \cutlink{\findpos[S(1)]{\pleft},\findpos[T(1)]{\pleft}}
      + 
      \cutlink{\findpos[S(1)]{\pright},\findpos[T(1)]{\pright}}
      + 
      \sum_{2\leq i \leq n} \cutlink{S(i),T(i)} \\
      = \;&
      S_0 
      + 
      T_0 
      +
      \cutlink{S_0(1),T_0(1)}
      + 
      \cutlink{S_0(2),T_0(2)}
      + 
      \sum_{2\leq i \leq n} \cutlink{S(i),T(i)} \\
      = \; &
      S_0 
      + 
      T_0 
      +
      \cutlink{S_0(1),T_0(1)}
      + 
      \cutlink{S_0(2),T_0(2)}
      + 
      \sum_{3\leq i \leq n+1} \cutlink{{S_0}(i),{T_0}(i)} \\
      = \; &
      S_0 
      + 
      T_0 
      + 
      \sum_{1\leq i \leq n+1} \cutlink{{S_0}(i),{T_0}(i)} \\
      = \; &
      S_0::T_0
    \end{align*}
  \end{minipage}
    }

    By calling the induction hypothesis on $S_0::T_0$
    we derive the reduction:
    $$S_0::T_0 \rightarrow^* 
    S_0^\maltese + T_0^\maltese + 
    \sum_{1\leq i \leq \#S_0}\sum_{\xi \in \maxadr{S_0(i)}} 
    \cutlink{\findpos[S_0(i)]{\xi}_S , \findpos[T_0(i)]{\xi}_T}
    $$

    First observe 
    $S_0^\maltese = S^\maltese$
    and $T_0^\maltese = T^\maltese$.
    Now note that 
    any for any $3\leq i \leq \#S_0$
    we have $S_0(i) = S(i-1)$
    while $T_0(i) = T(i-1)$
    thus the sum of cut links can be rewritten: 
    \newline 

    \scalebox{0.73}{
      \begin{minipage}{\linewidth}
      \begin{align*}
     & \sum_{1\leq i \leq \#S_0}\sum_{\xi \in \maxadr{S_0(i)}} 
    \cutlink{\findpos[S_0(i)]{\xi}_{S_0} , \findpos[T_0(i)]{\xi}_{T_0}}\\
    =&
    \sum_{1\leq i \leq 2}\sum_{\xi \in \maxadr{S_0(i)}} 
    \cutlink{\findpos[S_0(i)]{\xi}_{S_0} , \findpos[T_0(i)]{\xi}_{T_0}}
    +
    \sum_{3\leq i \leq \#S_0}\sum_{\xi \in \maxadr{S_0(i)}} 
    \cutlink{\findpos[S_0(i)]{\xi}_{S_0} , \findpos[T_0(i)]{\xi}_{T_0}}\\
    =&
    \sum_{1\leq i \leq 2}\sum_{\xi \in \maxadr{S_0(i)}} 
    \cutlink{\findpos[S_0(i)]{\xi}_{S_0} , \findpos[T_0(i)]{\xi}_{T_0}}
    +
    \sum_{3\leq i \leq \#S_0}\sum_{\xi \in \maxadr{S(i-1)}} 
    \cutlink{\findpos[S(i-1)]{\xi}_S , \findpos[T(i-1)]{\xi}_T} \\
    =&
    \sum_{1\leq i \leq 2}\sum_{\xi \in \maxadr{S_0(i)}} 
    \cutlink{\findpos[S_0(i)]{\xi}_{S_0} , \findpos[T_0(i)]{\xi}_{T_0}}
    +
    \sum_{2\leq i \leq \#S}\sum_{\xi \in \maxadr{S(i)}} 
    \cutlink{\findpos[S(i)]{\xi}_S , \findpos[T(i)]{\xi}_T}
    \end{align*}
  \end{minipage}
    }

    To conclude 
    observe that 
    $\maxadr{S_0(i)}$
    is equal 
    to $\{  s \mid \exists \mathsf i \in \{\pleft,\pright\}
    \mathsf i \at s \in \maxadr{S(i)} \}$
    Furthermore 
    $\findpos[T(1)]{\pleft}$ 
    is equal to 
    $T_0(1)$.
    Thus we derive the following:
    \begin{align*}
       &\sum_{1\leq i \leq 2}\;\;
      \sum_{\xi \in \maxadr{S_0(i)}} 
        \cutlink{
          \findpos[S_0(i)]{\xi}_{S_0} , 
          \findpos[T_0(i)]{\xi}_{T_0}
          } \\
     =& 
    \sum_{\mathsf i \in \{\pleft,\pright\}}\;\;
    \sum_{\xi \in \maxadr{\findpos[S(1)]{\mathsf i}_S}} 
      \cutlink{
        \findpos[X]{\xi}_{S_0} ,
        \findpos[X]{\xi}_{T_0}
       }  \\ 
     =&
     \sum_{\mathsf i \in \{\pleft,\pright\}} \;\;
     \sum_{\xi \in \maxadr{\findpos[S(1)]{\mathsf i}_S}} 
      \cutlink{
       \findpos[S(1)]{\mathsf i\xi}_{S} ,
       \findpos[T(1)]{\mathsf i\xi}_{T}
       }  \\ 
     =&
    \sum_{\xi \in \maxadr{S(1)}} 
      \cutlink{
       \findpos[S(1)]{\xi}_{S} ,
       \findpos[T(1)]{\xi}_{T}
       }
    \end{align*}

    Putting all together the sum of cut links can be written as:
    
  \scalebox{0.8}{
    \begin{minipage}{\linewidth}      
    \begin{align*}
    &\sum_{\xi \in \maxadr{S(1)}} 
     \cutlink{\findpos[S(1)]
     {\xi}_{S} ,
      \findpos[T(1)]
      {\xi}_{T}
      } 
      +
      \sum_{2\leq i \leq \#S}\;\;
      \sum_{\xi \in \maxadr{S(i)}} 
    \cutlink{\findpos[S(i)]{\xi}_S , \findpos[T(i)]{\xi}_T} \\
    = &
    \sum_{1\leq i \leq \#S} \;\; 
    \sum_{\xi \in \maxadr{S(i)}} 
    \cutlink{\findpos[S(i)]{\xi}_S , \findpos[T(i)]{\xi}_T}
  \end{align*}  
\end{minipage}}

    As a consequence we conclude that 
    $S_0::T_0$ reduces to 
    $$
    S^\maltese + T^\maltese + 
    \sum_{1\leq i \leq \#S}\;\;
    \sum_{\xi \in \maxadr{S(i)}} 
    \cutlink{\findpos[S(i)]{\xi}_S , \findpos[T(i)]{\xi}_T}
    $$

    Finally, because $S::T \rightarrow S_0::T_0$
    we conclude.
  \end{proof}

  \subsubsection{Proof of \autoref{prop:testsOrtho}}

  \begin{remark}
    Given a net $S = (\body S,\arrange S)$
    the net induced by its daimon $S^\maltese$
    comes with the natural order induced 
    by $\arrange S$ and the lexicographical order on address.
    More specifically 
    given $p$ and $q$ two points of $S^\maltese$
    both can be written as $\findpos[a]{\xi}$
    where $a=\rootof p _S$ is a conclusion of $S$ 
    while $\xi = \adrof p _S$.
    Then $p\leq q$ 
    if and only if $(\rootof p , \adrof p) \leq (\rootof q, \adrof q)$
    for the lexicographical order.
  \end{remark}

  \begin{proposition}\label{prop:orthoDualWitnesses}
    Given two cut free nets $S$
    and $T$ 
    atomically testable respectively 
    $A_1,\dots,A_n$
    and $A_1^\bot,\dots,A_n^\bot$.
    The assertions are equivalent;
    \begin{itemize}
      \item $S$ and $T$ are orthogonal.
      \item $S^\maltese$ and $T^\maltese$ are orthogonal.
    \end{itemize}
  \end{proposition}
  \begin{proof}
    Consider $S$ and $T$ 
    two nets (atomically) testable by dual sequents 
    so that $\# S = \# T = n$.
    The following reduction can be derived 
    (\autoref{prop:orthoDaiDescription}):
    $$
    S::T \rightarrow_{\mathtt{mult}}
    S^\maltese + T^\maltese + 
    \sum_{1\leq i \leq n}
    \sum_{\xi \in \maxadr{S(i)}} \cutlink{\findpos[S(i)]{\xi} , \findpos[T(i)]{\xi}}
    $$

    But observe 
    that the outputs of $S^\maltese$
    and $T^\maltese$
    are ordered according to the lexicographical order 
    on the pairs $(\rootof p, \adrof p)$.
    As a consequence two conclusion have the same index
    if and only if 
    their respective roots are $S(i)$ and $T(i)$ for some $i$
    and their address is the same.
    We derive the following 
    $$
    S^\maltese + T^\maltese + 
    \sum_{1\leq i \leq n}
    \sum_{\xi \in \maxadr{S(i)}} \cutlink{\findpos[S(i)]{\xi} , \findpos[T(i)]{\xi}} 
    =
    S^\maltese + T^\maltese + 
    \sum_{1\leq i \leq N}\cutlink{S^\maltese (i) , T^\maltese(i)}
    $$

    $1\Rightarrow 2.$ 
    By assumptioni $S::T\rightarrow \maltese_0$,
    because the multiplicative cuts
    get eliminated and can always be performed first 
    due to the commutation result (\autoref{prop:multiCommutes}),
    this ensures that $S::T\rightarrow S^\maltese :: T^\maltese$
    and since $S::T\rightarrow \maltese_0$
    it follows that $S^\maltese :: T^\maltese \rightarrow \maltese_0$.

    $2\Rightarrow 1.$
    We have shown with the previous equality that
    $S::T\rightarrow S^\maltese :: T^\maltese$
    thus we immediately conclude.
  \end{proof}

  \begin{proposition}[Orthogonality of nets and of natural partitions]\label{prop:orthoPartiGlue}
    Given $S^\maltese$ and $T^\maltese$ two nets 
    made only of daimon links,
    the assertions are equivalent:
    \begin{enumerate}
      \item 
      $S^\maltese$ and $T^\maltese$ are orthogonal.
      \item 
      $\nat{\daipart S}$
      and $\nat{\daipart T}$
      are orthogonal.
    \end{enumerate}
  \end{proposition}
  \begin{proof}
    Assume that 
    $S^\maltese :: T^\maltese$ reduces to $\maltese_0$.
    In other words the following net reduces to $\maltese_0$:
    $$S^\maltese + T^\maltese + \sum_{1\leq i \leq n} \cutlink{S^\maltese(i),T^\maltese(i)}.$$

    Equivalently this means that 
    $S^\maltese :: T^\maltese$ is an acyclic and connected graph.
    In particular if we contract each cut link $\cutlink{S^\maltese(i),T^\maltese(i)}$ in a single point 
    that we call $i$
    the acyclicity and connectedness is preserved.
    And in that case the resulting graph is 
    $\mathsf G(\nat{\daipart S} , \nat {\daipart T})$ 
    showning that these two partitions are orthogonal.
    The same argument in the other direction hold.
  \end{proof}

  \homogeneousOrtho*
  \begin{proof}
    This follows from the two previous proposition,
    \autoref{prop:orthoPartiGlue}
    and \autoref{prop:orthoDualWitnesses}.

      $\mathbf{1\Leftrightarrow 2.}$
      As for $1\Rightarrow 2$,
      we need to establish 
      a property of the rewriting of nets:
      if $N\redmult N'$
      and $N\reduce{} ^* \maltese_0$
      then $N'\reduce{}^* \maltese_0$.
      As for $2\Rightarrow 1$,
      it follows from 
      $S::T \rightarrow^* S^\maltese::T^\maltese$.
    
      $\mathbf{2\Leftrightarrow 3.}$
      One first proves that
      $S^\maltese::T^\maltese$
      is acyclic and connected (\acc)
      if and only if $\graph (\nat {\daipart S}, \nat{\daipart T})$ is.
    
      As for $3\Rightarrow 2$, 
      since no cycle occurs in $S^\maltese::T^\maltese$
      all (glueing) cuts are acyclic 
      and thus all cuts  
      can be eliminated,
      furthermore the elimination of a glueing cut 
      preserves \acc.
      The net $S^\maltese::T^\maltese$ contains 
      no connective link, thus its normal form 
      does not contain clash cuts.
      Hence because the 
      only net without conclusion that is 
      normal, \acc and contains no clash cuts
      is $\maltese_0$
      we conclude
      $S^\maltese::T^\maltese \rightarrow \maltese_0$.
      As for $2\Rightarrow 3$, 
      note that 
      anti--steps of cut elimination for glueing cuts 
      also preserve the \acc property 
      hence if $S^\maltese::T^\maltese \rightarrow \maltese_0$ 
      it must be that $S^\maltese::T^\maltese $ is \acc.
   
  \end{proof}

  \subsection{Description of tests 
        }

  The object of this section 
  is to show that the \emph{tests} of a formula $A$
  (\autoref{def:test})
  may be defined relatively to a single cut free net $S\witness A$:
  this is because the partitions $\uparrow^i \sigma S$
  for some switching of $S$,
  only depend on the $\otimes$-- and $\parr$--link 
  of $S$,
  which is the same for any cut free net $U\witness A$.

We now fix a special propositional variable 
that we denote $\placehold$.
A formula which contains only the $\placehold$
propositional variable is called a \emph{formula pattern}.

\begin{definition}
An hypergraph $M$ which contains only $\otimes$-- or $\parr$--links represents 
a formula pattern $F$, denoted $M\formurep F$ whenever:
\begin{itemize}
  \item 
  If $F=\placehold$ is the propositional variable $\placehold$
  then $M$ consists of $(\{ p \},\emptyset,\emptyset,\emptyset)$
  for some position $p$.
  \item 
  If $F= A\otimes B$
  then $M$ is equal to 
  $M_A + M_B + \tenslink{M_A(1),M_B(1)}{p}$
  where $p$ is a fresh position 
  and $M_A\formurep A$
  while $M_B\formurep B$.
  \item 
  If $F= A\parr B$
  then $M$ is equal to 
  $M_A + M_B + \parrlink{M_A(1),M_B(1)}{p}$
  where $p$ is a fresh position 
  and $M_A\formurep A$
  while $M_B\formurep B$.
\end{itemize}
\end{definition}

\begin{proposition}
Given an hypergraph $M$
and a formula pattern $A$.
If $M\formurep A$
then $M$ is target--surjective and source surjective,
and $M$ has a single conclusion.
\end{proposition}
\begin{proof}
By doing a simple induction on $A$.
\end{proof}

\begin{notation}
Given a formula pattern $A$ and a position $p$
we denote by $\formunet{A}{p}$
a hypergraph which represents 
the formula $A$ and has for conclusion $p$.
Furthermore we denote 
by $\formlink{A}{p_1,\dots,p_n}{p}$
a hypergraph $M$ of conclusion $p$ 
representing $A$ 
and such that 
$p_1,\dots,p_n$ are the 
input positions\footnote{that is positions which are the target of no link in $M$.}
of the hypergraph
ordered by lexicographical order on $\{\pleft,\pright\}^*$
for the addresses of the positions 
i.e. their addresses $\adrof {p_i}$ in $M$.
\end{notation}

\begin{lemma}[Decomposition Lemma]\label{lem:decomposition}
Given a cut--free net $S$ 
with $n$ conclusion 
there exists a unique sequence of $n$ formula patterns $A_1,\dots,A_n$
such that:
$$S = S^\maltese + \sum_{1\leq i \leq n}\formunet{A_i}{S(i)}.$$
\end{lemma}
\begin{proof}
One proceeds by induction 
on the number of connective links of $S$. 
\end{proof}

We prove that test of $A$ 
can be defined relatively to any single net $S\witness A$:
the next proposition 
ensures that the test of a formula 
may be described using any single net $S\witness A$:
this is because the switchings only depend on the $\parr$--
and $\otimes$--links 
of $S$ which are the same for all the nets $U\witness A$.

\begin{restatable}{proposition}{propTestDescri}\label{prop:testDescription}
  Given a formula $A$ and a net $S\witness A$.
  For any net $T$, 
  $T$ is a test of $A$ 
  iff 
  $T\witness A^\bot$
  and for some switching $\sigma S$ of $S$
  we have
  $\nat[T]{\daipart T} = \nat[S]{\uparrow^i \sigma S}$.    
\end{restatable}
\begin{proof}
  $2\Rightarrow 1.$ It is the obvious direction.

  $1\Rightarrow 2.$
  $S$ is a cut--free net atomically testable by $A$,
  consider any $U$ which is also cut--free and atomically testable by $A$:
  we show that for any switching $\sigma S$ of $S$ 
  there exists a switching $\tau U$ of $U$ such that 
  the partitions
  $\nat[S]{\uparrow^i \sigma S}$
  and 
  $\nat[U]{\uparrow^i \tau U}$
  are the same.

  $S$ and $U$ can be both written 
  as $S^\maltese + \formlink A {\vec u}{p}$
  and $U^\maltese + \formlink {A} {\vec v}{q}$.
  The nets $S^\maltese$
  and $U^\maltese$ are normal for the switching rewriting.
  Observe that the sum of two normal nets for the switching rewriting 
  is still normal for the switching rewriting;
  therefore 
  one can show that $S^\maltese + \formlink A {\vec u}{p} \switchto S'$
  implies $\formlink{A}{\vec u}{p} \switchto R$ (by showing the contraposition). 
  Indeed this will also be true for the net $U$.

  Therefore we derive the following:
  \begin{itemize}
    \item 
    A switching of $S$ is a net of the form 
    $S^\maltese + R_S$ where $R_S$ is a normal form of $\formlink A {\vec v}{q}$.
    \item 
    Similarly, a switching of $U$ is a net of the form 
    $U^\maltese + R_U$ where $R_U$ is a normal form of $\formlink A {\vec v}{q}$.
    \item 
    Finally observe that two cut--free nets $U^\maltese + R = V^\maltese + R$
  where $R$ is made of connective links only 
  then $\nat [U]{\uparrow^i U} = \nat[V]{\uparrow^i V}$.
  \end{itemize} 
  Consider then any two switchings of $S$ and $U$ with the same reduct $R$ 
  and denote them respectively $\sigma S$ and $\tau U$:   
  it follows that 
  $\nat[S]{\uparrow ^i \sigma S}$
  equals $\nat[T]{\uparrow^i \sigma T}$.

  To conclude, fix the net $S\witness A$.
  A test $T$ of $A$ is such that 
  $T\witness A^\bot$
  and for some $U\witness A$ and some of its switching $\tau U$
  $\nat[T]{\daipart T} = \nat[U]{\uparrow^i \tau U}$.
  By the previous considerations 
  there exists a switching $\sigma S$ of $S$ such that
  $\nat[S]{\uparrow ^i \sigma S}$
  equals $\nat[U]{\uparrow^i \tau U}$ 
  therefore 
  $\nat[T]{\daipart T} = \nat[S]{\uparrow ^i \sigma S}$:
  this allows us to conclude.    
\end{proof}

  \subsection{Orthogonality with tests -- proof of \autoref{prop:testsOrtho} 
  and \autoref{thm:danos}}

  Using the previous proposition and the well--known 
  result of Danos Regnier (\autoref{thm:standardcrit})
  we obtain the following proposition of \autoref{sect:Tests}.
  \orthoTest*

    \begin{proof}
      Consider a test $T$ of $\tests A$
      then $T\witness A^\bot$
      and $\nat[T]{\daipart T} = \nat[S]{\uparrow^i \sigma S}$ 
      for some switching $\sigma S$ of $S$ (\autoref{prop:testDescription}).
      By assumption $S$ and $T$ are orthogonal thus equivalently
      $\nat[S]{\daipart S}$ and $\nat[T]{\daipart T}$
      are orthogonal (\autoref{prop:orthoperfect}).
      Then
      $\nat[S]{\daipart S}$ is orthogonal to 
      each $\nat[S]{\uparrow^i \sigma S}$.
      we then 
      conclude that
      $\daipart S$ and $\uparrow^i \sigma S$ are orthogonal.
      Since this holds for any 
      switching $\sigma S$ thus
      calling theorem \ref{thm:standardcrit}
      we conclude that $S$ is a proof net.
      The other direction of the equivalence 
      is obtained using the same argument 
      in the opposite direction. 
   \end{proof}

  We can indeed generalise this result the case where $S$ 
  has multiple conclusions,
  as usual we do this by transformin 
  the net with multiple conclusions 
  in a net with one conclusion by adding a bunch of $\parr$--links.

  \begin{definition}[General connectives]
    A \emph{generalised $\parr$--link}
    on the positions $p_1,\dots,p_n$
    is a module 
    denoted $\link[\parr^n]{p_0,\dots,p_{n}}{p}$
    and defined by the following induction;
    \begin{itemize}
      \item $\link[\parr^1]{p_0,p_1}{p} = \parrlink{p_0,p_1}{p^1}$.
      \item For any $n>0$
      we defined 
      $\link[\parr^{n+1}]{p_1,\dots,p_{n+2}}{p} = 
      \link[\parr^n]{p^{n-1},p_{n}}{p} 
      + \parrlink{p^n,p_{n+1}}{p^{n+1}}$.
    \end{itemize}
  
    Similarly we defined 
    the generalised tensor links $\link[\otimes^n]{p_0,\dots,p_n}{p}$.
  \end{definition}

  \begin{proposition}
    Let $S$ be a net with conclusions 
    $p_0,\dots,p_n,q_1,\dots,q_k$.
    Let $S_1,\dots,S_n , T_1,\dots,T_k$
    be $n$ nets with one conclusion
    the assertions are equivalent;
    \begin{enumerate}
      \item $S\perp S_1\parallel \dots\parallel S_n
      \parallel T_1\parallel \dots\parallel T_k $.
      \item $S +\link[\parr^n]{p_0,\dots,p_n}{p} \perp 
      S_1 +\dots + S_n + \link[\otimes^n]{S_1(1)\dots,S_n(1)}{q}
      \parallel T_1\parallel \dots\parallel T_k$
    \end{enumerate}
  \end{proposition}
  \begin{proof}
    By a simple induction 
    of the size of the generalised $\parr$ connective.
  \end{proof}

  To derive the generalised theorem 
  one must observe that the tests 
  of $\parr$--formulas 
  are tensors of the tests of the subformulas.

  \begin{proposition}\label{prop:tensorTests}
    A test $T$ of $A\parr B$
    is of the form 
    $T_A\parallel T_B +\tenslink{T_A(1) , T_B(1)}{p}$
    where $p$ a fresh position,
    $T_A$ is a test of $A$
    and 
    $T_B$ is a test of $B$.
  \end{proposition}
  \begin{proof}
    This comes down to analysing the partitions 
    generated by a the switchings of a net $S\witness A\parr B$
    indeed this these 
    are partitions $P_A\cup P_B$
    where $P_A$ is a partition of the representation of $A$
    and $P_B$ is a partition of the representation of $B$.
  \end{proof}

  This easily leads to a proof of 
  the Danos Regnier theorem as stated in section 5.
  \danosregnier*
  
\subsection{Correctness of Tests and counter proofs -- proof of \autoref{thm:cortest}}

  \thmcortest*

  We have established the previous theorem 
  using the counter proof criterion of P.L. Curien
  found in \cite{curien:criterions}.
  Let us clearly state the theorem of counter proofs.
  \begin{theorem}[Counter Proof Criterion]\label{thm:CP}
    Given a cut--free net $S\witness A$
    the assertions are equivalent:
    \begin{itemize}
      \item 
      $S$ is orthogonal to each net $T\witness A^\bot$ representing a proof of $A^\bot$.
      \item 
      $S\vdash_{\mlldai} A$.
    \end{itemize}
  \end{theorem}

  \begin{remark}
    The counter proof criterion
    is formulated in \cite{curien:criterions}
    where non--homogeneous cut elimination does not exists.
    Because this is not the case in our work 
    we add the hypothesis 
    that $S\witness A$
    and the proof--opponents are of the form $T\witness A^\bot$
    so that only homogeneous cut will appear and be eliminated.
    Also observe that 
    a proof of $A^\bot$ such that $T\witness A^\bot$
    is an atomic proof of $\mlldai$
    that is, 
    a proof such that its daimon rules introduce 
    sequents which contains only propositional variables.
  \end{remark}

One can establish that 
the fact that tests of $A$
are proofs of $A^\bot$
is equivalent to an implication of the counter proof criterion.

\begin{proposition}
  Given a formula $A$
  the assertions are equivalent:
  \begin{enumerate}
    \item 
    Each test of $A$ 
    is the representation of a proof of $A^\bot$.
    \item 
    For any cut--free net $S\witness A$
    if $S$ is orthogonal to each (cut--free) $T\witness A^\bot$ representing a proof of $A^\bot$
    then $S\vdash_\mlldai A$.
  \end{enumerate}
\end{proposition}
\begin{proof}
  The fact that $2\Rightarrow 1$
  is the proof of \autoref{thm:cortest}
  (it uses the counter proof criterion \cite{curien:criterions}, \autoref{thm:CP}).
  To show $1\Rightarrow 2$,
  consider a net $S\witness A$
  and assume that each test $t$ of $A$
  represents a proof of $A^\bot$ 
  and is such that $t\witness A^\bot$.
  If the net $S$ is orthogonal to each 
  net $T\vdash_{\mlldai} A^\bot$
  with $T\witness A^\bot$
  then in particular it is orthogonal to 
  the set of tests $\tests A$:
  it follows that $S\vdash_{\mlldai} A$ (\autoref{thm:danos}).
\end{proof}

In fact the result of adequacy we have obtained for $\mlldai$
(\autoref{thm:adequacy})
allows to generalise the counter proof criterion of P.L. curien \cite{curien:criterions}
in presence of non homogeneous cuts.
More precisely, 
having an adequate basis implies 
a ``soundness'' result of a counter proof criterion,
i.e. any proof of $A$ is orthogonal to any proof of $A^\bot$.

\begin{proposition}\label{prop:adequacyImpliesCP}
  The first assertion implies the second:
  \begin{enumerate}
    \item There exists an adequate basis.
    \item Any $S\vdash_{\mlldai} A$
    is orthogonal to any $T\vdash_{\mlldai} A^\bot$. 
  \end{enumerate}
\end{proposition}
\begin{proof}
  $1\Rightarrow 2.$
  Say $\ibase$ is an adequate 
  basis then for any formula $A$ we have 
  $\proofs{A:\mlldai} \subseteq \reali[\ibase]{A}$.
  Now since $\reali[\ibase]{A}$ equals $\reali[\ibase]{A^{\bibot}}$
  we derive $\reali[\ibase]{A}\subseteq \reali[\ibase]{A^\bot}^\bot$ (\autoref{def:interpretmll}).
  This means that any realiser of $A$ is a orthogonal to any realiser of $A^\bot$,
  since by adequacy a net $S\vdash_{\mlldai}A$ realises $A$
  and a net $T\vdash_{\mlldai} A^\bot$ realises $A^\bot$ we conclude.
\end{proof}

Then we can obtain a general form of counter proof criterion 
namely when we don't have the restriction that 
the proofs are atomic i.e. $S\witness A$ and $T\witness A^\bot$.

\begin{proposition}[Generalised Counter Proof criterion]\label{prop:GCP}
  Given a formula $A$ and $S$ a cut--free net
  the assertions are equivalent:
  \begin{enumerate}
    \item $S$ is orthogonal to each net $T$ representing a proof of $A^\bot$.
    \item $S\vdash_{\mlldai} A$.
  \end{enumerate}
\end{proposition}
\begin{proof}
  $1\Rightarrow 2.$
  Assume that $S\witness A_0$ with $A_0$ such that there exists $\theta$
  with $\theta A_0 = A$
  and that
  $S$ is orthogonal to each proof nets $T\vdash_{\mlldai} A^\bot$.
  In particular any net 
  $T\vdash_{\mlldai} A_0^\bot$
  with $T\witness A_0^\bot$
  is such that 
  $T\vdash_{\mlldai} A^\bot$ 
  (\autoref{prop:smallproof}). 

  As a consequence:
  $S$ is orthogonal to each $T\vdash_{\mlldai} A_0^\bot$
  with $T\witness A_0^\bot$.
  by the counter proof criterion \cite{curien:criterions}
  we derive $S\vdash_{\mlldai} A_0$
  and thus $S\vdash_{\mlldai} A$ (again using \autoref{prop:smallproof}).

  $2\Rightarrow 1.$
  We aim at showing that 
  any proof $S\vdash_{\mlldai} A$
  is orthogonal to any proof $T\vdash_{\mlldai} A^\bot$
  which may have a different syntax tree (and here lies the novelty with respect to 
  P-L. Curien's theorem \cite{curien:criterions}).
  Showing this implication 
  by analyzing the rewriting rules of cut--elimination 
  can be delicate,
  however
  in the language of realisability this becomes a trivial consequence
  of the adequacy theorem (\autoref{prop:adequacyImpliesCP}).
  Hence because adequacy hold (\autoref{thm:adequacy}), we conclude.
\end{proof}

\subsection{On Substitutions and Testability -- Proof of \autoref{prop:collapse}}

Recall that a \emph{substitution} 
is a map $\theta$ which maps 
propositional variables to formulas.
Naturally substitutions can be lifted 
by induction to formulas and to sequents;
$\theta(A\square B) \triangleq \theta A \square \theta B$
and $\theta(\Gamma,A) = \theta(\Gamma), \theta A$.
A sequent $\Gamma$ is an \emph{instance} 
of a sequent $\Delta$ whenever there exists 
a substitution $\theta$
such that $\theta\Delta  = \Gamma$.
In that case we denote $\Delta \leq \Gamma$.



\begin{remark}
  Given two representation of sequents 
  $\Gamma= A_1,\dots, A_n$ and $\Delta= B_1,\dots, B_n$
  are such that $\Delta\leq \Gamma$
  this implies 
  that for each index $1\leq i \leq n$
  we have $B_i \leq A_i$
  specifically 
  $A_i[X_1\mapsto F_1,\dots,X_k\mapsto F_k] = B_i$.
\end{remark}

\begin{remark}
  Whenever a cut--free net $S$ is such that 
  $S\witnesseq \Gamma$
  then there exists a substitution $\theta$
  and a sequent $\Delta$
  such that $S\witness \Delta$
  and $\theta \Delta$.
\end{remark}

Let us provide a proof of the previous remark.

\begin{proposition}\label{prop:witness}
  Given a cut--free net $S$ and a sequent $\Gamma$.
  If $S\witnesseq \Gamma$ 
  then there exists a sequent $\Delta$
  and a substitution $\theta$
  such that $\theta\Delta = \Gamma$ 
  and $S\witness \Delta$.
\end{proposition}
\begin{proof}
  If $S\witness \Gamma$ we can calready conclude.
  Otherwise $S\witnesseq \Gamma$
  and let us denote $p_1,\dots,p_n$ the initial positions of $S$,
  and $\tau$ the formula labelling witnessing $S\witnesseq \Gamma$.
  Because formula labellings are total functions 
  each $p_i$ as a formula $\tau(p_i) = F_i$ associated with it
  consider then $X_1,\dots,X_n$ a family of  distinct propositional variables 
  and the substitution $\theta[X_1\mapsto F_1 ; \dots ;X_n\mapsto F_n]$
  then consider the (unique) atomic formula labelling 
  such that $\tau'(p_i) = X_i$
  then $S\witness \tau'(S(1)),\dots,\tau'(S(n))$.
  Indeed then applying $\theta$
  to the sequent $\tau'(S(1)),\dots,\tau'(S(n))$
  will result in $\tau(S(1)),\dots,\tau(S(n))$
  which is $\Gamma$.
\end{proof}

\propcollapse*
\begin{proof}
  Using remark \ref{rem:testinbase}
  and the fact that the tests of 
  a formula $B$ in $\Delta$
  are proofs of $B^\bot$
  and by proposition \ref{prop:smallproof}
  are proofs of $A^\bot$.
\end{proof}

\section{Complements to \autoref{sect:Completeness}}\label{ann:sec6add}

  \subsection{Decomposition}

  \begin{proposition}[Decomposition]\label{prop:decomp}
    Let $\ibase$ be an interpretation basis,
    $\hseq$ be an hypersequent, $A,B$ two formulas 
    and $S$ be a multiplicative net.
    \begin{itemize}
      \item If 
      $S=S_0+l$ has a terminal $\parr$ link $l$ above its last conclusion 
      and $S\in\reali[\ibase]{\hseq , A\parr B}$;
      then $S_0$ belongs to $\reali[\ibase]{\hseq , A,B}$
      \item If 
      $S=S_0+l$ has a terminal $\parr$ link $l$ above its last conclusion 
      and $S\in\reali[\ibase]{\hseq \parallel A\parr B}$;
      then $S_0$ belongs to $\reali[\ibase]{\hseq \parallel (A,B)}$
      \item If 
      $S=S_0+l$ has a terminal $\otimes$ link $l$ above its last conclusion 
      and $S\in\reali[\ibase]{\hseq , A\otimes B}$;
      then $S_0$ belongs to $\reali[\ibase]{\hseq , (A\parallel B)}$
      \item If 
      $S=S_0+l$ has a terminal $\otimes$ link $l$ above its last conclusion 
      and $S\in\reali[\ibase]{\hseq \parallel A\otimes B}$;
      then $S_0$ belongs to $\reali[\ibase]{\hseq \parallel A\parallel B}$
    \end{itemize}
  \end{proposition}
  \begin{proof}
    \begin{itemize}
      \item 
    Consider $S$ a net with a terminal $\parr$--link $l$
    decomposing $S=S_0+l$ such that $l$ outputs the only conclusion of $S$.
    Say $S$ belongs to $\reali[\ibase]{A\parr B}$
    equivalently 
    $S$ is orthogonal to $\reali[\ibase]{A}^\bot \otimes \reali[\ibase]{B}^\bot$,
    since a multiplicative cut can always be performed first
    this implies 
    that $S_0$ is orthogonal to $\reali[\ibase]{A}^\bot \parallel \reali[\ibase]{B}^\bot$
    thus $S_0$ belongs to $\reali[\ibase]{A} \compo \reali[\ibase]{B}$.
    By allowing types 
    to contain cuts, 
    this arguments easily adapts to $\reali[\ibase]{\hseq , A,B}$.

    \item
    The reasonment for a net in $\reali[\ibase]{A\otimes B}$
    is similar and also easily adapts to 
    the case $\reali[\ibase]{\hseq , A\otimes B}$.

    \item
    Consider on the other hand a 
    net $S$ in $\reali[\ibase]{\hseq \parallel A\otimes B}$,
    such that $S=S_0+l$ where $l$ is a tensor link and outputs the last conclusion of $S$.
    Equivalently $S$ is a net orthogonal to 
    $\reali[\ibase]{\hseq}^\bot \compo \reali[\ibase]{A\otimes B}^\bot$
    i.e.  $\reali[\ibase]{\hseq}^\bot \compo \reali[\ibase]{A}^\bot \parr \reali[\ibase]{B}^\bot$.
    Since the multiplicative cut commute to the left (\autoref{prop:multiCommutes})
    we derive 
    that 
    $S_0$ is orthogonal to 
    $\reali[\ibase]{\hseq}^\bot \compo \reali[\ibase]{A}^\bot \compo \reali[\ibase]{B}^\bot$.
    Equivalently (by duality)
    $S_0$ belongs to 
    $\reali[\ibase]{\hseq} \parallel \reali[\ibase]{A}  \parallel \reali[\ibase]{B} $.

    \item
    Similarly we treat the case $\reali[\ibase]{\hseq \parallel A\parr B}$.
    Consider on the other hand a 
    net $S$ in $\reali[\ibase]{\hseq \parallel A\parr B}$,
    such that $S=S_0+l$ where $l$ is a parr link and outputs the last conclusion of $S$.

    Equivalently $S$ is a net orthogonal to 
    $\reali[\ibase]{\hseq}^\bot \compo \reali[\ibase]{A\parr B}^\bot$
    i.e.  $\reali[\ibase]{\hseq}^\bot \compo \reali[\ibase]{A}^\bot \otimes \reali[\ibase]{B}^\bot$.
    Because multiplicative cuts can always be performed 
    first in a reduction sequence (\autoref{prop:multiCommutes})
    is follows that 
    $S_0$ is ortogonal to 
    $\reali[\ibase]{\hseq}^\bot \compo \reali[\ibase]{A}^\bot \parallel \reali[\ibase]{B}^\bot$.
    Equivalently 
    $S_0$ belongs to $\reali[\ibase]{\hseq , A \parallel B}$.
    \end{itemize}
  \end{proof}

\renewcommand*{\witnesseq}{\witness}

\subsection{The merge construction belongs to the composition}

From \autoref{prop:mergeCompute}
we can derive easily 
that an explicitly defined 
type 
(called merge and denoted $\mathbf A\merge \mathbf B$) belongs to 
the composition of two types 
$\mathbf A \compo \mathbf B$.
This will be used to show that 
the functional composition $(\compo)$ and 
the parallel composition $(\parallel)$
of full types remains full
(that is the \autoref{prop:openpreserve}
of \autoref{ann:sec6add})
therefore by duality this will imply
that the parallel and functional composition 
of daimon types is a daimon type.

\begin{definition}
  Given $A$ and $B$ two sets of nets 
  we denote by $A\merge[]B$ the following set:
  $$
  \{ a\merge[(d,d')] b \mid a \in A , b\in B ,\quad 
  d\in E_a ,\labl_a(d) = \maltese ,  
  \quad 
  d\in E_b ,\labl_b(d) = \maltese\}.$$
\end{definition}

\begin{proposition}[Merges belongs to the compositional construction]\label{prop:mergeInCompo}
  Given $\mathbf A$ 
  and $\mathbf B$
  two types;
  $$(\mathbf A \bowtie \mathbf B)^{\bot\bot}
  \subseteq \mathbf A \compo \mathbf B.$$
\end{proposition}
\begin{proof}
  Consider $S$ 
  a net in $A\bowtie B$.
  $S$ may therefore be written as 
  $a\merge[(d_a,d_b)] b$ for two elements $a\in\mathbf A$
  and $b\in\mathbf B$
  with $d_a$ a daimon of $a$ 
  and $d_b$ a daimon of $b$.
  We can thus decompose 
  $b$ as $d_b + b_0$
  and 
  $a\merge[(d_a,d_b)] b$ is equal to 
  $(a\merge [(d_a),d_b]d_b )+ b_0$ (\autoref{prop:mergeDecompo}).

  Now consider an element 
  $\overline a$ of $\mathbf A$,
  then 
  $a$ and $\overline a$ are orthogonal,
  therefore
  $(a\merge [(d_a),d_b]d_b ) ::\overline a$ reduces to $d_b$ (\autoref{prop:mergeCompute}).
  Furthermore 
  $(a\merge [(d_a),d_b]d_b )+ b_0 :: \overline a$ 
  is equal to 
  $((a\merge [(d_a),d_b]d_b ):: \overline a) + b_0 $
  because $\overline a$ only interacts with the conclusions of $a$
  thus that net reduces to 
  $d_b + b_0$ which is exactly $b$
  and so belongs to $\mathbf B$.
  Thus it follows that 
  the net
  $((a\merge [(d_a),d_b]d_b )+ b_0 ):: \overline a$
  that is $a\merge[(d_a,d_b)] b_0 :: \overline a$ belongs to $\mathbf B$
  for any $\overline a \in\mathbf A$
  thus that it belongs to $\mathbf A \compo \mathbf B$.

  We have shown 
  $\mathbf A \merge \mathbf B \subseteq \mathbf A \compo\mathbf B$
  and thus since inclusion is stable under bi orthogonal,
  $(\mathbf A \merge \mathbf B)^{\bibot} \subseteq \mathbf A \compo\mathbf B$  
\end{proof}

  \subsection{Full Types and daimon types}

  We introduce some terminology regarding types
  in order to identify 
  types which contain only nets 
  whose conclusions are outputs of daimon links.
  Such basis are for instance daimon basis,
  and, one can easily see that $\baseone$ is a daimon basis.

  \begin{definition}
    A type $\mathbf A$
    is \emph{full}
    when for each $1\leq i \leq \#\mathbf A$
    there exists a net $S$ (resp. $T$)
    such that 
    its conclusion $S(i)$ (resp. $T(i)$)
    is the output of a $\parr$--link (resp. $\otimes$--link).
  \end{definition}

  \begin{remark}
    Observe that in a daimon type $\mathbf A$
    each conclusion of a net $S\in\mathbf A$
    is the output of a daimon link (\autoref{rem:daimonTypeConclusion}): 
    as a consequence there exists an atomic (that is, containing only propositional variable)
    sequent $\Gamma$
    made of $\#\mathbf A$ formulas
    such that any cut free net $S$ of $\mathbf A$
    is testable by $\Gamma$
    i.e. 
    $S\witness \Gamma$
  \end{remark}
  
  \begin{proposition}\label{prop:openpreserve}
    Given two full types 
    $\mathbf A$ and $\mathbf B$;
    \begin{itemize}
      \item Their functional composition 
      $\mathbf A \compo \mathbf B$ is still full.
      \item 
      Their parallel composition 
      $\mathbf A \parallel \mathbf B$ is still full.
    \end{itemize}

    Given two daimon types 
    $\mathbf A$ and $\mathbf B$;
    \begin{itemize}
      \item Their functional composition 
      $\mathbf A \compo \mathbf B$ is still a daimon type.
      \item 
      Their parallel composition 
      $\mathbf A \parallel \mathbf B$ is still a daimon type.
    \end{itemize}
  \end{proposition}
  \begin{proof}
    Indeed the parallel $\mathbf A \parallel \mathbf B$
    remains full since it contains in particular
    the nets of the form $a\parallel b$
    with $a\in\mathbf A$
    and $b\in\mathbf B$.
    Similarly this remains true for $\mathbf A \compo\mathbf B$
    since it contains the merge $\mathbf A \merge \mathbf B$
    (\autoref{prop:mergeInCompo}).

    For the daimon types 
    we conclude using the duality results 
    (\autoref{prop:dualPreC}).
  \end{proof}

  \begin{proposition}
    Let $\mathbf A$ and $\mathbf A^\bot$
    be two orthogonal and non--empty types,
    the assertions are equivalent:
    \begin{itemize}
      \item $\mathbf A$ is a daimon type.
      \item Each conclusion of a net $S$ of $\mathbf A$ 
      is the output of a daimon link.
    \end{itemize}
  \end{proposition}
  \begin{proof}
    $\mathbf{1\Rightarrow 2.}$
    This is already discussed in \autoref{rem:daimonTypeConclusion}.
    A net $S$ in $\mathbf A$
    is orthogonal to any nets in $\mathbf A^\bot$,
    and because $\mathbf A^\bot$ is full 
    the conclusions of $S$ cannot be the outputs of connectives
    otherwise a clashing cut will 
    occur in some interaction $S::T$
    for some $T$ in $\mathbf A^\bot$.

    $\mathbf{2\Rightarrow 1.}$
    This requires a much more detailed analysis 
    of the rewriting rules of cut elimination.
    We will not give the detail here 
    because this implication is not used in this work
    and the proof 
    requires the addition of a lot
    of content in order to be obtained.
  \end{proof}

  \subsection{Daimon Basis are Compact -- proof of \autoref{prop:baseOneTestable}}

  The object of this section 
  is to provide a proof of \autoref{thm:closedIsCompact}.
  Let us define \emph{compact basis}
  as basis which satisfy the first item of \autoref{prop:collapse};
  such a basis if also enjoys adequacy, will obviously 
  be complete for $\mlldai$ (using \autoref{prop:collapse}).

  \begin{definition}
    A basis $\ibase$
    is \emph{compact}
    if for any $\Gamma$
    and any net $S$;
    $S\Vdash_\ibase \Gamma \Rightarrow S\witnesseq \Gamma$.
  \end{definition}  

  To obtain $\mlldai$--completeness (for cut--free nets) 
\autoref{prop:collapse}
ensure that it is enough 
to find a basis 
that is approximable and compact:
such kind of basis are \emph{daimon} basis (\autoref{thm:closedIsCompact}),
for instance the basis mapping 
each propositional variable to $\smash{\{ \maltese_1 \}^{\bibot}}$.
Technically this is where 
it is useful to interpret hypersequents; 
a basis is compact is equivalently
a basis such that for any hypersequent $\hseq$ which is $\parr-$ and $\otimes$--free 
the interpretation $\reali[\ibase]{\hseq}$
contain only nets of whose conclusions are targets of daimon links.

\newcommand{\format}[2][S]{\mathtt{f}_{#1}(#2)}


\begin{definition}[Daimon type]
  A type $\mathbf A$ is a
  \emph{daimon type}
  whenever for each $1\leq i \leq \# \mathbf A$
  its dual $\mathbf A^\bot$
  contains 
  a net $S$ with a terminal 
  $\parr$--link which outputs $S(i)$
  and a net $T$ with a terminal $\otimes$--link
  which outputs $T(i)$.
  A basis $\ibase$ is a
  \emph{daimon basis}
  when for each propositional variable $X$
  its interpretation $\reali[\ibase]{X}$
  is a daimon type.
\end{definition}

\begin{remark}\label{rem:daimonTypeConclusion}
  The conclusions 
  of a net which belong to a daimon type $\mathbf A$
  must be the outputs of daimon links (\autoref{rem:clashcut}).
\end{remark}

A daimon basis is indeed 
a compact basis;
the elements of a (non--empty) daimon type 
can only have conclusions which are outputs of a daimon link
otherwise clashing cuts occur (\autoref{rem:clashcut}),
then an induction on the formulas show compactedness of the basis.
First we must show a more general result which holds on hypersequents.

There is a simple inductive process to associate 
a sequent $\downarrow\hseq$ to an hypersequent $\hseq$;
$\downarrow A \triangleq A$
while
$\downarrow \hseq_1 , \hseq_2 \triangleq \downarrow \hseq_1,\downarrow \hseq_2$
and 
$\downarrow \hseq_1 \parallel \hseq_2  \triangleq\downarrow \hseq_1 , \downarrow \hseq_2$.

  \begin{restatable}[Daimon Basis Property]{lemma}{lemtrunc}\label{lem:trunc}
    Let $\ibase$ be a daimon basis.
    For any cut--free net $S$;
    $S\in\reali[\ibase]{\mathcal H}
    \Rightarrow S\witnesseq \downarrow{\mathcal H}.$
  \end{restatable}  

  \begin{proof} of lemma \ref{lem:trunc}.
    By induction on the hypersequent
    using the measure $(c,n)$
    where $c$ is the number of connectives in the hypersequent 
    and $n$ is the size of the hypersequent.
    \begin{itemize}
      \item If the hypersequent is made of one atomic formula $X$
      then $S\in\reali[\ibase]{X}$
      implies that $S$ belongs to a daimon type,
      thus the outputs of $S$ are outputs of a daimon (\autoref{rem:daimonTypeConclusion}).
      This by definition means $S\witness X$
      and therefore $S\witnesseq X$.
      \item 
      Say the hypersequent is made only of atomic formulas.
      The hypersequent may be of the form 
      $\hseq_1,\hseq_2$
      then $\reali[\ibase]{\hseq_1,\hseq_2}
      = (\reali[\ibase]{\hseq_1}^\bot\parallel\reali[\ibase]{\hseq_2}^\bot)^\bot$
      since,
      $\reali[\ibase]{\hseq_1}^\bot$
      and $\reali[\ibase]{\hseq_2}^\bot$
      are open types their parallel composition 
      remains an open type.
      $\reali[\ibase]{\hseq_1,\hseq_2}$
      is the orthogonal of an full type hence 
      it is a daimon type.
      Thus any sequent in that type 
      as outputs which comes from a daimon link 
      and thus is the approximation of any sequent of size $n$
      in particular $S\witnesseq \downarrow \hseq$.
      We do a similar reasonment when $\hseq =\hseq_1\parallel\hseq_2$.
  
      \item
      Case of non--atomic hypersequent 
      with a virgula as main connective.
      Assume that $S$ has a terminal connective link,
      say the hypersequent is of the form $\hseq,A\parr B$
      such that $S = S_0 + l$
      where $l$ is a $\parr$--link and 
      is the last conclusion of $S$.
      By proposition \ref{prop:decomp} 
      it follows that $S_0$
      belongs to $\reali[\ibase]{\hseq,A,B}$
      the measure of that hypersequent as decreased and 
      so we apply the induction hypothesis;
      $S_0\witnesseq \downarrow(\hseq, A, B)$
      indeed it follows that $S\witnesseq \downarrow(\hseq, A\parr B)$.
      A similar argument works for an hypersequent of the form $\hseq, A\otimes B$.
  
      \item
      Case of non--atomic hypersequent 
      with a parallel as main connective.
      On the other hand say 
      $S$ belongs to the interpretation $\reali[\ibase]{\hseq \parallel A\parr B}$
      then it is orthogonal to $\reali[\ibase]{\hseq}^\bot\compo \reali[\ibase]{A\parr B}^\bot$
      and thus $S_0$ is orthogonal to $\reali[\ibase]{\hseq}^\bot\compo \reali[\ibase]{A, B}^\bot$.
      Equivalently using proposition \ref{prop:decomp} 
      $S_0$ belongs to $\reali[\ibase]{\hseq \parallel A, B}$
      the size of the hypersequent as decreased and so we can apply the induction hypothesis;
      $S_0\witnesseq \downarrow(\hseq, A, B)$
      and thus $S\witnesseq \downarrow(\hseq, A\parr B)$.
      A similar argument works for an hypersequent of the form $\hseq\parallel A\otimes B$.
      Similarly we treat the case of $\hseq,A\otimes B$.
    \end{itemize}
  \end{proof}

  Indeed the previous lemma is a more general form of \autoref{thm:closedIsCompact}
  as stated in this subsection.
  One can now easily derive it:

\begin{restatable}{theorem}{daimonIsCompact}\label{thm:closedIsCompact}
  Let $\ibase$ be an approximable daimon basis, 
  $\Gamma$ a sequent,
  and $S$ a cut--free net;
  $S\Vdash_\ibase \Gamma \Rightarrow S\vdash_\mlldai \Gamma$.
\end{restatable}


  From this theorem 
  the fact that the base $\baseone$ is compact easily follows,
  because one can show that $\baseone$ is a daimon basis.

  \baseOne*
\begin{proof}  
  Observe that the type $\{ \maltese_1 \}^{\bibot}$
is indeed a daimon type 
(\autoref{rem:daimonOrthogonals}),
therefore the basis mapping each 
propositional variable $X$ to $\{\maltese_1\}^{\bibot}$
is a daimon basis.
We conclude using \autoref{thm:closedIsCompact}.
\end{proof}

  \subsection{Proof of \autoref{thm:mlldaiNewComplete}}

  \mlldaiCompletion*
  \begin{proof}
    As for (1)
    note that if $S$ realises $\Gamma$
    for any basis $\ibase$
    then in particular it realises $\Gamma$ for $\baseone$.
    Thus applying \autoref{prop:baseOneTestable}
    $S$ must be testable by $\Gamma$,
    furthermore $\baseone$ is approximable
    hence applying \autoref{prop:collapse}
    we conclude $S\vdash_{\mlldai}\Gamma$.

    As for (2) 
    to have the implication that $S\in\reali[\ibase]{\Gamma}$
    yields $S\vdash_{\mlldai}\Gamma$
    simply use the same argument as for point (1)
    merely noting that $\baseone$ is an approximable basis.
    To show that $S\vdash_{\mlldai}\Gamma$
    implies $S\in\reali[\ibase]{\Gamma}$
    is exactly the theorem of adequacy (\autoref{thm:adequacy}).
  \end{proof}

  \subsection{Proof of \autoref{thm:completeness}}

  \begin{figure}
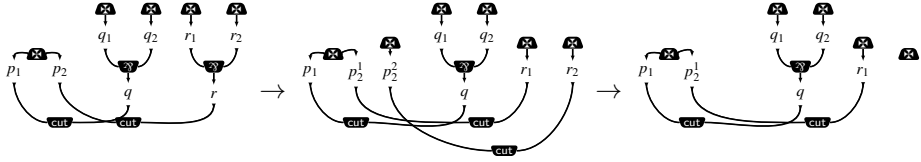

    \centering
    \scalebox{0.6}{\input{interactionb.tikz}}
    $\rightarrow$
    \scalebox{0.6}{\input{interactionb1.tikz}}
    $\rightarrow$
    \scalebox{0.6}{\input{interactionb2.tikz}}
    \caption{The interaction 
    of $\maltese_2$ with the net $\maltese_\parr\parallel\maltese_\parr$,
    this cannot reduce to $\maltese_0$
    since disconnection of the net 
    is preserved by cut--elimination 
    and $\maltese_0$ is connected.}\label{fig:cutfailure}
  \end{figure}

  \def\ground#1{\mathsf{g}(#1)}

  The \emph{ground form} of an $\mll$ formula 
  $A$ is the atomic hypersequent
  obtain when each $\otimes$
  is replaced by a $\parallel$
  and each $\parr$ is replaced by a $\compo$.

  \begin{proposition}\label{prop:parallelExactly}
    Given a cut--free net $a\parallel b$ with binary daimon links
    and two formulas $A$ and $B$ which have no variable in common
    if $a\parallel b$ belongs to $\reali[\ibase]{A}\parallel \reali[\ibase]{B}$
    for any base $\ibase$
    then $a\in\reali[\ibase]{ A}$
    and $b\in\reali[\ibase]{B}$
    for any basis $B$.
  \end{proposition}
  \begin{proof}
    We rely on the following fact,
    $a\parallel b$ belongs $\reali[\ibase]{A}\parallel \reali[\ibase]{B}$
    so in particular $a\parallel b$ 
    is orthogonal to $\reali[\ibase]{A}^\bot\merge \reali[\ibase]{B}^\bot$
    (\autoref{prop:mergeInCompo}).

    By induction on the formula $A$ and $B$
    if both are equal to $X$ and $Y$ some propositional variable,
    then 
    $a\parallel b$ is orthogonal to 
    $\reali[\ibase]{X}\merge \reali[\ibase]{Y}$
    for any choice of basis
    therefore the conclusion of $a\parallel b$
    may only be daimon links.
    Since furthermore 
    $\bigcap_{\ibase}\reali[\ibase]{X}$ is empty 
    it follows that 
    $\reali[\ibase]{X}\parallel \reali[\ibase]{Y}$
    is also empty thus the proposition hold.

    Note that because $a\parallel b \in\reali[\ibase]{A\parallel B}$
    for any basis and in particular for $\baseone$ which is daimon basis
    it follows that 
    $a\witnesseq A$
    and $b\witnesseq B$ (\autoref{lem:trunc}).
    Since $a\parallel b \perp\reali[\ibase]{A}^{\bot}\merge \reali[\ibase]{B}^\bot$
    from \autoref{prop:factorisation}
    one can show using \autoref{prop:decomp} 
    that this equivalently means that
    $a^\maltese \parallel b^\maltese$ is orthogonal to
    $\reali[\ibase]{\ground A}^{\bot}\merge \reali[\ibase]{\ground B}^\bot$.
    $a^\maltese \in \reali[\ibase]{\ground A}$
    iff $a\in\reali[\ibase]{A} $
    we assume that $a^\maltese\not\reali[\ibase]{\ground A}$ for some basis.
    Consider a basis $\ibase$
    such that $a^\maltese :: s$
    does not reduce to $\maltese_0$,
    although for this to be possible the outputs 
    of a same daimon link in $a$ 
    must be realising $X,Y$.
    one can easily check 
    that 
    $a^\maltese \parallel b^\maltese$ 
    is not orthogonal to $s\merge s'$
    for all $s'\in\reali[\ibase]{\ground B}$ 
    (Because $A$ and $B$ don't share any propositional variable 
      the interpretation of $\ground B$
      can be chosen so that orthogonality fails)
    and thus that it does not belong to
    $\reali[\ibase]{A}\parallel \reali[\ibase]{B}$
    for all basis. 
  \end{proof}

  \begin{lemma}[Splitting a Parallel Interaction]\label{lem:splitting}
    Given a sequent 
    $\Gamma$ with $n$ formulas and two formulas $A$ and $B$
    a net $S::=S_1\parallel S_2$ with $n+2$ conclusions
    such that $S(n+1)$ is a conclusion of $S_1$
    and $S(n+2)$ is a conclusion of $S_2$.
    There exists two unique sequents $\Gamma_1$
    and $\Gamma_2$ such that 
    for any basis:

    $ S\in \reali[\ibase]{\Gamma, A\parallel B} 
    \Rightarrow  
    S\in \reali[\ibase]{\Gamma_1,A}\parallel\reali[\ibase]{\Gamma_2,B}$
  \end{lemma}
  \begin{proof}
    The inclusions of 
    $    \reali[\ibase]{\Gamma_1,A}\parallel\reali[\ibase]{\Gamma_2,B}$
    in $\reali[\ibase]{\Gamma, A\parallel B}$
    is obtained by standard calculations.

    on the other hand 
    consider a net $S$ in $\reali[\ibase]{\Gamma, A\parallel B}$
    and is of the form $S_1\parallel S_2$.
    $S_1\parallel S_2 ::\overline\gamma$ 
    belongs to $A\parallel B$.
    And such that $S_1\parallel S_2 :: \overline\gamma$ belongs to $A\parallel B$.
    The opponent $\overline\gamma$ of $\reali[\ibase]{\Gamma}^\bot$
    is an element of 
    $\reali[\ibase]{A_1}^\bot \parallel\dots\parallel \reali[\ibase]{A_n}^\bot$,
    thus it can range over 
    $\reali[\ibase]{A_1}^\bot \parallel^{-}\dots\parallel^{-} \reali[\ibase]{A_n}^\bot$
    and we can set 
    $\overline{\gamma}= U_1\parallel \dots \parallel U_n$.
    In $S_1\parallel S_2 :: (U_1\parallel \dots \parallel U_n)$
    each $U_i$ is cut with a conclusion of $S_1\parallel S_2$
    i.e. a conclusion of $S_1$ or a conclusion of $S_2$.
    We can therefore split the $U_i$'s
    according to those that are cut with $S_1$
    and those that are cut with $S_2$
    let us denote the nets made of the parallel of these families 
    respectively $U^1$
    and $U^2$.
    This can be rewritten using \autoref{prop:commutationInterSumlink}
    to obtain:
    \begin{align*}
      & S_1\parallel S_2 :: (U_1\parallel \dots \parallel U_n) \\
      =\quad &
      S_1\parallel S_2 :: (U^1 \parallel U^2) \\
      =\quad &
      (S_1 :: U^1 \parallel S_2 ):: U^2 &(\mbox{\autoref{prop:commutationInterSumlink}}) \\
      =\quad &
      (S_1 :: U^1) \parallel (S_2:: U^2 ) &(\mbox{\autoref{prop:commutationInterSumlink}}) 
    \end{align*}
    Now $U^1$ range in 
    $\reali[\ibase]{\Gamma_1}$
    while $U^2$ range in 
    $\reali[\ibase]{\Gamma_2}$.
    Since 
    $S_1\parallel S_2:: (U_1\parallel \dots \parallel U_n)$
    belongs to $\reali[\ibase]{A}\parallel \reali[\ibase]{B}$
    it follows that 
    $S_1 :: U^1 \parallel S_2:: U^2 $
    belongs to $\reali[\ibase]{A}\parallel \reali[\ibase]{B}$.
    By \autoref{prop:parallelExactly}
    it follows that $S_1 :: U^1$
    belongs to $\reali[\ibase]{A}$
    and $S_2:: U^2 $ belongs to $\reali[\ibase]{B}$.
  \end{proof}

  \thmcomp*

  \begin{proof}
    Assume that $S$ realise the sequent $\Gamma$
    for any basis.
    In particular then $S$ realises $\Gamma$
    for the basis $\baseone$.
    By \autoref{prop:baseOneTestable}
    and \autoref{prop:collapse}
    this implies that $S\vdash_{\mlldai}\Gamma$.
    To conclude one must show that the proof 
    tree $\pi$ represented by $S$
    is indeed from $\mll$:
    this means that (1) 
    its daimon rules should introduce only sequents
    of size $2$ i.e. of the form $X,Y$
    and (2) these sequents must be such that $Y = X^\bot$.
    The point (1) is guaranteed by the assumption that the daimons of $S$
    are binary \footnote{in terms of classical realisability 
    one can call the nets with binary daimon \emph{proof like}.}.

    To prove (2), we reason by induction on the proof tree $\pi$ represented by $S$
    however we will need to show in the inductive cases 
    that the assumption of the theorem hold.
    
    \emph{Base case.}
    If $\pi$ is made of a single daimon rule,
    then since the daimons are binary 
    it must introduce a sequent $X,Y$,
    now we show that it cannot be that $X\neq Y$,
    we assume that $S\in \reali[\ibase]{X,Y}$
    for any base $\ibase$
    meaning $S\in \reali[\ibase]{X}\compo\reali[\ibase]{Y}$
    or $S\perp \reali[\ibase]{X}^\bot \parallel \reali[\ibase]{Y}^\bot$ (\autoref{prop:dualC}).
    Since we range over all basis we may map 
    $X$ and $Y$ to the same type 
    in particular we can map both atomic proposition 
    to the type 
    $\{ \maltese_\parr \}^\bot$
    where 
    $\maltese_\parr ::=  \dailink{p_1} + \dailink{p_2} + \parrlink{p_1,p_2}{p}$.
    In that case; 
    \begin{align*}
      & \maltese_2\perp \reali[\ibase]{X}^\bot \parallel \reali[\ibase]{Y}^\bot \\
      \Rightarrow \quad &
      \maltese_2\perp (\{ \maltese_\parr \}^\bot)^\bot \parallel (\{ \maltese_\parr \}^\bot)^\bot \\
      \Leftrightarrow \quad &
      \maltese_2\perp \{ \maltese_\parr \}^{\bibot} \parallel \{ \maltese_\parr \}^{\bibot}\\
      \Leftrightarrow \quad &
      \maltese_2\perp \{ \maltese_\parr \}^{\bibot} \parallel \{ \maltese_\parr \}^{\bibot}\\
      \Leftrightarrow \quad &
      \maltese_2\perp \{ \maltese_\parr \}^{\bibot} \parallel^- \{ \maltese_\parr \}^{\bibot}\\
      \Rightarrow \quad &
      \maltese_2\perp \{ \maltese_\parr \} \parallel^- \{ \maltese_\parr \}\\
      \Rightarrow \quad &
      \maltese_2\perp  \maltese_\parr  \parallel\maltese_\parr\\
    \end{align*}

    Indeed one can check that 
    the last assertion is false
    by computing the reduction 
    of $\maltese_2 :: \maltese_\parr \parallel \maltese_\parr$,
    this is illustrated in \autoref{fig:cutfailure}.
    This shows that $\maltese_2$
    can realise in all basis only sequents of the form $X,X^\bot$.

    \emph{Inductive cases.}
    Assume that $S$
    represents a proof terminating with a tensor
    then it is of the form 
    $S:: =S_1 + S_2 + \tenslink{S_1(1),S_2(1)}{p}$ where $p$ is a fresh position
    (without loss of generality assume the tensor is made on the first formula 
    of the subproofs).
    Let us show we can call again the hypothesis 
    on the subproof $S_1$ and $S_2$.
    Since $S$ represents 
    a proof terminating with a tensor 
    it proves a sequent $\Gamma,A\otimes B$.
    by \autoref{prop:decomp}
    it follows that $S_1 \parallel S_2$
    belongs to $\Gamma,A\parallel B$
    and thus using \autoref{lem:splitting}
    $S_1\parallel S_2$ belongs to 
    $\reali[\ibase]{\Gamma_1,A}\parallel \reali[\ibase]{\Gamma_2,B}$.
    Since $S$ is the tensor of the two proofs $\pi_1$ and $\pi_2$ 
    represented 
    by $S_1$ and $S_2$ respectively 
    we may assume that $\pi_1$ and $\pi_2$
    prove sequents which have no propositional variable in common.
    We apply now \autoref{prop:parallelExactly}
    and conclude that 
    $S_1$ belongs to $\reali[\ibase]{\Gamma_1,A}$
    while 
    $S_2$ belongs to $\reali[\ibase]{\Gamma_2,B}$
    for any basis $\ibase$.
    Calling the induction hypothesis 
    this yields $S_1\vdash_{\mll}\Gamma_1,A$
    and $S_2\vdash_{\mll} \Gamma_2,B$
    therefore, 
    $S\vdash_{\mll}\Gamma ,A\otimes B$.

    Assume that $S$
    represents a proof terminating with a par--rule
    then it is of the form 
    $S:: =S_0 + \parrlink{S_0(1),S_0(1)}{p}$ where $p$ is a fresh position.
    Let us show we can call again the hypothesis 
    on the subproof $S_0$.
    Since $S$ represents 
    a proof terminating with a $\parr$
    and belongs to $\reali[\ibase]{\Gamma,A\parr B}$.
    by \autoref{prop:decomp}
    it follows that $S_0$
    belongs to $\reali[\ibase]{\Gamma,A, B}$
    for any basis,
    applying the induction hypothesis 
    it follows that 
    $S_0$ is a proof of $\mll$
    of $\Gamma,A,B$
    and thus $S\vdash_{\mll} \Gamma,A\parr B$.

  \end{proof}


\end{document}